\documentclass{lmcs}
\usepackage[utf8]{inputenc}
\pdfoutput=1

\usepackage{lastpage}
\lmcsdoi{19}{2}{3}
\lmcsheading{}{\pageref{LastPage}}{}{}%
{Aug.~18,~2021}{Apr.~20,~2023}{}

\usepackage{graphicx}
\usepackage{xcolor}
\usepackage{url}
\usepackage{hyperref}
\hypersetup{colorlinks, urlcolor = {cyan}}
\usepackage{bbm}
\usepackage{diagrams}

\newcommand{\myomit}[1]{}

\newcommand{\myor}{\mathtt{or}}

\newcommand{\myorA}{\mathtt{or_A}}
\newcommand{\myorE}{\mathtt{or_E}}

\newcommand{\BEQ}{::=}
\newcommand{\BOR}{\;\mid\;}
\newcommand{\FV}{\mathrm{FV}}
\newcommand{\type}{:}
\newcommand{\den}[2]{#1[\![ #2 ]\! ]   }
\newcommand{\sden}[1]{\den{}{#1}}

\newcommand{\B}{\mathbb{B}}
\newcommand{\unit}{\mathtt{Unit}}
\newcommand{\fst}{\mathtt{fst}}
\newcommand{\snd}{\mathtt{snd}}
\newcommand{\BT}{\mathrm{BTypes}}
\newcommand{\tuple}[1]{\langle#1\rangle}

\newcommand{\myif}{\mathtt{if}\;}
\newcommand{\mythen}{\;\mathtt{then}\;}
\newcommand{\myelse}{\;\mathtt{else}\;}
\newcommand{\bool}{\mathtt{Bool}}
\newcommand{\true}{\mathtt{t}\!\mathtt{t}}
\newcommand{\false}{\mathtt{f}\!\mathtt{f}}

\newcommand{\Val}{\mathrm{Val}}
\newcommand{\val}{\mathrm{val}}
\newcommand{\Op}{\mathrm{Op}}
\newcommand{\SOp}{\mathrm{Op}_{\mathrm{s}}}
\newcommand{\eqdef}{=_{\scriptsize \mbox{def}}}

\newcommand{\POPLomit}[1]{ #1}

\newcommand{\argmax}{\mathrm{argmax}}
\newcommand{\E}{\mathcal{E}}
\newcommand{\R}{\mathrm{R}}
\newcommand{\Str}{\mathrm{Str}}
\newcommand{\Out}{\mathrm{Out}}
\newcommand{\Rew}{\mathrm{Rew}}

\newcommand{\reward}{\mathtt{Rew}}
\newcommand{\give}{\mathtt{reward}}

\newcommand{\st}{\mathrm{st}}
\renewcommand{\S}{\mathrm{S}}
\newcommand{{\T}}{\mathrm{T}}
\newcommand{\M}{\mathrm{M}}
\newcommand{\K}{\mathrm{K}}
\newcommand{\op}{op}

\newcommand{\Ob}{\mathrm{Ob}}
\newcommand{\Ex}{\mathbf{E}}
\newcommand{\un}[1]{\underline{#1}}
\newcommand{\sem}{\sden{\mbox{-}}}

\renewcommand{\max}[1]{ {\mathrm{max}}_{#1} }
\renewcommand{\Pr}{\mathcal{D}_\mathrm{f}}
\newcommand{\Ax}{\mathrm{Ax}}

\newcommand{\Msem}{\mathcal{M}}
\newcommand{\Ssem}{\mathcal{S}}

\newcommand{\CF}{\mathrm{CF}}
\newcommand{\mylet}{\mathrm{let}\,}
\newcommand{\mybe}{\,\mathrm{be}\,}
\newcommand{\myin}{\,\mathrm{in}\,}

\newcommand{\myand}{\,\mathrm{and}\,}
\newcommand{\norm}[1]{||#1||}

\newcommand{\VDis}{\mathrm{VDis}}
\newcommand{\Con}{\mathrm{Con}}
  \newcommand{\W}{\mathrm{W}}
  
\newcommand{\vsupp}{\mathrm{vsupp}}

\newcommand{\El}{\mathtt{E}}
\newcommand{\Ab}{\mathtt{A}}
\newcommand{\id}{\mathrm{id}}
\newcommand{\supp}{\mathrm{supp}}

\newcommand{\myk}{\mathrm{k}}
\newcommand{\myK}{\mathrm{K}}
\newcommand{\myF}{\mathrm{F}}

\newcommand{\order}{\mathrm{o}}

\newcommand{\cutproof}[1]{#1}

\newcommand{\ropen}[1]{[#1)} 
\newcommand{\lopen}[1]{(#1]} 

\title{Smart Choices and the Selection Monad}
\address{Google Research}
\email{abadi@google.com, plotkin@google.com}

\author[M.~Abadi]{Mart\'{\i}n Abadi}[]
\author[G.~Plotkin]{Gordon Plotkin}[]

\begin{document}

\begin{abstract}
Describing systems in terms of choices and their resulting costs and
rewards offers the promise of freeing algorithm designers and
programmers from specifying how those choices should be made;
in implementations, the choices can be realized by optimization techniques and,
increasingly, by machine-learning methods.  We study this approach from a
programming-language perspective. We define two small languages  that
support decision-making abstractions: one with choices and rewards, and the other
 additionally with probabilities. We give both operational and denotational
semantics.

In the case of the second language we consider three denotational semantics,
with varying degrees of correlation between possible program values and expected rewards.
The operational semantics combine the usual semantics of
standard constructs with optimization over spaces of possible execution strategies.
The denotational semantics, which are compositional,
rely
on the selection monad, to handle choice, augmented with an auxiliary
monad to handle other effects, such as rewards or probability.

We establish
adequacy theorems that the two semantics coincide in all cases. We also prove full
abstraction at base   types, with varying notions of observation in the probabilistic case corresponding to the various degrees of correlation. We present axioms for choice combined with rewards and probability, establishing completeness at base   types for the case of rewards without probability.
\end{abstract}

\maketitle

 \tableofcontents

\section{Introduction}

Models and techniques for decision-making, such as Markov Decision
Processes (MDPs) and Reinforcement Learning (RL), enable the description of
systems in terms of choices and of the resulting costs and rewards.
For example,
an agent that plays a board game may
be defined by its choices in moving pieces and by how many points
these yield in the game.  An implementation of such a system may aim
to make the choices following a strategy that results in attractive
costs and rewards, perhaps the best ones. For this purpose it may rely
on classic optimization techniques or, increasingly, on forms of
machine-learning (ML). Deep RL has been particularly
prominent in the last decade, but contextual bandits and ordinary supervised learning can
also be useful.

In a programming context,
several languages and libraries support choices, rewards, costs, and related notions in a general way
(not specific to any
application, such as a particular board game).
McCarthy's \texttt{amb} operator~\cite{jmc:mtc} may be seen as
an early example of a construct for making choices.
More recent work includes many libraries for RL
(e.g.,~\cite{rlax2020github}),
languages for planning
such as DTGolog~\cite{Boutilier00a} and some descendants (e.g.,~\cite{rddl}) of the
Planning Domain Definition Language~\cite{pddl},
a ``credit-assignment'' compiler for learning to search built on the Vowpal-Rabbit learning library~\cite{Chang16},
and Dyna~\cite{Dyna-mapl}, a programming language for machine-learning applications based on MDPs.
It also includes SmartChoices~\cite{smartchoices}, an ``approach to making
machine-learning (ML) a first class citizen in programming languages'',
one of the main inspirations for our work.
SmartChoices and several other recent industry projects in this space (such as Spiral~\cite{spiral})
extend mainstream programming languages and systems with the ability to make data-driven decisions by
coding in terms of choices (or predictions) and feedback (in
other words, perceived costs or rewards), and thus aim to have widespread impact on programming practice.

The use of decision-making abstractions has the potential to
free algorithm designers and programmers from taking care of many
details. For example, in an ordinary programming system, a programmer that implements quicksort should
consider how to pick pivot elements and when to fall back to a
simpler sorting algorithm for short inputs. Heuristic solutions to
such questions abound, but they are not always optimal, and they require coding and sometimes
maintenance when the characteristics of the input data or the
implementation platform change.  In contrast,
SmartChoices enables the programmer to code in terms of choices and costs---or, equivalently, rewards, which
we define as the opposite of costs---, and
to let the implementation of decision-making take care of the
details~\cite{smartchoices}.  As another example, consider the
program in Figure~\ref{smart-ex} that does binary search in a sorted array.
\begin{figure}[h]

\begin{verbatim}
 let binsearch(x : Int, a : Array[Int], l : Int, r : Int) =
  if l > r then None  //the special value None represents failure
           else choose m:[l,r]  in  //choose an integer in [l,r]
                    if a[m] = x then m
                    else cost(1);  //pay to recurse
                      if a[m] < x then binsearch(x, a, m+1, r)
                      else binsearch(x, a, l, m-1)
\end{verbatim}
\caption{Smart binary search}%
\label{smart-ex}
\end{figure}
This pseudocode is a simplified version of the one in~\cite[Section
  4.2]{smartchoices}, which also includes a way of recording observations
of the context of choices (in this example, \texttt{x}, \texttt{a[l]}, and \texttt{a[r]})
that facilitate machine-learning.  Here, a
choice determines the index \texttt{m} where the array is split.  Behind the scenes, a
clever implementation can take into account the distribution of the
data in order to decide exactly how to select \texttt{m}.
For example, if \texttt{x} is half way between \texttt{a[l]} and \texttt{a[r]} but the distribution of the values
in the array favors smaller values, then the selected \texttt{m} may be closer to \texttt{r} than to \texttt{l}.
In order to
inform the implementation, the programmer calls \texttt{cost}: each call
to \texttt{cost} adds to the total cost of an execution, for the notion
of cost that the programmer would wish to minimize. In this example,
the total cost is the number of recursive calls. In other examples,
the total cost could correspond, for instance, to memory requirements
or to some application-specific metric such as the number of points in a game.

In this paper, which is a full version of~\cite{AP21}, we study decision-making abstractions from a
programming-language perspective. We define two small languages that
support such abstractions,
one with choices and rewards, and the other one additionally with probabilities.
In the spirit of SmartChoices (and in
contrast with DTGolog and Dyna, for instance), the languages are  mostly
mainstream: only the decision-making abstractions are special.  We
give them both operational and denotational semantics.
In the case of the language with probabilities we provide three denotational semantics, modeling
 varying degrees of correlation between possible program values and expected rewards.

Their operational semantics combine the usual semantics of standard
constructs with optimization over possible strategies (thinking of programs as providing one-person games).
Despite the global character of optimization, our results include a tractable, more local formulation of their
operational semantics (Theorems~\ref{op-comp1} and~\ref{op-comp2}).
Their denotational semantics are based on
the selection
monad~\cite{escardo10,escardo11,escardo2011sequential,escardo12,escardo15,hedges15,escardo17,bolt18},
which we explain below.

We establish that operational and denotational semantics coincide, proving adequacy results for both languages
(Theorems~\ref{sel-ad1} and~\ref{sel-ad2}).
We also investigate questions of full abstraction (at base types)  and program equivalences.
Our full abstraction results (particularly Theorems~\ref{theorem:equivalences} and~\ref{fullabT}, and Corollary~\ref{concrete-case}) provide further evidence of the match between denotational and operational semantics.
We prove
 full abstraction results at base   types for each of our denotational semantics, in each case with respect to appropriate notions of observation.
Program equivalences can justify program transformations,
and we develop proof systems for them.
For example, one of our axioms concerns the commutation of choices and rewards.
In particular, in the case of the language for rewards we establish
(Theorem~\ref{theorem:equivalences}) the soundness and completeness of our proof system
with respect to concepts of
observational equivalence and semantic equivalence (at base   types).
In the case of  the language with probabilities, finding such completeness results is an open problem. However, we show that our proof systems are complete with respect to  proving effect-freeness. For the language without probabilities this holds in all circumstances (Corollary~\ref{purity1}); for the language with probability it holds under reasonable assumptions (Theorem~\ref{genpurity2}).

A brief, informal discussion of the semantics of \texttt{binsearch} may
provide some intuition on the two semantics and on the role of the
selection monad.
\begin{itemize}
\item If we are given the sequence of values picked by the
\texttt{choice} construct in an execution of \texttt{binsearch}, a standard
operational semantics straightforwardly allows us to construct the rest of
the execution. We call this semantics the \emph{ordinary operational semantics}.
For each such sequence of values, the ordinary operational semantics implies a resulting total cost, and thus
a resulting total reward.
We define the \emph{selection operational semantics}
by requiring that the sequence of values be the one that
maximizes this total reward.

Although they are rather elementary, these operational semantics are not always
a convenient basis for reasoning, because (as usual for operational semantics) they are not compositional,
and in addition the selection operational semantics
is defined in terms of sequences of choices and accumulated rewards in multiple executions.
On the other hand, the chosen values are simply plain integers.

\item In contrast, in the denotational semantics, we look at
  each choice of \texttt{binsearch} as being made locally, without implicit reference to the rest of the execution or other executions,
  by a higher-order function of type $(\texttt{Int} \rightarrow {\R}) \rightarrow
  \texttt{Int}$
  (where $\texttt{Int}$  is a finite set of machine integers),
whose expected argument
is a reward function $f$ that maps each possible value of the choice to the corresponding reward of type~$\R$ of the program. We may view $f$ as a reward continuation.
One possible such higher-order function
is
the
function ${\argmax}$ that picks a value for the argument $x$
for $f$ yielding the largest reward $f(x)$.
(There are different versions of $\argmax$, in particular with different ways of breaking ties, but informally one often identifies them all.)

The type $(\texttt{Int} \rightarrow \R) \rightarrow \texttt{Int}$
of this example
is a simple instance of the selection monad, $\S(X) = (X \rightarrow \R) \rightarrow X$, where $X$ is any type,
and ${\argmax}$ is an example of a selection function.
More generally,
we use
$\S(X) = (X \rightarrow \R) \rightarrow {\T}(X)$, where
$\T$ is another, auxiliary, monad, which can be used to model other computational effects, for example, as we do here,
rewards and probabilities. For our language with rewards, we employ the writer monad $\T(X) = \R \times X$. For our language with rewards and probabilities we employ three auxiliary monads modeling the various correlations between final values and rewards. Of these, the simplest is $\T(X) = \Pr(\R \times X)$, the combination  of the finite probability distribution monad with the writer monad.

 The monadic approach
leads to a denotational semantics that is entirely compositional,
and therefore
facilitates proofs of program equivalences of the kind mentioned above.
The denotational semantics may be viewed as
an implementation by translation to a language in which there are no
primitives for decision-making, and instead one may program with selection functions.
\end{itemize}

\noindent
Sections~\ref{selection} and~\ref{basic} concern supporting theory for our two decision-making languages.
 In Section~\ref{selection}, we review the selection monad, with and without an auxiliary monad, and investigate its algebraic operations. We  show how algebraic operations for the selection monad  with an auxiliary monad can be obtained from algebraic operations of the auxiliary monad (Equation~\ref{algalgops}); we give a general notion of selection  operations (Equation~\ref{selopdef}) and characterize them in terms of generic effects for the  selection monad $\S(X)$; and we investigate the equations obeyed by binary selection operations (Theorems~\ref{genax} and~\ref{distributes}).
 In Section~\ref{basic}, we present a general language with algebraic operations, give a general adequacy theorem (Theorem~\ref{basic-ad}),  and briefly discuss a calculus for program equivalences. This section is an adaptation of prior work (see~\cite{PP01}). While useful for our project, it is not specific to it.

In Section~\ref{first}, we define and study our first language with decision-making abstractions; it is
a simply typed, higher-order $\lambda$-calculus, extended
with a binary choice operation $-\,\myor\,-$ and a construct for adding rewards.  Full abstraction for this language is defined in terms of observing both final values and the corresponding rewards obtained.
Theorem~\ref{vfullab} shows that this notion does not change if we observe only the final value; in contrast Corollary~\ref{notrab} shows that it does change if we observe only the final reward: in that case we cannot distinguish programs with different final values but the same optimal final reward.

In Section~\ref{second}, we proceed to our second language, which adds
probabilistic choice  to the first. Regarding full abstraction,
Theorem~\ref{obweak2} (an analogue of Theorem~\ref{vfullab}) shows that this notion does not change from that
associated to our third semantics for probability and rewards if we observe only the distribution of final values. Probabilistic choices
are not subject to optimization, but,  combined with binary choice,
they enable us to imitate the choice capabilities of MDPs.
Unlike MDPs, the language does not support infinite
computations. We conjecture they can be treated via a metric approach to semantics; at any rate, there is no difficulty in adding a primitive recursion operator to the language without changing the selection monads, permitting MDP runs of arbitrary prescribed lengths.

In sum, we regard
the main contributions of this paper as being (1): the connection between programming languages with decision-making abstractions and the selection monad,
and (2): the definition and study of operational and denotational semantics for those languages, and the establishment of adequacy and full abstraction theorems for them. The adequacy theorems show that global operationally-defined optimizations can be characterized compositionally using a semantics based on the selection monad.

As described above, the selection operational semantics and the
denotational semantics with the ${\argmax}$ selection function
both rely on maximizing rewards.
In many cases, optimal solutions are expensive.
Even in the case of \texttt{binsearch},
an optimal solution that, without ever recursing, immediately picks \texttt{m} such that \texttt{a[m]} equals \texttt{x}
 seems unrealistic.
For efficiency, the optimization may be approximate and
data-driven. In particular, as an alternative to the use of maximization
in the selection operational semantics, we may sometimes be able to
make the choices with contextual-bandit techniques, as
in~\cite[Section 4.2]{smartchoices}. In the denotational semantics,
with $\R$ the type of real numbers, we may use other
selection functions than ${\argmax}$. (Using  ${\argmax}$ is convenient, but our approach
does not require it.) For example, instead of
computing ${\argmax}(f)$, we may approximate $f$ by a
differentiable function over the real numbers, represented by a neural
network with learned parameters, and then find a local maximum of this
approximation by gradient ascent. We have explored such
approximations only informally so far;  Section~\ref{conclusion}
briefly mentions aspects of this and other subjects for further work.


\section{The selection monad and algebraic operations}\label{selection}

In this section
we present material on the basic selection monad, on the selection monad  augmented with an auxiliary monad, and  on generic effects and algebraic operations for general monads. This material includes a discussion of generic effects and algebraic operations for the selection monad (whether basic or augmented) and of the equations these operations satisfy.  Such algebraic operations are either so-called selection operations arising from the basic selection monad or operations arising from the auxiliary monads and then lifted to the augmented selection monad. For a first reading, it suffices to read the definitions of the selection monads and of generic effects and algebraic operations for general monads. (We repeat the definitions of the specific generic effects and algebraic operations for our two languages when discussing their denotational semantics in Sections~\ref{den-sem1} and~\ref{den-sem2}.)

\subsection{The selection monad}\label{selmonad}

The selection monad
\[\S(X) = (X \rightarrow \R) \rightarrow X\]
introduced in~\cite{escardo10}, is a strong monad available in any cartesian closed category, for simplicity discussed here only in the category of sets. One can think of the $F \in \S(X)$ as \emph{selection functions} which, viewing $\R$ as a \emph{reward} type, choose an element  $x \in X$, given a \emph{reward function} $\gamma\type X \rightarrow \R$.  In a typical example, the choice $x$ optimizes, perhaps maximizing, the reward $\gamma(x)$.
Computationally, we may understand $F\in \S(X)$ as producing $x$ given a \emph{reward continuation} $\gamma$,  a function giving the reward of the remainder of the computation.

The selection monad has strong connections to logic, similar to those of the continuation monad  $\K(X) = (X \rightarrow \R) \rightarrow \R$.
For example, as explained in~\cite{escardo12}, whereas  logic translations using $\K$, taking $\R$ to be $\perp$, verify the double-negation law $\neg\neg P \supset P$, translations using $\S$ verify the instance  $((P \supset \R) \supset P) \supset P$ of Peirce's law. Again, with $\R$ the truth values, elements of $\K(X)$ correspond to quantifiers, and elements of $\S(X)$ correspond to selection operators, such as Hilbert's  $\varepsilon$-operator.

\newcommand{\ER}[3]{\mathbf{R}_{#1}(#2|#3)} 

The selection monad has unit $(\eta_\S)_X\type X \rightarrow \S(X)$, where $\eta_\S(x) = \lambda \gamma\in  X \rightarrow \R.\, x$. (Here, and below, we may drop subscripts when they are evident from the context.)
The  Kleisli extension  is a little involved, so we explain it in stages. First, for any $F \in \S(X)$ and reward continuation  $\gamma \type Y \rightarrow \R$  we write $\ER{}{F}{\gamma}$ for the reward given by the (possibly optimal)  $x\in X$ chosen by $F$, i.e.:
\[\ER{}{F}{\gamma} \eqdef \gamma(F \gamma)\]
(Here, and below, we may omit function application parentheses to improve readability.)
For the Kleisli extension,  given $f\type X \rightarrow \S(Y)$ we need a function $f^{\dagger_\S}\type \S(X) \rightarrow \S(Y)$.
Equivalently, using $f$, we need to pick an element of $Y$, given a computation $F\in \S(X)$ and a reward continuation $\gamma \type Y \rightarrow \R$. We do so as follows:
\begin{itemize}
\item For a given $x \in X$, the reward associated to the possibly optimal element of $Y$ picked by  $f(x)$ is $\ER{}{f(x)}{\gamma}$.
\item Thus we have a reward function from $X$, viz.\ $\mathit{rew} = \lambda x\in X.\, \ER{}{f(x)}{\gamma}$.
\item Using this reward function as the reward continuation of $F$, we can use $F$ to choose the (possibly optimal) element of $X$ for it, viz.~${\mathit opt} = F({\mathit{rew}})$.
\item Now that we know the best choice of $x$, we use it to get the desired element of $Y$, viz.~$f(\mathit{opt})(\gamma)$.
\end{itemize}
Intuitively, $F$ chooses the $x\in X$ which gives the optimal $y\in Y$,
and then $f$ uses that $x$.

Writing all this out, we find:
\[\begin{array}{lclcl}f^{\dagger_\S}F\gamma & = & f(\mathit opt)(\gamma)\\
                                                   & = & f(F({\mathit rew}))(\gamma)\\
                                                   & = & f F(\lambda x \in  X.\, \ER{}{f(x)}{\gamma})\gamma\\
\end{array}\]

The selection monad has strength $(\st_\S)_{X,Y}\type X \times \S(Y) \rightarrow \S(X \times Y)$ where:
\[(\st_\S)_{X,Y}(x,F) = \lambda \gamma\in  X \times Y \rightarrow \R.\, \tuple{x,F(\lambda y\in  Y.\, \gamma(x,y))}\]
There is a generalization of this basic selection monad obtained by augmenting it with a  strong \emph{auxiliary} monad ${\T}$. This generalization  proves useful when combining additional effects with selection. Suppose that $\R$ is a ${\T}$-algebra with algebra map $\alpha_{\T}\type {\T}(\R) \rightarrow \R$. Then, as essentially proved in~\cite{escardo17} for any cartesian closed category, we can define a strong monad $\S_{\T}$ (which may just be written $\S$, when $\T$ is understood) by setting:
\[\S_\T(X) = (X \rightarrow \R) \rightarrow {\T}(X)\]
It has unit $(\eta_{\S_\T})_X\type X \rightarrow \S_\T(X)$ where $(\eta_{\S_\T})_X(x) = \lambda \gamma\in  X \rightarrow \R.\, (\eta_{\T})(x)$.
The  Kleisli extension $f^{\dagger_{\S_\T}}\type \S_\T(X) \rightarrow \S_\T(Y)$ of a function $f\type X \rightarrow \S_\T(Y)$  is given, analogously to the above. First, for $F \in \S_\T(X)$ and $\gamma \type X \rightarrow \R$, generalizing that for $\S$, we define the reward associated to the $\T$-computation selected by $F$ using $\gamma$ by:
\begin{equation}\ER{}{F}{\gamma} =  (\alpha_\T\circ \T(\gamma))(F(\gamma))\label{SelectionMonadReward} \end{equation}
Then the Kleisli extension function $f^{\dagger_{\S_\T}}\type \S_\T(X) \rightarrow \S_\T(Y)$ of a given $f\type X \rightarrow \S_\T(Y)$ is:
\begin{equation}\label{dagger}
f^{\dagger_{\S_\T}}F\gamma\;\; \eqdef \;\; (\lambda x \in  X.\, f(x)(\gamma))^{\dagger_{\T}}(F(\lambda x \in  X.\,  \ER{}{fx}{\gamma}))
\qquad (F \in \S_\T(X), \gamma \type Y \rightarrow \R)\end{equation}
As an example, suppose that $\R$ is a commutative monoid, and $\T$ is the writer monad $\W(X) = \R \times X$. Using the monoid operation, we can  set $\alpha_\W(r,s) = r + s$. We then have:
\[\ER{}{F}{\gamma} = \pi_1(F(\gamma))  + \gamma(\pi_2(F(\gamma)))\]

For $(\mu_\T)_X = \id_{\S_\T(X)}^{\dagger_{\S_\T}}$ we find:
\begin{equation}\label{mu}
(\mu_\T)_X F\gamma\;\; = \;\; (\lambda d \in \S_\T(X).\, d\gamma)^{\dagger_{\T}}(F(\lambda d \in  \S_\T(X).\, \ER{}{d}{\gamma}
))\end{equation}
The selection monad has strength $(\st_{\S_\T})_{X,Y}\type X \times \S_\T(Y) \rightarrow \S_\T(X \times Y)$ where:
\[(\st_{\S_\T})_{X,Y}(x,F) = \lambda \gamma\in X \times Y \rightarrow \R.\, (\st_{\T})_{X,Y}(x,F(\lambda y\in  Y.\, \gamma(x,y)))\]
 We remark that if $\T$ is the free algebra monad for an equational theory $\mathrm{Th}$,  the categories of $\T$-algebras and  of models of $\mathrm{Th}$ (i.e., algebras satisfying the equations) are equivalent. In particular the $\T$ algebra $\alpha_A\type \T(A) \rightarrow A$ corresponding to a model $A$ is the homomorphism $\id_A^{\dagger_\T}$, and $h\type A \rightarrow B$ is a homomorphism between models of the theory iff it is a $\T$-algebra morphism (from $\alpha_A$ to $\alpha_B$). Note that the Kleisli extension of a map $f\type X \rightarrow A$ to a model $A$ is the same as the Kleisli extension of $f$ regarded as a map to a $\T$-algebra.

We can define reward functions  for general monads $\M$ equipped with an $\M$-algebra $\alpha_M\type \M(\R) \rightarrow \R$.
For $u \in \M(X)$ and $\gamma \type X \rightarrow \R$, set:
\begin{equation} \ER{\M_X}{u}{\gamma} = (\alpha_\M\circ \M(\gamma))(u) \quad (= \gamma^{\dagger_\M}(u))\label{MonadReward}\end{equation}
 Note that, for $x\in X$, $\ER{\M}{\eta_\M(x)}{\gamma} = \gamma(x)$ and that $\ER{\M}{-}{\gamma}$ is an $\M$-algebra morphism. Further, in case  $\M$ is the free algebra monad for an equational theory, and $\alpha_M$ is the
 $\M$-algebra corresponding to a model on $\R$,   $\ER{\M}{-}{\gamma}$ is a homomorphism.

 We remark (see~\cite{K80,K93,HL07}) that, using the reward function, one can  define a morphism $\theta_\M$ from $\M$ to the continuation monad, by setting
\[\theta_\M(u) = \lambda \gamma \in (X \rightarrow \R).\, \ER{\M}{u}{\gamma} \quad (u \in \M(X))\]

In the case of the selection monad,  define $\alpha_{\S_\T}\type \S_\T(\R) \rightarrow \R$ by:
\[\alpha_{\S_\T}(F) = \alpha_\T(F(\id_\R))\]
\begin{fact} $\alpha_{\S_\T}\type \S_\T(\R) \rightarrow \R$ is an $\S_\T$-algebra.
\end{fact}
\begin{proof}
We write $\S$ for $\S_\T$.
We have to show that $\alpha_{\S}$ satisfies the unit and multiplication requirements to be an $\S_\T$-algebra, i.e., that
$\alpha_{\S}\circ (\eta_{\S})_\R = \id_\R$ and $\alpha_\S \circ (\mu_S)_\R= \alpha_\S \circ \S(\alpha_\S)$.

For the first requirement we have:
\[\alpha_{\S_\T}((\eta_\S)_\R(r)) = \alpha_\T((\eta_\S)_\R(r)(\id_\R)) = \alpha_\T((\eta_\T)_\R(r)) = r\]

For the second requirement, for $F \in \S(\S(\R))$ we calculate, first, that:

\[\begin{array}{lcl}
\alpha_{\S}((\mu_\T)_\R F)
   & =  & \alpha_{\S}(\lambda \gamma\type\R \rightarrow \R.\, (\lambda d \in \S(\R).\, d\gamma)^{\dagger_{\T}}(F(\lambda d
       \in  \S(\R).\, \ER{}{d}{\gamma})))\\
   & =  & \alpha_{\T}((\lambda d \in \S(\R).\, d\,\id_\R)^{\dagger_{\T}}(F(\lambda d \in  \S(\R).\, \ER{}{d}{\id_\R})))\\
   & =  & \alpha_{\T}((\lambda d \in \S(\R).\, d\,\id_\R)^{\dagger_{\T}}(F(\lambda d \in  \S(\R).\, \alpha_\T(d\,\id_\R))))\\
   & =  & (\lambda d \in \S(\R).\, \alpha_\T(d\,\id_\R))^{\dagger_{\T}}(F(\alpha_\S))\\
   & =  & \alpha_\S^{\dagger_{\T}}(F(\alpha_\S))\\
   \end{array}\]
 (where the fourth equality uses the fact that for any $\M$-algebra $\alpha\type \M(X) \rightarrow X$ and any $f\type Y \rightarrow \M(X)$ we have $\alpha\circ f^{\dagger_\M} = (\alpha\circ f)^{\dagger_\M}$) and, second, that:
   \[\begin{array}{lcl}  \alpha_\S(\S(\alpha_\S)(F))
        & =  & \alpha_\S( \lambda \gamma\type \R \rightarrow \R.\, \T(\alpha_\S)(F(\gamma \circ \alpha_S)))\\
                & =  & \alpha_\T(\T(\alpha_S)(F(\alpha_S)))\\
                & =  & \alpha_\S^{\dagger_\T}(F(\alpha_\S))\\   \end{array}\]
(where the last equality uses the fact that for any $\M$-algebra $\alpha\type \M(X) \rightarrow X$, and any $f\type Y \rightarrow X$ we have $\alpha\circ \M(f) = f^{\dagger_\M}$). This concludes the proof.
\end{proof}

Using the general formula for the reward function for monads equipped with an algebra on $\R$,
we then calculate
 for $F \in \S_\T(X) = (X \rightarrow \R) \rightarrow \T(X)$ and $\gamma\type X \rightarrow \R$  that:
\[\begin{array}{lcl}
\ER{\S_\T}{F}{\gamma}& = & \alpha_{\S_\T} (\S_\T(\gamma)(F))\\
                          & = & \alpha_{\S_\T}(\lambda \gamma'\type \R \rightarrow \R.\,\T(\gamma)(F(\gamma'\circ\gamma )))\\
                          & = & \alpha_{\T}(\T(\gamma)(F\gamma))\\
\end{array}\]
As desired, this is  the reward function of Definition~\ref{SelectionMonadReward}.
Note that $\ER{\S_\T}{F}{\gamma} = \ER{\T}{F\gamma}{\gamma}$.

\subsection{Generic effects and algebraic operations}\label{genalgops}

  In order to be able to give semantics to effectual operations such as probabilistic
 choice, we use the apparatus of generic effects and algebraic
 operations in the category of sets discussed  in~\cite{PP03} (in a much more general
 setting).
Suppose that $\M$ is  a (necessarily strong) monad on the category of sets.  A \emph{generic effect} $g$ with arity $(I,O)$ (written $g\type (I,O)$) for $\M$
 is just a Kleisli map:
\[g\type O \rightarrow \M(I)\]
An  \emph{$M$-algebraic operation} $\op$  with arity $(I,O)$ (written $\op\type (I,O)$)
 is a family of functions
\[\op_X\type O \times \M(X)^I \rightarrow \M(X)\]
natural with respect to Kleisli maps in the sense that the following diagram commutes for all $e\type X \rightarrow \M(Y)$:

\begin{diagram}
O \times \M(X)^I &\rTo^{\op_X} &\M(X)\\
\dTo^{O \times (e^{\dagger_{\M}})^I} & &\dTo_{e^{\dagger_{\M}}}\\
O \times \M(Y)^I &\rTo^{\op_Y} & \M(Y)
\end{diagram}

There is a 1--1 correspondence between $(I,O)$-ary generic effects and $(I,O)$-ary algebraic operations. In one direction, given $g$, one sets
\begin{equation}
\op_X(o,a) = a^{\dagger_{\M}}(g(o))\label{go}
\end{equation}

In the other direction, given such a family $\op$, one sets
\begin{equation}
g(o) = \op_I(o,(\eta_{\M})_I)\label{og}
\end{equation}
Naturality implies a weaker but useful property, that the above diagram commutes for maps $\M(f)$, for any $f\type X \rightarrow Y$.
In other words, if we regard $\M(X)$ and $\M(Y)$ as algebras equipped with (any) corresponding algebraic operation components, such maps are homomorphisms $\M(f)\type \M(X) \rightarrow \M(X)$. Naturality also implies that, as monad mutltiplications $(\mu_\T)_{X}$ are Kleisli extensions, they too act homomorphically on algebraic operations.

We  generally obtain the algebraic operations we need via their generic effects.
When $O$ is a product $O_1 \times \cdots \times O_m$,
we obtain  semantically useful functions
\[\op^\dagger_X\type (\M(O_1) \times \cdots \times  \M(O_m)) \times \M(X)^I \rightarrow \M(X)\]
from an algebraic operation
\[\op_X\type (O_1 \times \cdots \times  O_m) \times \M(X)^I \rightarrow \M(X)\]
This can be done  by applying iterated Kleisli extension to the curried version
\[\op'_X\type O_1 \rightarrow  \ldots \ O_m  \rightarrow \M(X)^I \rightarrow \M(X)\]
 of $\op$, or, equivalently, using Kleisli extension and  the monoidal structure \[(m_{\T})_{X,Y}\type \M(X) \times \M(Y) \rightarrow \M(X \times Y)\] induced by the monadic strength (see~\cite{kock1972strong}).

When $O = \mathbbm{1}$, we generally ignore it and equivalently  write $g\type I$ and $g\in \M(I)$ for generics and $\op\type I$ and $\op_X\type \M(X)^I \rightarrow \M(X)$ for algebraic operations. We adopt similar conventions below for related occurrences of $\mathbbm{1}$. Note that $(I,O)$-ary algebraic operations  $\op$  are in an evident  correspondence with indexed families $\op_o \; (o \in O)$ of $I$-ary algebraic operations;
in particular, when $I = [n]$  (as usual, $[n] = \{i\, |\, i <n \}$),
the $\op_o$ can be considered to be families $(\op_o)_X\type \M(X)^n \rightarrow \M(X)$ of $n$-ary functions. (Here, and below, it is convenient to confuse $[n]$ with $n$.)

The $[n]$-ary algebraic operations include the
projections $\pi_{n,i}\type \M(X)^n \rightarrow \M(X)$, for $i = 0,\ldots,n-1$, and are closed under composition, meaning that if $\op$ is an $[n]$-ary algebraic operation, and $\op_i$ are $[m]$-ary algebraic operations, then so is
$\op\circ\tuple{\op_0,\ldots,\op_{n-1}}$ where:
\[(\op\circ\tuple{\op_0,\ldots,\op_{n-1}})_X(u_0,\ldots, u_{m-1})= \op_X((\op_0)_X(u_0),\ldots, (\op_{m-1})_X(u_{m-1}))\]
There are natural corresponding generic effects and operations on them. This is part of a much larger picture. The generic effects of a monad $\M$ form its Kleisli category,  with objects all sets. This category has all small sums, and so its opposite, termed the large Lawvere theory of $\M$ (see~\cite{ dub06,HLPP07}), has all small products.  The algebraic operations also form a category, again with objects all sets, and with morphisms from $I$ to $O$  the $(I,O)$-ary algebraic operations
(identity and composition are defined componentwise). The correspondence between generic effects and algebraic operations forms an isomorphism between these two categories.

We say that algebraic operations $\op_1\!\type\! [n_1], \ldots, \op_k\type [n_k]$ \emph{satisfy}
equations over function symbols $f_1\type n_1, \ldots, f_k\type n_k$ iff for any $X$, $(\op_1)_X,\ldots,(\op_k)_X$ do, in the usual sense, i.e., if the equations hold with the $f_i$ interpreted as $(\op_i)_X$ for $i = 1,\ldots, k$.
In the case where $\M$ is the free-algebra monad for an equational theory  $\mathrm{Th}$ with function symbols $\op\type n$ of given arity, the  $\op_X\type \M(X)^n \rightarrow \M(X)$ form  $[n]$-ary algebraic operations (indeed, in this case all algebraic operations occur as compositions of these ones and the projection algebraic operations). These algebraic operations satisfy all the equations of $\mathrm{Th}$.

Given an $\M$-algebra, $\alpha\type \M(X) \rightarrow X$, and an $\M$-algebraic operation $\op\type (I,O)$ we can induce a corresponding map
$\op_\alpha\type O \times X^I \rightarrow X$, by setting
\[\op_\alpha(o,u) =  \alpha(\op_X(o,\eta_X\circ u))\]
and $\alpha$ is then a homomorphism between $\op_{\M(X)}$ and the induced map. Given a collection of operations $\op_1\!\type\! [n_1], \ldots, \op_k\type [n_k]$, the corresponding induced maps satisfy the same equations the operations do. So, in particular, if $\M$ is the free-algebra monad for an equational theory  $\mathrm{Th}$ with function symbols $\op_i\type n_i$, $X$ becomes a model of the theory via the $(\op_i)_{\M(X)}$. Conversely, if $X$ is a model of the theory then we can define a corresponding $\M$-algebra by setting $\alpha = \id_X^{\dagger_\M}$. These two correspondences  yield an isomorphism between the categories of $\M$-algebras and models of the theory
(the isomorphism is the identity on morphisms).

Given a monad morphism $\theta \type \M \rightarrow \M'$, any generic effect $g\type  O \rightarrow \M(I)$ yields a generic effect $g' = \theta_I\circ g$ for $\M'$. Then, see~\cite{HPP06}, $\theta$ is a homomorphism of the corresponding algebraic operations,
$\op_X\type O \times \M(X)^I \rightarrow \M(X)$ and $\op'_X\type O \times \M'(X)^I \rightarrow \M'(X)$ in the sense that, for all sets $X$,  the following diagram commutes:
\begin{diagram}
O \times \M(X)^I &\rTo^{\op_X} &\M(X)\\
\dTo^{O \times (\theta_X)^I} & &\dTo_{\theta_X}\\
O \times \M'(X)^I &\rTo^{\op'_X} & \M'(X)
\end{diagram}

\newcommand{\ogamma}{\overline{\gamma}}
\newcommand{\wop}{\widetilde{\op}}
\newcommand{\wg}{\widetilde{g}}
\newcommand{\aux}{\mathrm{aux}}
We next consider algebraic operations for  the selection monad  $S_{\T}$. Modulo currying,
$(I,O \times \R^I )$-ary generic effects $g \type (O \times \R^I) \rightarrow  \T(I)$ for $\T$ are in bijective correspondence with $(I,O)$-ary generic effects
$\wg\type O \rightarrow  \T(I)^{ \R^I}$ for $\S_\T$.
There is therefore a corresponding bijective correspondence between $(I,O \times \R^I )$-ary $\T$-algebraic operations $\op$  and $(I,O)$-ary $\S_\T$-algebraic operations $\wop$.
This correspondence has a pleasing component-wise expression going from $\T$ to $\S_\T$. An  intermediate function family notion is useful. We define the \emph{auxiliary function family} $\aux_{X}\type  O \times \R^X \times \T(X)^I \rightarrow \T(X)$ associated to a $(I,O \times \R^I )$-ary $\T$-algebraic operation $\op$  by:
\begin{equation} \aux_{X}(o,\gamma,u) = \op_X(\tuple{o, \lambda i \in I.\, \ER{\T}{ui}{\gamma}}, u)\label{auxdef}\end{equation}
Below we write  $\aux_{X, o,\gamma}$ for the function $\aux_{X}(o,\gamma, -)$.

\begin{prop}\label{gengen}
Let $\op$ be an $(I,O\times \R^I)$-ary $\T$-algebraic operation. In terms of its associated auxiliary function family $\aux$, the corresponding
$(I,O)$-ary $\S_\T$-algebraic operation $\wop$ is given by:
\[\wop_X(o,a)= \lambda \gamma \in \R^X.\,\aux_{X, o,\gamma}(\lambda i\in I.\, ai\gamma) \]
Conversely, we have:
\[\aux_{X, o,\gamma}(u) = \wop_X(o,\lambda i\in I.\lambda \gamma \in X \rightarrow \R.\, ui)\gamma\]
\end{prop}
\begin{proof}
The generic effect $g \type O \times \R^I \rightarrow \T(I)$ corresponding to $\op$ is given by Equation~\ref{og}:
\[g(o,\ogamma) = \op_I(\tuple{o,\ogamma},(\eta_\T)_I)\]
Currying, we obtain $\wg \type O \rightarrow \S(I)$ where:
\[\wg(o) = \lambda \ogamma \in \R^I.\, \op_I(\tuple{o,\ogamma},(\eta_\T)_I)\]
and then, using Equation~\ref{go}, we have:
\[\wop_X(o,a) = a^{\dagger_{\S_\T}}(\wg(o))\]
We next choose a reward continuation $\gamma \in \R^X$ and examine
$a^{\dagger_{\S_\T}}(\wg(o))\gamma$. To this end we first obtain a reward continuation in $\R^I$ from  $a$ and $\gamma$,  namely:
\[\ogamma \eqdef \lambda i\in I.\, \ER{\T}{ai}{\gamma}\]
and, setting
\[u \eqdef (\lambda i \in I.\, ai\gamma) \in I \rightarrow \T(X)\] we have:
\[\begin{array}{lcll}
a^{\dagger_{\S_\T}}(\wg(o))\gamma & = & u^{\dagger_\T}(\wg(o)(\ogamma))\\
                                                                   & = & u^{\dagger_\T}(\op_I(\tuple{o,\ogamma},(\eta_\T)_I))\\
                                                     & = & \op_X(\tuple{o,\ogamma},u^{\dagger_\T}(\eta_\T)_I)
                                                     & (\mbox{by the Kleisli naturality} \\ &&& \;\mbox{ of algebraic operations})\\
                                                     & = & \op_X(\tuple{o,\ogamma},u)\\
\end{array}\]

Putting these facts together, we have:
\[\begin{array}{lcl}(\wop)_X(o,a)\gamma  & = & a^{\dagger_{\S_\T}}(g_{\S_\T}(o))\\
                                           & = & \op_X(\tuple{o,\ogamma},u)\\
                                           & = & \op_X(\tuple{o,\lambda i\in I.\, \ER{\T}{ai}{\gamma}},\lambda i \in I.\, ai\gamma)\\
                                           & = &  \aux_{X, o,\gamma}(\lambda i\in I.\, ai\gamma)
\end{array}\]
as required.
That
\[\aux_{X, o,\gamma}(u) = \wop_X(o,\lambda i\in I.\lambda \gamma \in X \rightarrow \R.\, ui)\gamma\]
is an immediate consequence, for, setting $a = \lambda i \in I.\, \lambda \gamma\in X \rightarrow \R.\, u(i)$,  we find:
\begin{align*}
\aux_{X, o,\gamma}(u) & = \aux_{X, o,\gamma}(\lambda i\in I.\, ai\gamma)\\
                                      & = (\wop)_X(o,a)\gamma\\
                                      & = (\wop)_X(o,\lambda i \in I.\, \lambda \gamma\in X \rightarrow \R.\, u(i))\gamma \qedhere
\end{align*}
\end{proof}

Note that, as is natural,  the proposition expresses that $(\op_{\S_\T})_X$ uses the reward function in $\R^I$ which assigns to $i \in I$ the reward obtained by following the $i$th branch.

The correspondences between the two kinds of algebraic operations and auxiliary functions fit well with fixing parameters. Given an  $(I,O \times \R^I )$-ary $\T$-algebraic operation $\op$, we obtain an $(I,\R^I )$-ary $\S_\T$-algebraic operation $\op'$ by fixing an $o\in O$. The corresponding auxiliary function family $\aux'_X\type \R^X \times \T(X)^I \rightarrow \T(X)$ is, as one would expect,  $\lambda \gamma, a.\, \aux_X(a,\gamma,a)$; the corresponding $I$-ary algebraic operation
 for $\S_\T$ is $\wop_o$.

Using Proposition~\ref{gengen}, we can reduce questions of equational satisfaction by $\S_\T$-algebraic operations to corresponding questions about their auxiliary functions, so reducing questions about  $\S_\T$ to questions about $\T$. We first need a lemma.
\begin{lem}\label{opaux}
\hfill
\begin{enumerate}
\item The auxiliary function family $\aux_{X,\gamma}$ corresponding to an $[n]$-ary projection $\S_\T$-algebraic operation  $\pi_{n,i}$ is the  family $(\pi_{n,i})_{\T(X)}\type \T(X)^n \rightarrow \T(X)$ of projections.

\item Let $\op$ be an $[n]$-ary $\S_\T$-algebraic operation, and, for $i = 0,\ldots, n-1$, let $\op_i$ be $[m]$-ary $\S_\T$-algebraic operations  for $\S_\T$, and let their corresponding  auxiliary function families be $\aux_{X}$ and $(\aux_i)_{X}$, respectively. Then the auxiliary function family $\aux'_{X}$ corresponding to
the composition $\op'$ of $\op$ with the $\op_i$ is the corresponding composition of auxiliary functions:
\[\begin{array}{l}\aux'_{X,\gamma}(u_0,\ldots, u_{m-1}) \\
\hspace{45pt}= \; \aux_{X,\gamma}((\aux_0)_{X,\gamma}(u_0,\ldots, u_{m-1}),\ldots, (\aux_{n-1})_{X,\gamma}(u_0,\ldots, u_{m-1}))\end{array}\]
\end{enumerate}
\end{lem}
\begin{proof}
\hfill
\begin{enumerate}
\item This is immediate from the second part of Proposition~\ref{gengen}.
\item Making use of both parts of Proposition~\ref{gengen} we calculate:
\begin{align*}
&\aux'_{X,\gamma}(u_0,\ldots, u_{m-1}) \\
&                 \hspace{25pt} =  \op'_X(\lambda \gamma. u_0,\ldots,\lambda \gamma. u_{m-1})\gamma\\
&                 \hspace{25pt} =\op_X((\op_0)_X(\lambda \gamma. u_0,\ldots, \lambda \gamma. u_{m-1}),\ldots, (\op_{n-1})_X(\lambda \gamma. u_0,\ldots, \lambda \gamma. u_{m-1}))\gamma\\
&                  \hspace{25pt}=  \aux_{X,\gamma}((\op_0)_X(\lambda \gamma. u_0,\ldots, \lambda \gamma. u_{m-1})\gamma,\ldots, (\op_{n-1})_X(\lambda \gamma. u_0,\ldots, \lambda \gamma. u_{m-1})\gamma)\\
&                  \hspace{25pt} = \aux_{X,\gamma}((\aux_0)_{X,\gamma}(u_0,\ldots, u_{m-1}),\ldots, (\aux_{n-1})_{X,\gamma}(u_0,\ldots, u_{m-1})) \qedhere
\end{align*}
 \end{enumerate}
\end{proof}

\begin{prop}\label{geneq} Let $\op_i$ be $([n_i],O_i\times \R^{[n_i]})$-ary $\T$-algebraic operations,  for $i = 1,\ldots, k$,
and choose $o_i \in O_i$ ($i = 1,\ldots,k$).
Then an equation is satisfied by $(\widetilde{op_1})_{o_1},\ldots,(\widetilde{op_k})_{o_k}$, if, for all sets $X$ and $\gamma\type X \rightarrow \R$,
it is satisfied by $(\aux_1)_{X,o_1,\gamma},\ldots,(\aux_k)_{X,o_k,\gamma}$, where, for $i=1,\ldots, k$, $\aux_i$ is the auxiliary function family obtained from $\op_i$.
\end{prop}
\begin{proof} We can assume  without loss of generality that the $O_i$ are all $\mathbbm{1}$ and so can be ignored. The interpretation of an algebraic term $t$ with $m$ free variables built from  function symbols $f_1\type n_1, \ldots, f_k\type n_k$  can be considered as an $m$-ary function, and an equation $t= u$ over $m$ free variables holds in the interpretation if the two such interpretations are equal.

Fixing a term $t$ with $m$ free variables, for any set $X$, using the  $(\wop_i)_X$ to interpret the $f_i$, we obtain functions
$\den{\mathcal{O}}{t}_{X}\type \S_\T(X)^m \rightarrow \S_\T(X)$, say, and  for any set $X$ and $\gamma\type X \rightarrow \R$, using the $(\aux_i)_{X,\gamma}$ we obtain functions
$\den{\mathcal{A}}{t}_{X,\gamma}\type \T(X)^m \rightarrow \T(X)$, say. As the projections are algebraic operations and as algebraic operations are closed under composition,  a straightforward structural induction shows that the family $\den{\mathcal{O}}{t}$ is an $m$-ary algebraic operation for $\S_\T$. Using Lemma~\ref{opaux}, a further straightforward structural induction shows that $\den{\mathcal{A}}{t}_{X,-}$ is the corresponding auxiliary function family.

Now suppose an equation $t=u$ over $m$ variables is satisfied by $(\aux_1)_{X,\gamma},\ldots,(\aux_k)_{X,\gamma}$ for all sets $X$ and $\gamma\type X \rightarrow \R$, that is, suppose that $\den{\mathcal{A}}{t}_{X,\gamma} = \den{\mathcal{A}}{u}_{X,\gamma}$, for all such $X$ and $\gamma$.
Then,  using  Proposition~\ref{gengen}, we see that:
\[\begin{array}{lcl}\den{\mathcal{O}}{t}_{X}(F_0,\ldots, F_{m-1})\gamma
& = &
\den{\mathcal{A}}{t}_{X,\gamma}(F_0\gamma,\ldots, F_{m-1}\gamma) \\
& = & \den{\mathcal{A}}{u}_{X,\gamma}(F_0\gamma,\ldots, F_{m-1}\gamma) \\
& = & \den{\mathcal{O}}{u}_{X}(F_0,\ldots, F_{m-1})\gamma\\
\end{array}\]
holds for all sets $X$ and $\gamma\type X \rightarrow \R$, concluding the proof.
\end{proof}

We next see that, as one would expect, we can use algebraic operations for $\T$-effects to obtain corresponding ones for $\S_\T$-effects.
If $\op$ is an $(I,O)$-ary $\T$-algebraic operation, it can be considered to be an $(I \times \R^I,O)$-ary algebraic operation which ignores its reward function argument.  The auxiliary  functions $\aux_{X,o,\gamma}$ are the same as the $(\op_o)_X$ and Proposition~\ref{gengen} then yields a $(I,O)$-ary $\S_\T$-algebraic operation $\wop$, where:
\begin{equation} \wop_X(o,a)= \lambda \gamma \in \R^X.\, \op_X( o,\lambda i\in  I.\, ai\gamma)\label{algalgops} \end{equation}
which is the natural pointwise definition.
In the case where $I = [n]$ this can be written as:
\begin{equation} \wop_X(o,F_1,\ldots,F_n)= \lambda \gamma \in \R^X.\, \op_X( o,F_1\gamma, \ldots, F_n\gamma)%
\label{finalgalgops} \end{equation}
From Proposition~\ref{geneq} we further have (as is, in any case, evident from a pointwise argument):
\begin{cor}\label{auxeq} Let $\op_i$ be $[n_i]$-ary $\T$-algebraic operations, for $i = 1,\ldots, k$.
Then an equation is satisfied by $\widetilde{op_1},\ldots, \widetilde{op_k}$, if
it is satisfied by $\op_1,\ldots,\op_k$.
\end{cor}

Another way to obtain algebraic operations is to start from the basic selection monad $\S$. Consider an $(I, O \times  \R^I)$-ary generic effect $g\type O \times  \R^I\rightarrow I$ for the identity monad (equivalent via currying to an $(I, O)$-ary generic effect for $\S$).
 Viewed as a $\T$-generic effect $\eta_I\circ g$, via the unit for $\T$ and using Equation~\ref{go}, we obtain an $(I, O \times  \R^I)$-ary $\T$-algebraic operation $\op_g$ where, for $o\in O, \ogamma \in \R^I, u \in \T(X)^I$:
\[(\op_g)_X(\tuple{o,\ogamma},u) = u^\dagger_\T( \eta_I(g(o,\ogamma))) = u (g(o,\ogamma))\]
Then the corresponding auxiliary functions $(\aux_g)_{X}\type  O \times \R^X \times \T(X)^I \rightarrow \T(X)$ are given, using Definition~\ref{auxdef},  by:
\[(\aux_g)_{X,o,\gamma}(u)
    = (\op_g)_X(\tuple{o, \lambda i \in I.\, \ER{\T}{ui}{\gamma}}, u)
    = u(g(o,\lambda i \in I.\, \ER{\T}{ui}{\gamma}))\]

Finally, via Proposition~\ref{gengen}, we obtain the $(I,O)$-ary $\S_\T$-algebraic operation $\wop_g$ corresponding to $\op_g$. For
for $o \in O, a \in \S_\T(X)^I, \gamma \in \R^X$, we have:

\begin{equation} (\widetilde{op_g})_X(o,a)\gamma
= (\aux_g)_{X, o,\gamma}(\lambda i\in I.\, ai\gamma)
= a(g(o,\lambda i \in I.\, \ER{\S_\T}{ai}{\gamma}))\gamma%
 \label{selop}\end{equation}

So  each component $(\widetilde{op_g})_X$ of $\widetilde{op_g}$ uses $g$ to select a branch of $a$, depending only on the parameter $o \in P$ and the rewards $\ER{\S_\T}{ai}{\gamma}$ associated to the  branches of  $a$ relative to the reward continuation $\gamma$. We can turn this observation into a definition. Say that a family of functions
\[f_X\type O \times \S_\T(X)^I \rightarrow \S_\T(X)\]
is an  $(I,O)$-ary-\emph{selection operation} if there is a function $g: O \times \R^I \rightarrow I$ such that
\begin{equation} f_X(o,a)\gamma = a(g(o,\lambda i \in I.\, \ER{\S_\T}{ai}{\gamma}))\gamma\label{selopdef}\end{equation}
Equation~\ref{selop}  then tells us that the selection operations are exactly the algebraic operations of the form
 $\widetilde{\op_g}$ where, modulo currying, $g$ is a basic selection monad generic effect.

We next consider a particular case: binary selection  operations. Here
$O = \mathbbm{1}$ and $I = [2]$. Such operations arise from $[2]$-ary generics $g\type \R^{[2]} \rightarrow [2]$ for the basic
selection monad. Viewed as a binary algebraic operation on $\S_\T$, Equation~\ref{selop} 
becomes:

\[ (\widetilde{op_g})_X(G_0,G_1)\gamma
= \left \{ \begin{array}{ll}
         G_0\gamma & (g(\lambda i \in I.\, \ER{\S_\T}{G_i}{\gamma}) = 0) \\
         G_1\gamma & (g(\lambda i \in I.\, \ER{\S_\T}{G_i}{\gamma}) = 1) \\
\end{array}\right.\]
Note that $[2]$-ary generics $g\type \R^{[2]} \rightarrow [2]$ for the basic selection monad are in bijection with  binary relations $B$ on $\R$, with relations $B$ corresponding to generics $g_B$, where:
 \[ g_B(\ogamma) = 0  \equiv_{\scriptsize \mathrm{def}} \ogamma(0)B\ogamma(1)\]
(read $rBs$ as ``$r$ beats $s$'').
Defining  $op_B$ to be the binary $\S_\T$-algebraic operation $\widetilde{op_{g_B}}$, we have:
\[ (op_B)_X(G_0,G_1)\gamma
= \left \{ \begin{array}{ll}
         G_0\gamma & ( \ER{\S_\T}{G_0}{\gamma}  \, B \, \ER{\S_\T}{G_1}{\gamma})\\
         G_1\gamma & (\ER{\S_\T}{G_0}{\gamma} \, B \ER{\S_\T}{G_1}{\gamma})\\
\end{array}\right. \]

For optimization purposes it is natural to assume $B$ is a  total  order $\geq$. We define $\myor$ to be the resulting binary algebraic operation on $\S_\T$; it is this operation that we  use for the semantics of decision-making in our two languages. Explicitly we have:

\[\myor_X(G_0,G_1)(\gamma) = \left \{\begin{array}{ll} G_0\gamma &
(\mbox{if\ } \ER{\S_\T}{G_0}{\gamma}  \, \geq \, \ER{\S_\T}{G_1}{\gamma} )
\\
                                                                        G_1\gamma & (\mbox{otherwise})
                                       \end{array}\right.\]
This can be usefully rewritten.  For  $\gamma\type X \rightarrow \R$ define  $\max{\gamma}\type X^2\rightarrow X$ (written infix) by:
\begin{equation} x \,\max{\gamma}\, y = \left \{\begin{array}{ll} x & (\mbox{if\ }\gamma(x) \geq \gamma(y))\\
                                                                        y & (\mbox{otherwise})
                                       \end{array}\right.%
\label{maxdef}
\end{equation}
Then:
\begin{equation}
\myor_X(G_0,G_1)(\gamma) = G_0\gamma\;  \max{\ER{\T}{-}{\gamma}} \; G_1\gamma%
\label{ordef}
\end{equation}

Taking $B$ to be a  total  order is equivalent to  using a version of $\argmax$ as a generic effect. First, for finite   totally ordered sets $I$, assuming a  total  order $\geq$ on $\R$, we define
$\argmax_I \in \S(I) = (I \rightarrow \R) \rightarrow I$, by taking $\argmax \gamma$ to be the least  $i \in I$ among those maximizing  $\gamma(i)$. Then $B$ corresponds to $\argmax_{[2]}$, with $[2]$ ordered by setting $0 < 1$. We could as well have used generics picking from finite  totally ordered sets, with resulting choice functions of corresponding arity.

We next investigate the equations that the algebraic operations $\op_B$ obey and their relation to properties of the relations $B$. Define $(\aux_B)_{X}\type  \R^X \times \T(X)^2 \rightarrow \T(X)$ to be  $(\aux_{g_B})_{X}$. So:
\[(\aux_B)_{X,\gamma}(u_0,u_1)
 \eqdef (\aux_g)_{X,\gamma}(\lambda i \in  [2].\, u_i)
 = u_{g(\lambda i \in [2].\, \ER{\T}{u_i}{\gamma})}\]
and we have:
\[ (\aux_B)_{X,\gamma}(u_0,u_1) = \left \{\begin{array}{ll} u_0 &
(\mbox{if\ } \ER{\T}{u_0}{\gamma} \, B \, \ER{\T}{u_1}{\gamma}) \\
                                                                                    u_1 & (\mbox{otherwise})
                                                          \end{array}\right.\]
In particular, for $x_0,x_1 \in X$ we have:
\begin{equation}(\aux_B)_{X,\gamma}(\eta_{\T}(x_0), \eta_{\T}(x_1))(\gamma) = \left \{\begin{array}{ll} \eta_\T(x_0) & (\mbox{if\ } \gamma(x_0) B \gamma(x_1))\\
                                                                        \eta_\T(x_1) & (\mbox{otherwise})
                                       \end{array}\right.\label{etaop} \end{equation}
We see from Proposition~\ref{geneq} that $\op_B$ satisfies an equation if, and only if, $(\aux_B)_{X,\gamma}$ does for every $X$ and $\gamma\type X \rightarrow \R$.

Say that a binary function $f$ is \emph{left-biased} if the following equation holds:
\[f(x,f(y,x)) =  f(x,y) \]
and is \emph{right-biased} if the following equation holds:
\[f(f(x,y),x) =  f(y,x) \]
and recall that a relation $R$ is
\emph{strongly connected} iff, for all $x, y$, either $xRy$ or $yRx$.

\begin{thm}\label{genax} For every binary relation $B$ on $\R$ we have:
\begin{enumerate}
\item $\op_B$ is idempotent.
\item  $\op_B$ is  associative iff $B$ and its complement is transitive.
\item $\op_B$ is left-biased  iff $B$ is 
strongly connected.
\item $\op_B$ is right-biased  iff the complement of $B$ is 
strongly connected.
\item $\op_B$ is not commutative (assuming $\R$ non-empty).

\end{enumerate}
\end{thm}
\begin{proof} Throughout the proof, we use the fact that, like any monad  unit, all components of $\eta_{\S_\T}$ are 1--1.
\begin{enumerate}
\item  This is evident.
\item \begin{enumerate}
\item

Suppose $\op_B$ is  associative,  and choose $r_0,r_1, r_2 \in \R$.
Define $\gamma\type [3] \rightarrow \R$ by:  $\gamma(i) = r_i$, for $i = 0,1,2$, and set
$f = (\aux_B)_{\gamma,[3]}$ and $x_i =(\eta_{\T})_{[3]}(i)$, for $i = 0,1,2$. Note that the $x_i$ are all different.
By Proposition~\ref{geneq} $f$ is associative as $\op_B$ is.

Suppose first that $r_0 B r_1$ and  $r_1B r_2$. Using Equation~\ref{etaop} we see that, as $r_0 B r_1$ and  $r_1 B r_2$, $f(x_0, f(x_1,x_2)) = f(x_0, x_1)  = x_0$,
and $f(f(x_0, x_1),x_2) = f(x_0,x_2)$. So, as $f$ is associative  $f(x_0,x_2) = x_0$. As $x_0 \neq x_2$, we have $x_0Bx_2$, as required.
Suppose next that $\neg\, r_0 B r_1$ and  $\neg\, r_1B r_2$. Then, using Equation~\ref{etaop} again, we see that $f(x_0, f(x_1,x_2)) = f(x_0,x_2)$ and $f(f(x_0, x_1),x_2) = f(x_1,x_2) = x_2$, and, similarly to before, we conclude that $\neg\, r_0B r_2$.
\item For the converse, suppose that $B$ and its complement is transitive. It suffices to prove that every
$f = (\aux_B)_{X,\gamma}\type \T(X)^2 \rightarrow \T(X)$ is associative. Choose $u_i \in \T(X)$ ($i = 0,2$) and set
$r_i = \ER{}{u_i}{\gamma}$.  
The proof divides into cases according as  each of $r_0Br_1$ and $r_1Br_2$ does or does not hold:
\begin{enumerate}
\item Suppose that $r_0Br_1$ and $r_1Br_2$ (and so $r_0Br_2$).
By the definition of $\aux_B$,  we then have:
\[f(u_0, f(u_1,u_2))   =  f(u_0,u_1)  =  u_0  =  f(u_0,u_2) = f(f(u_0, u_1),u_2) \]
\item Suppose that $r_0Br_1$ and $\neg\, r_1Br_2$. Then:
\[f(u_0, f(u_1,u_2))  = f(u_0,u_2)   = f(f(u_0, u_1),u_2) \]
\item Suppose that $\neg\,  r_0Br_1$ and $r_1Br_2$. Then:
\[f(u_0, f(u_1,u_2))  =  f(u_0,u_1)  = u_1  = f(u_1,u_2 )  = f(f(u_0, u_1),u_2) \]
\item Suppose that $\neg\,  r_0Br_1$ and $\neg\, r_1Br_2$.Then $\neg r_0Br_2$, and we have:
\[f(u_0, f(u_1,u_2))  = f(u_0,u_2)  = u_2   = f(u_1,u_2)  = f(f(u_0, u_1),u_2) \]
\end{enumerate}
So in all cases we have
\[f(u_0, f(u_1,u_2)) =  f(f(u_0, u_1),u_2)\]
and so $\op_B$ is associative, as required.
\end{enumerate}

\item

\begin{enumerate}
\item Suppose $\op_B$ is  left-biased and choose $r_0,r_1 \in \R$. Define $\gamma\type [2] \rightarrow \R$ by:  $\gamma(i) = r_i$, for $i = 0,1$, and set $f = (\aux_B)_{\gamma,[2]}$ and $x_i =(\eta_{\T})_{[2]}(i)$, for $i = 0,1$. Note that  $x_0 \neq x_1$.  By Proposition~\ref{geneq} $f$ is left-biased as $\op_B$ is.

 Suppose that $\neg r_1 B r_0$. Then we have:
\[ f(x_0,x_1) = f(x_0,f(x_1,x_0)) =f(x_0,x_0) = x_0 \]
and so, as $x_0 \neq x_1$, $r_0 B r_1$.
\item For the converse, suppose the relation $B$  is strongly connected. It suffices to prove that every
$f = (\aux_B)_{X,\gamma}\type \T(X)^2 \rightarrow \T(X)$ is left-biased. Choose $u_i$ in $\T(X)$ ($i = 0,1$) and set
$r_i = \ER{\T}{u_i}{\gamma}$.    %
Suppose  first that $r_1 B r_0$ holds. Then
$f(u_0,f(u_1,u_0)) =  f(u_0,u_1) $.
Otherwise, as $B$ is strongly connected, we have $\neg\, r_1 B r_0$ and $r_0 B r_1$, and so,
$f(u_0,f(u_1,u_0)) =  f(u_0,u_0) = u_0 = f(u_0,u_1)$.
So in either case we have
$f(u_0,f(u_1,u_0)) =  f(u_0,u_1) $
 as required.

\end{enumerate}
\item
\begin{enumerate}
\item Suppose $\op_B$ is  right-biased and choose $r_0,r_1 \in \R$. Define $\gamma\type [2] \rightarrow \R$ by:  $\gamma(i) = r_i$, for $i = 0,1$, and set $f = (\aux_B)_{\gamma,[2]}$ and $x_i =(\eta_{\T})_{[2]}(i)$, for $i = 0,1$. Note that  $x_0 \neq x_1$.  By Proposition~\ref{geneq} $f$ is right-biased as $\op_B$ is.

Suppose that $\neg(\neg r_0 B r_1)$, i.e., that $r_0 B r_1$. Then we have:\[f(f(x,y),x) =  f(y,x) \]
\[f(x_1,x_0) = f(f(x_0,x_1),x_0)  =f(x_0,x_0) = x_0 \]
and so, as $x_0 \neq x_1$, $\neg r_1 B r_0$.
\item  For the converse, suppose that $\neg B$  is strongly connected. It suffices to prove that every
$f = (\aux_B)_{X,\gamma}\type \T(X)^2 \rightarrow \T(X)$ is right-biased. Choose $u_i \in \T(X)$ ($i = 0,1$) and set
$r_i = \ER{\T}{u_i}{\gamma}$. %
Suppose first that $\neg r_0 B r_1$ holds. Then we have that
$ f(f(x_0,x_1),x_0) = f(x_1,x_0) $.
Otherwise, as $B$ is strongly connected, we have $r_0 B r_1$ and $\neg\, r_1 B r_0$, and so,
$f(u_0,f(u_1,u_0)) =  f(u_0,u_0) = u_0 = f(u_0,u_1)$.
So in either case we have
$f(u_0,f(u_1,u_0)) =  f(u_0,u_1) $
 as required.

\end{enumerate}
\item Choose $r \in \R$.   Define $\gamma\type [2] \rightarrow \R$ by:
 $\gamma(0) = \gamma(1) = r$, and set $f= (\aux_B)_{\gamma, [2]}$ and $x_i = (\eta_{\S_\T})(i)$, for $i =0,1$. Note that $x_0 \neq x_1$. By Proposition~\ref{geneq} it suffices to prove that $f$ is not commutative.

 In case  $rBr$ holds, we have:
\[f(x_0, x_1) = x_0 \neq x_1 = f(x_1, x_0)\]
In case  $rBr$ does not hold, we have:
\[f(x_0, x_1) = x_1 \neq x_0 = f(x_1, x_0)\]
In either case $(\op_B)_{[2]}$ is not commutative.	 \qedhere
\end{enumerate}
\end{proof}

\noindent
Given a binary relation $B$ on $\R$ and an $\S_\T$-algebraic operation $\op\type [n] $, we say that
$\op$ \emph{distributes} over $op_B$ iff for all $X$, $i \in [n]$, $F_j \in \S_\T(X) \; (j\in [n], j \neq i)$, and $G_0,G_1 \in \S_\T(X)$,  we have:
\[\begin{array}{l}\op_X(F_0,\ldots,F_{i-1}, op_B(G_0,G_1), F_{i+1},\ldots,F_{n-1} )  = \\
\hspace{30pt} op_B(\op_X(F_0,\ldots,F_{i-1}, G_0, F_{i+1},\ldots,F_{n-1} ),\op_X(F_0,\ldots,F_{i-1}, G_1, F_{i+1},\ldots,F_{n-1} ))
\end{array} \]
Also, given a binary relation $B$ on $\R$ and a function $f\type \R^n$ we say that $f$ \emph{distributes} over $B$ iff
for all
$0 \leq i  < n$, $r_j \in \R \; (0 \leq j  < n, j \neq i)$, and $s_0,s_1 \in \R$ we have:
\[s \, B \, t \iff f(r_0,\ldots,r_{i-1}, s_0, r_{i+1},\ldots,r_{n-1}) \, B \, f(r_0,\ldots,r_{i-1}, s_1, r_{i+1},\ldots,r_{n-1})\]
We say that an $n$-ary function $f\type \R^n \rightarrow \R$, where $n \geq 1$,  \emph{distributes} over a binary relation $B$ on $\R$ iff it preserves and reflects $B$ in each argument, i.e., iff for $1 \leq i \leq n$ and  $x_1,\ldots,x_{i-1}, y_0,y_1, x_{i+1},\ldots,x_n\in \R$ we have:
\[B(f(x_1,\ldots,x_{i-1}, y_0, x_{i+1},\ldots,x_n), f(x_1,\ldots,x_{i-1}, y_1, x_{i+1},\ldots,x_n))\;\iff\; B(y_0,y_1)\]
\begin{thm}\label{distributes} Let $\op\type [n] $ be a $\T$-algebraic operation,  and  let $B$ be a binary relation on $\R$. If $\op_\R$  distributes over  $B$ then $\wop$ distributes over $op_B$.
\end{thm}
\begin{proof} To keep notation simple we  suppose that $\wop$ is binary and establish distributivity in its second argument. That is, we prove, for any $X$, that:
\[\wop(F, op_B(G,H)) = op_B(\wop(F, G),\wop(F,H))\quad (F,G,H \in \S_\T(X))\]
To do so we use Proposition~\ref{geneq}  and establish the corresponding equation for the auxiliary functions of these operations.
The auxiliary function $\aux_{X,\gamma}\type  \T(X)^2 \rightarrow \T(X)$ of $\wop$ is $\op_X$. So we need to show for any $\gamma\type X \rightarrow \R$ that
\[\op(u, (\aux_B)_{X,\gamma}(v_0,v_1)) = (\aux_B)_{X,\gamma}(\op(u, v_0),\op(u,v_1))\quad (u,v_0,v_1 \in \T(X))\]
From the definition of the auxiliary function of $op_B$ we see that each side of this equation is either $\op(u, v_0)$ or
$\op(u, v_1)$, and that the LHS is $\op(u, v_0)$ iff
\[\ER{\T}{v_0}{\gamma} \, B \, \ER{\T}{v_1}{\gamma}  \qquad (*)\]
and that the RHS is $\op(u, v_0)$ iff
\[\ER{\T}{\op(u, v_0)}{\gamma} \, B \, \ER{\T}{\op(u, v_1)}{\gamma} \qquad (**)\]
As both $\T(\gamma)$ and $\alpha_\T$ are homomorphisms, so is $\ER{\T}{-}{\gamma} = \alpha_\T\circ \T(\gamma)$, and so this last condition is equivalent to:
\[\op_\R(\ER{\T}{u}{\gamma}, \ER{\T}{v_0}{\gamma}) \, B \, \op_\R(\ER{\T}{u}{\gamma}, \ER{\T}{v_1}{\gamma})\]
and we see, using the fact that $\op_\R$  distributes over $B$, that the conditions $(*)$ and  $(**)$ are equivalent. \qedhere

 \end{proof}

\section{A general language with algebraic operations}\label{basic}

The goal of this section is to give some definitions and results---in
particular an adequacy theorem---for a general language with algebraic
operations. We treat our two languages  of later sections as instances of this language via such algebraic operations.

\subsection{Syntax}\label{gen-syntax}

We make use of a standard call-by-value $\lambda$-calculus equipped with algebraic operations.
Our language is a convenient variant of the one in~\cite{PP01} (itself building on  Moggi's computational $\lambda$-calculus~\cite{Moggi89}). The somewhat minor  differences from~\cite{PP01} are that we allow a variety of base types, our algebraic operations may have parameters, and 
we make use of  general big-step transition relations as well as small-step ones.

The types $\sigma,\tau, \ldots$ and terms $L,M,N,\ldots$ of our language are built from:
 \begin{itemize}
 \item[-] a   \emph{basic vocabulary}, consisting  of:
 \begin{enumerate}
 \item \emph{base types}, $b$
(including $\bool$);
\item \emph{constants}, $c \type  b$ of given base types $b$
(including $\true, \false \type \bool$); and
\item  first-order
 \emph{function symbols}, $f  \type   b_1\ldots b_m   \rightarrow   b$, of given arity $b_1\ldots b_m$ and co-arity $b$ (including equality symbols $=_b \type   b\,  \times \,  b  \rightarrow  \bool$),
\end{enumerate}
 together with
 \item[-]  \emph{algebraic operation symbols}    $\op\!\type\! b_1\ldots b_n;m$,
 with given parameter base types $b_1,\ldots , b_n$ and arity $m \in \mathbb{N}$.
 \end{itemize}
 The types are given by:
 \[\sigma \BEQ  b \BOR \unit \BOR \sigma \times \sigma \BOR \sigma \rightarrow \sigma\]
and the terms are given by:
\[\begin{array}{lcl} M  & \BEQ  &  x \BOR c \BOR f(M_1,\ldots,M_m) \BOR \myif L \mythen M \myelse N \BOR \\
&& \op( N_1,\ldots,N_n; M_1,\ldots,M_m)  \BOR \\
 &&\ast \BOR \tuple{M,N} \BOR \fst(M) \BOR \snd(M) \BOR\lambda x\type \sigma.\, M \BOR MN
 \end{array}\]

The languages considered in the next two sections provide examples of this general setup.
We write $\BT$ for the set of base types and $\Con_b$ for the set of constants of type $b$.
We define the \emph{order} (or \emph{rank}) of types  by:
\[\order(b) = \order(\unit) = 0 \qquad \order(\sigma \times \tau) = \max{}{(\order(\sigma),\order(\tau))} \qquad \order(\sigma \rightarrow \tau) = \max{}{(\order(\sigma)+1,\order(\tau))}\]

We work up to $\alpha$-equivalence, as usual, and free variables and substitution are also defined as usual. The typing rules are standard, and omitted, except for that for the algebraic operation symbols, which, aside from their parameters, are polymorphic:
\[\frac{\Gamma \vdash N_1\type b_1, \ldots, \Gamma \vdash N_n\type b_n\quad \Gamma \vdash M_1\type \sigma, \ldots ,\Gamma \vdash M_m\type \sigma}{\Gamma \vdash \op(N_1,\ldots,N_n; M_1,\ldots,M_m) \type \sigma}
\qquad (\op\type b_1\ldots b_n;m)\]
where $\Gamma = x_1\type \sigma_1,\ldots, x_n\type \sigma_n$ is an environment. We write $M\type \sigma $ for $\vdash M \type \sigma$ and say then that the (closed) term $M$ is \emph{well-typed}; such terms are the \emph{programs} of our language.
We employ standard  notation, for example for local definitions writing
$\mylet x\type \sigma \mybe M \myin N$
for  $(\lambda x\type \sigma.\, N)M$.
We also use a cases form
\[\mathtt{cases} \, M_1 \!\Rightarrow\! N_1\,  \mid \, \ldots \, \mid \, M_n\! \Rightarrow \! N_n \myelse N_{n+1} \; (\mbox{for }n \geq 0)\]
defined by iterated conditionals (where the $M_i$ are boolean).

Moggi's language has local definitions and computational types $T\sigma$ (with associated term syntax)  as primitives; these can be viewed as abbreviations in our language,  in particular setting $T\sigma = 1\rightarrow \sigma$.

\subsection{Operational semantics}\label{gen-opsem}

The operational semantics of programs is given in three parts: a small-step semantics, a big-step semantics, and an evaluation function.
We make use of evaluation contexts, following~\cite{FF87}. The set of \emph{values} $V,W,\ldots $ is given by:
\[V \BEQ c \BOR \ast \BOR \tuple{V,W} \BOR \lambda x\type\sigma.\, M\]
where we restrict  $\lambda x\type\sigma.\, M$ to be closed.  We write  $\Val_\sigma$ for the set of values of type $\sigma$, i.e., the $V$ such that $V\type\sigma$.

The \emph{evaluation contexts} are given by:
\[\begin{array}{lcl}\E & \BEQ & [\;] \BOR f(c_1,\ldots,c_{k-1}, \E, M_{k+1}, \ldots, M_m) \BOR  \myif \E \mythen M \myelse N \BOR\\
&&  \op(c_1,\ldots,c_{k-1}, \E, N_{k+1}, \ldots, N_n; M_1,\ldots,M_m)\\
&&   \tuple{\E,N} \BOR \tuple{V,\E}\BOR \fst(\E) \BOR \snd(\E)  \BOR \E N
\BOR (\lambda x\type \sigma.\, M)\E
\end{array}\]
and are restricted to be closed.
The \emph{redexes} are defined by:
\[\begin{array}{lcl} R & \BEQ & f(c_1,\ldots,c_m) \BOR  \myif   \true \mythen M \myelse N \BOR \myif \false \mythen M \myelse N \BOR\\
&& \op(c_1,\ldots,{c_n};M_1,\ldots,M_m) \\
&&  \fst(\tuple{V,W})\BOR \snd(\tuple{V,W}) \BOR (\lambda x\type \sigma.\, M)V
 \end{array}\]
and are restricted to be closed. Any
program is of one of two mutually exclusive forms: it is either a value $V$ or else has the form $\E[R]$ for a unique evaluation context $\E$ and redex~$R$.

 We define two small-step transition relations on redexes, \emph{ordinary} transition relations
 and algebraic operation symbol transition relations:
 \[ R \rightarrow N \qquad \mbox{and} \qquad R \xrightarrow[\op_i]{c_1,\ldots,c_n} N \;\; (\op\type b_1\ldots b_n;m {\rm \ and\ } i=1,m)\]
 The idea of the algebraic operation symbol transitions is to indicate with which parameters an operation is being executed, and which of its arguments is then being followed.
 The definition of the first kind of transition is standard; we just mention that for each function symbol  $f\type b_1\dots b_m \rightarrow b$
 and constants $c_1\type b_1,\ldots, c_m\type b_n$, we assume we are given a constant $\val_{f}(c_1, \ldots, c_m)\type b$, where, in the case of equality, we have:
 \[\val_{=_b}(c_1,c_2) = \left \{\begin{array}{ll} \true & (\mbox{if $c_1 = c_2$})\\
                                                                         \false & (\mbox{otherwise})\end{array} \right.\]
We then have the ordinary transitions:
 \[ f(c_1,\ldots,c_m) \rightarrow c \quad (\val_f(c_1, \ldots, c_m) = c)\]
 The algebraic operation symbol transition relations are  given by the following rule:
\[\op(c_1,\ldots,c_n;M_1,\ldots,M_m) \xrightarrow[\op_i]{c_1,\ldots,c_n} M_i\]
 We next extend these transition relations to corresponding ordinary and algebraic operation symbol transition relations on programs
 \[M \rightarrow M' \quad \mbox{and} \quad M \xrightarrow[\op_i]{c_1,\ldots,c_n} M'\]
 To do so, we use evaluation contexts in a standard way by means of the following rules:
  \[\frac{R \rightarrow M'}{\E[R] \rightarrow \E[M']} \qquad \quad
   \frac{R \xrightarrow[\op_i]{c_1,\ldots,c_n} M'}{\E[R] \xrightarrow[\op_i]{c_1,\ldots,c_n} \E[M']}\]
  These transition relations are all deterministic.

   For any program $M$ which is not a value, exactly one of two mutually exclusive possibilities holds:
\begin{itemize}
\item[-]  For some program $M'$ \[M\rightarrow M'\] In this case $M'$ is determined and of the same type as
  $M$.
  \item[-]   For some $\op\type b_1\ldots b_n;m$ and
  $c_1\type b_1,\ldots,c_n\type b_n$ \[M \xrightarrow[\op_i]{c_1,\ldots,c_n} M_i\] for all $i = 1,\ldots, n$ and some $M_i$.
  In this case $\op$, the $c_j$ and the $M_i$ are uniquely determined and the $M_i$ have the same type as $M$.
  \end{itemize}
    We say a program $M$ is \emph{terminating} if there is no infinite chain of (small-step) transitions from $M$.
  \begin{lem}\label{ter} Every program is terminating.
  \end{lem}
  \begin{proof} This is a standard computability argument; see the proof of Theorem 1 in~\cite{PP01} for some detail.
  One defines a computability predicate on values by induction on types, and then extends it to well-typed terms by taking such a term $M$ to be computable if there is no infinite chain of
  (small-step) transitions from $M$, and every terminating sequence of small-step transitions from $M$ ends in a computable value.
  \end{proof}

  Using the small-step relations one defines big-step ordinary and algebraic operation symbol transition relations  by:
  \[\frac{M \rightarrow^* V}{M \Rightarrow V}  \qquad
  \frac{M \rightarrow^* M' \quad M' \xrightarrow[\op_i]{c_1,\ldots,c_n} M''}
          {M \xRightarrow[\op_i]{c_1,\ldots,c_n} M''}\]
   For any program $M$ which is not a value, similarly to the case of the small-step relations, exactly one of two mutually exclusive possibilities holds:
\begin{itemize}
\item[-]  For some value $V$ \[M\Rightarrow V\] In this case $V$ is determined and of the same type as
  $M$.
  \item[-]   For some $\op\type b_1\ldots b_n;m$ and
  $c_1\type b_1,\ldots,c_n\type b_n$ \[M \xRightarrow[\op_i]{c_1,\ldots,c_n} M_i\] for all $i = 1,\ldots,n$ and some $M_i$.
  In this case $\op$, the $c_j$ and the $M_i$ are uniquely determined and the $M_i$ have the same type as $M$.
  \end{itemize}
  The big-step transition relations from a given program $M$ form a finite tree with values at the leafs, with all transitions, except for those leading to values, being algebraic operation symbol transitions, and with transitions of algebraic operation symbols  of type $(w;n)$ branching $n$-fold. We write $\norm{M}$ for the height of this tree.

 Rather than use trees, we follow~\cite{PP01} and use \emph{effect values}  $E$. These give the same information and, conveniently, form a subset of our programs. They are defined as follows:
 \[E \BEQ V \BOR \op(c_1,\ldots,c_n; E_1,\ldots,E_m)\]
 (Our effect values are a finitary version of the interaction trees of~\cite{XiaZHHMPZ20}).
 Every program
 $M\type\sigma$ has an effect value $\Op(M)\type\sigma$ defined using the big-step transition relations:
 \[\Op(M) = \left \{ \begin{array}{ll} V & (\mbox{if\ }M \Rightarrow V)\\
                                                         \op(c_1,\ldots,c_n; \Op(M_1),\ldots,\Op(M_m)) &
                                            ( \mbox{if\ }M \xRightarrow[\op_i]{c_1,\ldots,c_n} M_i \mbox{ for } i = 1,m)
                               \end{array}
                     \right.\]
This definition is justified by induction on  $\norm{M}$.
Note that $\Op(E) = E$, for any effect value  $E\type \sigma$. Further, program transitions and evaluations closely parallel each other, indeed:
\begin{equation}\label{opequiv1} M \Rightarrow V \iff \Op(M) = V \end{equation}
and
\begin{equation}\label{opequiv2} M \xRightarrow[c_1,\ldots, c_n]{\op_i} M_i \iff \Op(M) \xRightarrow[c_1,\ldots, c_n]{\op_i} \Op(M_i)\end{equation}

We next give a proof-theoretic account of the evaluation function $\Op$ to help us prove our general adequacy theorem. There is a natural equational theory for the operational semantics, with evident rules, which establishes judgments of the form $ \vdash_o  M = N\type \sigma$, taken to be well-formed in case  $ M\type \sigma$  and  $ N\type \sigma$. The axioms are the small-step reductions for the redexes together with a commutation schema that algebraic operations commute with evaluation contexts; they are given (omitting type information) in Figure~\ref{axioms}.
\begin{figure}[h]
\[\begin{array}{c}
f(c_1,\ldots,c_m) = c \quad (\val_f(c_1, \ldots, c_m) = c)\\\\
\myif   \true \mythen M \myelse N = M \qquad  \myif \false \mythen M \myelse N = N\\\\
 \fst(\tuple{V,W})= V \qquad  \snd(\tuple{V,W}) = W\\\\
  (\lambda x\type \sigma.\, M)V = M[V/x]\\\\
  \E[\op(c_1,\ldots, c_n; M_1,\ldots, M_m)] =   \op(c_1,\ldots, c_n; \E[M_1],\ldots, \E[M_m])\\
\end{array}\]
\caption{Axioms}%
 \label{axioms}
 \end{figure}

  \begin{lem}\label{proof-lemma} For any well-typed term $M\type \sigma$ we have:
  \begin{enumerate}
  \item \(M \rightarrow M' \implies \vdash_o M = M'\type\sigma \)
  \item
  \(M \xrightarrow[\op_i]{c_1,\ldots,c_n} M_i, \mbox{ for } i = 1,\ldots,m \implies
  \vdash_o M = \op(c_1,\ldots,c_n; M_1,\ldots, M_m)\type\sigma\)

  \item  \(M \Rightarrow V \implies \vdash_o M = V \type\sigma \)
  \item
    \(M \xRightarrow[\op_i]{c_1,\ldots,c_n} M_i, \mbox{ for } i = 1,\ldots,m \implies
  \vdash_o M = \op(c_1,\ldots,c_n; M_1,\ldots, M_m)\type\sigma\)
  \end{enumerate}
  \end{lem}
  \noindent
The following proposition is an immediate consequence of this lemma:
  \begin{prop}\label{proof} For any program $M\type \sigma$ we have:
  \[\vdash_o M = \Op(M)\type\sigma\]
  \end{prop}

There is a useful substitution lemma. Given any effect value $E\type b$, a nonempty finite set $u \subseteq \Con_b$ that includes all the constants of type $b$ in $E$, and a function $g$ from $u$ to programs of type $b'$, 
 $E[g]\type b'$, the substitution $g$ of programs for constants,  is defined homomorphically by:
\[\begin{array}{ccc}
c[g]& = & g(c) \\
 \op(c_1,\ldots,c_n; E_1,\ldots,E_m)[g] & = &  \op(c_1,\ldots,c_n; E_1[g],\ldots,E_m[g])\\
 \end{array}\]
Let $c_1,\ldots,c_n$ enumerate $u$ (the order does not matter) and define $\myF_g\type b \rightarrow b'$ to be
\[\lambda x:b.\, \mathtt{cases} \, x = c_1 \Rightarrow g(c_1)\,  \mid \, \ldots \, \mid \, x = c_{n-1} \Rightarrow g(c_{n-1}) \myelse g(c_{n})\]

With this notation we have:
\begin{lem}\label{consub}
\[\Op(\myF_gE) = \Op(E[g])\]
\end{lem}
\begin{proof} The proof is a structural  induction on $E$.  For $E$ a constant $c$ we have:
\[\begin{array}{lcl}\Op(\myF_gE) & = & \Op(\mathtt{cases} \, c = c_1 \Rightarrow g(c_1)\,  \mid \, \ldots \, \mid \, c = c_{n-1} \Rightarrow g(c_{n-1}) \myelse g(c_{n}))\\
& = & \Op(g(c))\\ & = & \Op(c[g])
\end{array}\]
and for $E$ of the form $\op(c_1,\ldots,c_n; E_1,\ldots,E_m)$ we have:
\begin{align*}
\Op(\myF_gE) &=  \op(c_1,\ldots,c_n; \Op(\myF_gE_1),\ldots,\Op(\myF_gE_m))\\
& = \op(c_1,\ldots,c_n; \Op(E_1[g]),\ldots,\Op(E_m[g]))\\
& = \Op(\op(c_1,\ldots,c_n; E_1[g],\ldots,E_m[g])\\
& = \Op(\op(c_1,\ldots,c_n; E_1,\ldots,E_m)[g])\\
& = \Op(E[g]) \qedhere
\end{align*}
\end{proof}

  \subsection{Denotational semantics}\label{gen-densem}
  The semantics of our language makes use of a given strong monad, following that of Moggi’s computational $\lambda$-calculus~\cite{Moggi89}.  In order to be able to give semantics to effectual operations we use the apparatus of generic effects and algebraic operations as discussed above.
  For the sake of simplicity we work in the category of sets, although the results go through much more generally, for example in any cartesian closed category with binary sums.

  To give the semantics of our language a number of ingredients are needed. We assume given:
  \begin{itemize}
 \item[-] a (necessarily) strong monad $\M$ on the category of sets,
  \item[-] nonempty sets $\sden{b}$ for the base types $b$ (with $\sden{\bool} = \B \eqdef \{0,1\}$),
  \item[-] elements $\sden{c}$ of $\sden{b}$ for constants $c\type b$ (with $\sden{\true} = 1$ and $\sden{\false} = 0$),
  \item[-]  functions $\sden{f} \type \sden{b_1} \times \cdots \times \sden{b_m} \rightarrow \sden{b}$ for function symbols $f \type b_1\ldots b_m \rightarrow b$, and
  \item[-] generic effects
\[g_{\op}\type \sden{b_1}\times \cdots \times \sden{b_n} \rightarrow \M([n])\]
for algebraic operation symbols $\op\type b_1\ldots b_n;m$.
  \end{itemize}
  We further assume that different constants of the same type receive different denotations, i.e., the
  $\sden{\mbox{-}}\type \Val_b \rightarrow \sden{b}$ are 1--1 (so we can think of constants as just names for their denotations, just as one thinks of numerals), and that the given denotations of function symbols are consistent with their operational semantics in that:
  \begin{equation}\val_f(c_1, \ldots, c_m) = c \implies \sden{f}(\sden{c_1}, \ldots, \sden{c_m}) = \sden{c}\label{opdencon}\end{equation}

  With these ingredients, we can give our language its semantics. Types are interpreted by putting:
  \[\begin{array}{lcl}\den{\Msem}{b} & = & \sden{b}\\
  \den{\Msem}{\sigma \times \tau} & = & \den{\Msem}{\sigma} \times \den{\Msem}{\tau}\\
  \den{\Msem}{\sigma \rightarrow  \tau} & = & \den{\Msem}{\sigma} \rightarrow \M(\den{\Msem}{\tau})
  \end{array}\]
  To every  term
  \[\Gamma \vdash N\type \sigma\]
  we associate a function
  \[\den{\Msem}{\Gamma \vdash N\type \sigma}\type \den{\Msem}{\Gamma} \rightarrow \M(\den{\Msem}{\sigma})\]
  where $\den{\Msem}{x_1:\sigma_1,\ldots,x_n:\sigma_n} \eqdef \den{\Msem}{\sigma_1} \times \cdots \times \den{\Msem}{\sigma_n}$. When the typing $\Gamma \vdash N\type \sigma$ is understood, we generally write $\den{\Msem}{N}$ rather than $\den{\Msem}{\Gamma \vdash N\type \sigma}$.

  The semantic clauses for conditionals and the product and function space terms are standard, and we omit them.
For constants $c\type b$ we put:
\[\den{\Msem}{c}(\rho) \;=\;   {(\eta_\M)}_{\sden{b}}(\sden{c})\]
For function symbol applications $f(M_1,\ldots,M_m)$, where $f\type b_1\ldots b_m \rightarrow b$, we put:
\[\den{\Msem}{f(M_1,\ldots,M_m)}(\rho)\;=\; \sden{f}^{\sim}(\sden{\Msem}{(M_1)}(\rho),\ldots, \sden{\Msem}{(M_m)}(\rho))\]
where
\[\sden{f}^{\sim} \type \M(\sden{b_1}) \times \cdots \times \M(\sden{b_m})\rightarrow \M(\sden{b})\]
is obtained from $\sden{f}$ in a standard way e.g., via iterated Kleisli extension.
 For
  terms $\Gamma \vdash \op(N_1,\ldots, N_n; M_1,\ldots,M_m)\type \sigma$,
 where  $\op\type b_1\ldots b_n;m$, we  make use of the algebraic operation
  \[\op_X\type (\sden{b_1} \times \cdots \times \sden{b_n}) \times \M(X)^{[m]} \rightarrow \M(X)\]
 corresponding to the generic effects $g_{\op}$  and put:
  \[
  \begin{array}{l}\den{\Msem}{\op(N_1,\ldots, N_n; M_1,\ldots,M_m)}(\rho) =\\
  \qquad\qquad\quad \op^\dagger_{\sden{\sigma}}(\tuple{\den{\Msem}{N_1}(\rho),\ldots, \den{\Msem}{N_n}(\rho)},
    \tuple{\den{\Msem}{M_1}(\rho),\ldots, \den{\Msem}{M_m}(\rho)})
  \end{array}\]
  where $\op^\dagger_{\sden{\sigma}}$ is again defined in a standard way, as discussed in Section~\ref{genalgops}.
We further give values $V\!\type \!\sigma$ an \emph{effect-free}  (or \emph{pure}) semantics $\den{\Msem_p}{V} \!\in \!\den{\Msem}{\sigma}$:
\[\begin{array}{lcl}
\den{\Msem_p}{c} & =  &  \sden{c}\\
\den{\Msem_p}{\tuple{V,V'}} & =  &  \tuple{\den{\Msem_p}{V},\den{\Msem_p}{V'}}\\
\den{\Msem_p}{\lambda x\type\tau.\, N} & =  & \den{\Msem}{x\type\tau\vdash  N\type \tau'}
\end{array} \]
  This effect-free semantics of values $V\!\type \!\sigma$ determines their denotational semantics:
  \[\den{\Msem}{V}(\rho)= (\eta_{\M})_{\den{\Msem}{\sigma}}(\den{\Msem_p}{V})\]
 Below, we regard the effect-free semantics as providing functions:
\[\Msem_p\type \Val_{\sigma} \rightarrow \den{\Msem}{\sigma}\]
  \subsection{Adequacy}\label{gen-ad}

  Our proof system is consistent relative to our denotational semantics:
 \begin{lem}\label{consistent} If $\vdash_o M = N\type\sigma$ then $\den{\Msem}{M} = \den{\Msem}{N}$.
 \end{lem}
The proof of this lemma uses the naturality condition
on algebraic operations  to establish the soundness of the commutation schema.

Our general adequacy theorem  is an immediate consequence of  Proposition~\ref{proof} and Lemma~\ref{consistent}:
  \begin{thm}\label{basic-ad} For any program $N$ we have:
  $\den{\Msem}{N} = \den{\Msem}{\Op(N)}$.
  \end{thm}
  This adequacy theorem differs somewhat from the usual ones where the denotational semantics determines termination and the denotation of any final result; further, for base types they generally determine the  value produced by the operational semantics. In our case the first part is not relevant as terms always terminate. We do have that the denotational semantics determines the denotation of any final result. For base types (as at any type) it determines the effect values produced up to their denotation, though the extent of that determination  depends on the choice of the generic effects.

  \subsection{Program equivalences and purity}\label{gen-equations}

  The equational system of Section~\ref{gen-opsem}, helps prove adequacy, but  is too weak  for our purposes which are to establish completeness results for programs of base type. Moggi gave a suitable consistent and complete system  for his computational $\lambda$-calculus in~\cite{Moggi89}. His system has equational  assertions $\Gamma \vdash M = N\type \sigma$ and  purity (meaning effect-free) assertions   $\Gamma \vdash M\downarrow_{\sigma}$; we always assume that the terms are appropriately typed, and may omit types or environments when the context makes them clear. One can substitute a term $M$ for a variable  in Moggi's system only if one can prove $M\downarrow_{\sigma}$.

 Our $\lambda$-calculus is an extension of Moggi's and we extend his logic correspondingly; an alternate approach, well worth pursuing, would be to use instead the purely equational fine-grained variant of the computational $\lambda$-calculus: see~\cite{LevyPT03}. We keep Moggi's axioms and rules, other than those for computational types $T\sigma$, but extended to our language. (If we set $T\sigma = 1 \rightarrow \sigma$, then the rules for computational types, extended to our language, are derived.)

  For conditionals we add:
 \[ \myif \true \mythen M \myelse N = M \qquad  \myif \false \mythen M \myelse N = N\]
 \[ u(x) = \myif x \mythen u(\true) \myelse u(\false)\]
 For the algebraic operations we add two equations,
 one:
 \begin{equation}\label{scontext} u(\op(y_1,\ldots,y_n; M_1,\ldots, M_n)) = \op(y_1,\ldots,y_n; u(M_1),\ldots,u(M_n))\end{equation}
 expressing their naturality (and generalizing the commutation schema of Figure~\ref{axioms}), and the other:
\begin{equation}\label{parord} {\small \begin{array}{l}\op(N_1,\ldots,N_n; M_1,\ldots,M_m) =
 \begin{array}{ll}\mathtt{let}\,  y_1\!\type\! b_1,\ldots, y_n\!\type\! b_n\, \mathtt{be}\, N_1,\ldots,N_n\\
 \mathtt{in}\,  \op(y_1,\ldots,y_n; M_1,\ldots,M_m)
 \end{array}\;  (\mbox{no $y_j$ in any $\FV(M_i)$})
 \end{array}}\end{equation}
expressing the order of evaluation of the parameter arguments of $\op$. For function symbols and constants we add the purity axiom $c\downarrow$ and the equation in Figure~\ref{axioms}. This equation enables us to evaluate function symbol applications to constants within our proof system. One could certainly add further useful axioms and rules (e.g., that some function on base types is commutative or a form of induction if the natural numbers were a base type); indeed it would be natural to extend to a predicate logic. However, such extensions are not needed for our purposes.

We write
\[\Gamma \vdash_{\Ax} M = N\type \sigma \;\;\mbox{and}\;\; \Gamma \vdash_{\Ax} M\downarrow_\sigma\]
to mean $M=N$ (resp.\ $M\downarrow_\sigma$) is provable from a  set of equational or purity axioms $\Ax$ (where $\Gamma \vdash M\type \sigma$ and $\Gamma \vdash N\type \sigma$). In particular all the axioms of Figure~\ref{axioms} are provable.
An equational  assertion is true (or holds) in $\Msem$,
written $\Gamma \models_{\Msem} M = N\type\sigma$
if $\den{\Msem}{M} = \den{\Msem}{N}$;
similarly, a purity assertion is true (or holds) in $\Msem$,
written $\Gamma\models_{\Msem} M\downarrow_{\sigma}$, if $\exists a \in \den{\Msem}{\sigma}.\, \den{\Msem}{M}(\rho) = \eta_{M}(a)$.
A \emph{theory}, i.e., a set of axioms, $\Ax$ is \emph{valid} in $\Msem$ if all the assertions in $\Ax$ are true in $\Msem$.

\emph{Equational consistency} holds, meaning that, if a theory $\Ax$ is valid in $\Msem$ then:
\[\Gamma \vdash_{\Ax} M = N\type \sigma \implies \Gamma \models_{\Msem} M = N\type \sigma\]
as does the analogous \emph{purity consistency}.

We can use $\Ax$ to give axioms for particular algebraic operations. For example, we  consider languages  with a binary decision algebraic operation symbol $\myor\type \varepsilon; 2$ with semantics given by the algebraic operation family $\myor_X$ of Definition~\ref{ordef}. Here the associative axioms
\[(L\,\myor\, M) \, \myor \, N = L\,\myor\, (M \myor \, N)\]
hold at all types as, by Theorem~\ref{genax}, every component $\myor_X$ is associative. We will do this extensively for our two languages, as in Figures~\ref{termequivs} and~\ref{ptermequivs}, below.

  \section{A  language of choices and rewards}\label{first}

  Building on the framework of Section~\ref{basic}, in this section we define and study a language with constructs for choices and rewards.

  \subsection{Syntax}\label{subsec:first-syntax}

  For the basic vocabulary of our language,  in addition to the boolean primitives of Section~\ref{gen-syntax}, we assume available: a base type $\reward$; a constant $0\type \reward$; and function symbols $+\type \reward\,\reward\rightarrow \reward$ and $\leq\type \reward \reward \rightarrow \bool$. There are exactly  two  algebraic operation symbols: a \emph{choice operation} $\myor\type \varepsilon;2$ to make binary choices, and a \emph{reward operation} $\give\type \reward; 1$, to prescribe rewards.
  We leave any other base type symbols, constants, or function symbols unspecified.

  We may use infix for $+$ and $\leq$. Similarly,  
 we may use infix notations   $M_0 \, \myor\, M_1$ or $N\cdot M$ for the algebraic operation terms $\myor(; M_0,M_1)$ and $\give(N;M)$.
  The signature $\myor\,\type\, \varepsilon;2$ means that $M_0$ and $M_1$ must have the same type and that is then the type of $M_0 \hspace{0.75pt}\myor\hspace{0.75pt} M_1$; the signature $\give\type \reward; 1$ means $N\cdot M$ has the same type as $M$ and that $N$ must be of type $\reward$.
For example, assuming that $5$ and $6$ are two constants of
  type $\reward$, we may write the tiny program:
\[(5 \cdot \true) \hspace{1pt}\myor\hspace{0.75pt} (6 \cdot \false)\type \bool\]
Intuitively, this program could potentially return either $\true$ or $\false$, with respective rewards $5$ and $6$. In the intended semantics that maximizes rewards, then, the program returns $\false$ with reward $6$.

When designing our language, we could as well have used
choice functions of any finite arity, as in the example in Figure~\ref{smart-ex}.
However we felt that binary choice
was sufficiently illustrative.

  \subsection{Rewards and additional effects}\label{subsec:firstrewards}

 For both the operational and denotational semantics of our language we need  a set of rewards $\R$ with appropriate structure and a monad employing it.
 So, we assume such a set $\R$ is available, and that it is  equipped with:
 \begin{itemize}
 \item a commutative monoid structure, written additively, and
 \item  a total order with addition preserving and reflecting the order in its first argument (and so, too, in its second), in that, for all $r,s,t \in \R$:
\[r \leq s \iff r + t \leq s + t\]
 \end{itemize}
For example, $\R$ could be the  reals (or the nonnegative reals) with addition, or the positive reals with multiplication, in all cases with the usual order.
We further assume that there is an element $\sden{c}$ of $\R$ for each $c\type \Rew$ (with, in particular, $\sden{0} = 0$), and that  $\R$ is \emph{expressively non-trivial} in that there is a $c\type\Rew$ with $\sden{c} \neq 0$.

Our monad is the so-called \emph{writer} monad $\W(X) = \R \times X$, defined using the commutative monoid structure on $\R$.
The operational semantics defined below evaluates programs $M$ of type $\sigma$ to pairs $\tuple{r,V}$, with $r \in \R$ and $V\type\sigma$, that is to elements of $\W(\Val_\sigma)$. The denotational semantics uses the selection monad augmented with  the writer monad, as described in Section~\ref{selection}.

The writer monad is the free-algebra monad for  $\R$-actions, i.e., the algebras with an $\R$-indexed family of unary operations, which we write as $\give(r, -)$ or $r \cdot -$, satisfying the equations
\begin{equation} \label {reward-eqn}0 \cdot x = x \qquad r \cdot s \cdot x = (r +s) \cdot x\end{equation}
The resulting algebraic operation $(\give_W)_X\type \R \times \W(X) \rightarrow \W(X)$ is given by:
\[(\give_W)_X(r, \tuple{s,x}) = \tuple{r+s,x}\]
and is induced by the generic effect $(g_{\W})_{\give}: \R \rightarrow \W([1])$, where $(g_{\W})_{\give}(r) \eqdef \tuple{r,\ast}$.
We generally write applications of $(\give_{\W})_X$ using an infix operator, $(\cdot_{\hspace{0.3pt}\W})_X$, and, in either case, may drop subscripts when they can be understood from the context.  As $\R$ is itself an $\R$-action (setting $r\cdot s = r + s$), we obtain a $\W$-algebra  $\alpha_\W\type \W(\R) \rightarrow \R$  as described in Section~\ref{genalgops},  finding that $\alpha_\W = +$.

\subsection{Operational semantics}\label{sopsem1}

While the operational semantics of Section~\ref{basic} is ordinary and does not address optimization, the
\emph{selection operational semantics} selects  an optimal choice strategy, 
as suggested in the Introduction.
Below we prove an adequacy
  result relative to a  denotational semantics using the selection monad $\S_{\W}$. We thereby give a compositional account of a global quantity: the optimal reward of a program.

 For the ordinary operational semantics, we assume available functions $\val_f$ for the function symbols $f$ of the basic vocabulary, as discussed in Section~\ref{gen-opsem}. The global operational semantics selects strategies maximizing the reward they obtain. To define such strategies we employ the version of $\argmax$ defined in Section~\ref{selection}: given a finite  totally-ordered set $S$ and a reward function
$\gamma\! \type\! S \rightarrow \R$,  $\argmax_X(\gamma)$ selects the least $s \in S$ maximizing $\gamma(s)$.
So,  totally ordering $S$ by:
 \[s \preceq_\gamma s'\; \iff\;\gamma(s) > \gamma(s') \vee (\gamma(s) = \gamma(s') \wedge s \leq s')\]
 the selection is of the least element in this  total order.
 It is convenient to use the notation $\argmax\; s \type S.\,  e$ for $\argmax(\lambda s \in S.\, e)$.

 We next define our strategies. The idea is to view an effect value $E\type \sigma$ as a one-player game for Player. The subterms $E'$ of $E$ are the positions of the game.
 In particular:
\begin{itemize}
\item[-] if $E'$ is a value, then $E$ is a final position and the reward is $0$;
\item[-]  if $E' = E'_0\, \myor\, E'_1$ then Player can choose whether to move to the position $E'_0$ or the position $E'_1$; and
\item[-]  if $E' = c\cdot E''\type \sigma$  then Player moves to $E''$ and $\sden{c}$ is added to the final reward.
\end{itemize}

\noindent
The finite set $\Str(E)$
of strategies of an effect value
$E$ is defined by the following rules, writing $s\type E$ for $s \in \Str(E)$:
\[\ast\type V \qquad \frac{s\type E_1}{1s \type E_1\, \myor\, E_2} \qquad \frac{s\type E_2}{2s \type E_1\, \myor\, E_2}\qquad \frac{s\type E}{s\type c\cdot E}\]

These strategies can be reformulated as boolean functions on choice subterms; though standard, this is less convenient. Equivalently, one could work with boolean functions on choice nodes (terms) of the tree naturally associated to a term by the big-step reduction relation, noting that this tree is isomorphic to effect values considered as trees (as we see from equivalences~\ref{opequiv1} and~\ref{opequiv2}). In this way we would obtain an equivalent optimizing operational semantics which makes no use of effect values. We preferred to work with effect values as they provide a convenient way to work directly with trees formulated as terms. There are also probabilistic strategies, although, as is  generally true for MDPs~\cite{Bel57}, they would not change the optimal expected reward.

For any effect value $E\type\sigma$,
the outcome
$\Out(s,E) \in \R \times \Val_{\sigma} = \W(\Val_{\sigma})$
of  a strategy $s\type E$
is defined by:
\[\begin{array}{lcl}
\Out(\ast, V) & = & \tuple{0,V}\quad  (= \eta_{\W}(V))\\[0.2em]
\Out(1s, E_1 \,\myor\, E_2) & = & \Out(s, E_1)\\[0.2em]
\Out(2s, E_1 \,\myor\, E_2) & = & \Out(s, E_2)\\[0.2em]
\Out(s,c\cdot E) & = & \give_W(\sden{c},\Out(s,E) )
\end{array}\]
We can then define the reward of such a strategy by:
\[ \Rew(s, E) = \pi_1(\Out(s,E))\]
Note that $\pi_1\type \R \times X \rightarrow \R$ can be written as $\ER{\W}{-}{0} = \alpha_{\W}\circ\W(0_X)$, with $0_X\type X \rightarrow \R$  the constantly $0$ reward function.

As there can be several strategies maximizing the reward of a game, we need a way of choosing between them. We therefore define  a  total order $\leq_E$ on the strategies of a given game $E\type \sigma$:
\begin{itemize}
\item Game is $V$:
\[\ast \leq_V \ast\]
\item Game is $E_1 + E_2$:
\[(i,s) \leq_{E_1 + E_2} (j,s') \iff
\begin{array}{l} i < j \quad \vee\\
                                                i = j = 1 \wedge s\leq_{E_1}s' \quad \vee \\
                                               i = j = 2 \wedge s\leq_{E_2}s'
\end{array}\]
\item Game is $c\cdot E$:
\[s \leq_{c\cdot E} s' \iff s \leq_E s' \]

\end{itemize}

\noindent
We can now give our selection operational semantics $\SOp(M)\in \R \times \Val_{\sigma} = \W(\Val_\sigma)$ for programs $M:\sigma$. We first find the $\Op(M)$-strategy $s_{\mathrm{opt}}$ maximizing the reward; if there is more than one such strategy, we take the least, according to the $\Op(M)$-strategy  total order $\leq_{\Op(M)}$.
So we set:
\[s_{\mathrm{opt}} \eqdef \argmax \, s\type \Op(M).\,\Rew(s,\Op(M))\]
and then we use that strategy to define $\SOp(M)$ by setting:
\[\SOp(M) = \Out(s_{\mathrm{opt}}, \Op(M))\]
With this idea, the definition is:
\[\SOp(M) \eqdef \Out(\argmax \, s\type \Op(M).\,\Rew(s,\Op(M)), \Op(M))\]
Note that $\SOp(M) = \SOp(\Op(M))$. (This follows from the form of the definition of the optimizing operational semantics and the fact that $\Op^2(M) = \Op(M)$.)

While the operational semantics is defined by a global optimization over all strategies, it can be equivalently given locally without reference to any strategies. We  first need two lemmas. Their statements use the  $\max{\gamma}$ infix notation introduced in Definition~\ref{maxdef}. We omit their straightforward proofs.
\begin{lem}\label{trivial}Given functions $X \xrightarrow{g} Y \xrightarrow{\gamma} \R$, for all $u,v \in X$ we have
\[g(u \,\max{ \gamma \circ g} v ) \;\; = \;\; g(u) \,\max{\gamma}\, g(v)\]
\end{lem}
\begin{lem}[First argmax lemma]\label{argmax1} Let $S_1 \cup S_2$ split a finite  total order $\tuple{S,\leq}$ into two with $S_1 < S_2$ (the latter in the sense that $s_1 < s_2$ for all $s_1 \in S_1$ and $s_2 \in S_2$). Then, for all $\gamma\type S \rightarrow \R$, we have:
\[\argmax \,\gamma = \argmax\, (\gamma|S_1) \;\max{\gamma} \; \argmax\, (\gamma|S_2)\]
\end{lem}

We now have our local characterization of the operational semantics:
\begin{thm}\label{op-comp1} For well-typed effect values we have:
\begin{enumerate}
\item $\SOp(V) = \tuple{0,V}\;\, (= \eta_{\W}(V))$\label{partone-1}
\item $ \SOp(E_1 \,\myor\, E_2)  =   \SOp(E_1) \,\max{\pi_1}\, \SOp(E_2) \;\, ( = \SOp(E_1) \,\max{\ER{\W}{-}{0}}\, \SOp(E_2) )$\label{parttwo-1}.
\item $ \SOp(c\cdot E) = \give_{\W}(\sden{c},\SOp( E) )$\label{partthree-1}
\end{enumerate}
\end{thm}
\begin{proof}
\POPLomit{
For Part~\ref{partone-1} we calculate:
\[\begin{array}{lcl}\SOp(V) & = &\Out(\argmax\, s\type V.\,\Rew(s,V), V)\\
                                          & = &\Out(\ast, V)\\
                                          & = & \tuple{0,V}\\
           \end{array}\]
The first equality is as $\Op(V) = V$; the second is as values $V$ have only  one strategy,  $\ast$.

For part~\ref{parttwo-1}, we calculate:
\[\begin{array}{lcl}
\SOp(E_1 \,\myor\, E_2)  & =  &   \Out(\argmax\, s\type E_1 \,\myor\, E_2.\,\Rew(s,E_1 \,\myor\, E_2), E_1 \,\myor\, E_2)\\\\

  & = & \Out\left ( \left(
 \begin{array}{c} \argmax\, 1s\type E_1 \,\myor\, E_2.\,\Rew(1s,E_1 \,\myor\, E_2)\\
  \;\; \max{\Rew(-,E_1 \,\myor\, E_2)} \;\;\\
  \argmax\, 2s\type E_1 \,\myor\, E_2.\,\Rew(2s,E_1 \,\myor\, E_2)\\
   \end{array}\right), E_1 \,\myor\, E_2\right)\\
   && \hspace{110pt} (\mbox{by the first argmax lemma (Lemma~\ref{argmax1})})\\\\
 & = &  \Out\left (\left(\begin{array}{c} \argmax\, 1s\type E_1 \,\myor\, E_2.\,\Rew(s,E_1)\\
  \;\; \max{\pi_1\circ \Out(-,E_1 \,\myor\, E_2)} \;\;\\
  \argmax\, 2s\type E_1 \,\myor\, E_2.\,\Rew(s,E_2)\\
   \end{array}\right), E_1 \,\myor\, E_2\right)\\\\

    & = &
 \begin{array}{c} \Out(\argmax\, 1s\type E_1 \,\myor\, E_2.\,\Rew(s,E_1), E_1 \,\myor\, E_2)\\
  \;\; \max{\pi_1} \;\;\\
  \Out(\argmax\, 2s\type E_1 \,\myor\, E_2.\,\Rew(s,E_2), E_1\,\myor\, E_2)\\
   \end{array}\\
      && \hspace{130pt} (\mbox{by Lemma~\ref{trivial}})\\\\
   & = &  \begin{array}{c} \Out(\argmax\, s\type E_1.\,\Rew(s,E_1), E_1)\\
  \, \max{\pi_1} \,\\
  \Out(\argmax\, s\type E_2.\,\Rew(s,E_2),  E_2)\\
  \end{array}\\\\
   & = &  \SOp(E_1) \,\max{\pi_1}\, \SOp(E_2)
    \end{array}\]

And for part~\ref{partthree-1} we calculate:
\[\begin{array}{lcl}\SOp(c\cdot E)
  & = & \Out(\argmax\, s\type c\cdot E.\,\Rew(s, c\cdot E), c\cdot E)\\
  & = & \give_{\W}(\sden{c}, \Out(\argmax\, s\type c\cdot E.\, \sden{c} + \Rew(s, E), E))\\
  & = & \give_{\W}( \sden{c}, \Out(\argmax\, s\type c\cdot E.\,\Rew(s, E), E))\\
  & = & \give_{\W}( \sden{c}, \Out(\argmax\, s\type E.\,\Rew(s, E), E))\\
  & = & \give_{\W}( \sden{c},  \SOp(E))
               \end{array}\]
where
the third equality holds as the monoid preserves and reflects the ordering of $\R$.
\end{proof}}

Using this theorem we can show that substitutions of constants for constants can equivalently be done via $\W$. This will prove useful for our investigations of observational equivalence in Section~\ref{subsec:equiv1}.
\begin{lem}\label{sopsub1} Suppose $E\type b$ is an effect value, and that $f\type \Val_b \rightarrow \Val_{b'}$.
Let
$g\type u \subseteq \Con_b$ be the restriction of $f$ to a finite set that includes all the constants of type $b$ in $E$. Then:
\[\SOp(E[g]) = \W(f)(\SOp(E))\]
\end{lem}
\begin{proof} The proof is by structural induction. In case $E$ is a constant we have:
\[\W(f)(\SOp(c)) = \W(f)(\tuple{0,c}) = \tuple{0,f(c)} = \SOp(c[g])\]

In case $E$ has the form $E_1\,\myor\, E_2$ we have:
\[\begin{array}{llll} \W(f)(\SOp(E_1\,\myor\, E_2)) \\
\hspace{70pt} =  \;\W(f)(\SOp(E_1) \,\max{\pi_1}\, \SOp(E_2))  & (\mbox{by Theorem~\ref{op-comp1}.\ref{parttwo-1}})\\
 \hspace{70pt} =  \;  \W(f)(\SOp(E_1) \,\max{\pi_1\circ \W(f)}\, \SOp(E_2))& (\mbox{as $\pi_1\circ \W(f)  = \pi_1$})\\
\hspace{70pt} =  \;  \W(f)(\SOp(E_1)) \,\max{\pi_1}\, \W(f)(\SOp(E_2))) & (\mbox{by Lemma~\ref{trivial}})\\
\hspace{70pt} =  \;  \SOp(E_1[g]) \,\max{\pi_1}\, \SOp(E_2[g])\\
\hspace{70pt} =  \; SOp(E_1[g] \,\myor\, E_2[g]) & (\mbox{by Theorem~\ref{op-comp1}.\ref{parttwo-1}})\\
\hspace{70pt} =  \;  \SOp((E_1 \,\myor\, E_2)[g])\\
\end{array}\]

In case $E$ has the form $c\cdot E_1$ we have:
\begin{align*}
\W(f)(\SOp(c\cdot E_1)) & = \W(f)(\sden{c}\cdot \SOp(E_1))
& (\mbox{by Theorem~\ref{op-comp1}.\ref{partthree-1}})\\
                                                                               & = \sden{c}\cdot ( \W(f)(\SOp(E_1))&
 (\mbox{as $\W(f)$ is homomorphic})\\
                                                                               & = \sden{c}\cdot \SOp(E_1[g])\\
                                                                               & = \SOp(c \cdot (E_1[g]))&
(\mbox{by Theorem~\ref{op-comp1}.\ref{partthree-1}})\\
                                                                               & = \SOp((c \cdot E_1)[g]) \tag*{\qedhere}
\end{align*}
\end{proof}

\subsection{Denotational semantics}\label{den-sem1}

For the denotational semantics,  as discussed in Section~\ref{selmonad},  we need an auxiliary monad $\T$, here to handle the reward effect.
and we take $\T$ to be  $\W = \R \times \mbox{--}$, the  writer monad, and we have the $\W$-algebra $\alpha_\W\type \W(\R) \rightarrow \R$
where $\alpha_\W(\tuple{r,s}) = r+s$ as discussed in Section~\ref{subsec:firstrewards}.
We therefore have a strong monad
\[\S(X) =  (X \rightarrow \R) \rightarrow \R \times X\]
and use this monad to give the denotational semantics
\[\den{\Ssem_\T}{M}\type \den{\Ssem_\T}{\Gamma} \rightarrow \S(\den{\Ssem_\T}{\sigma}) \qquad (\mbox{for}\; \Gamma \vdash M\type \sigma )\]
  of our language, following the pattern explained in the previous section. (We often drop the subscript on $\Ssem_\T$ below.)

We assume available semantics of base types, constants, and function symbols, as discussed in Section~\ref{gen-densem} with, in particular:
$\sden{\reward} = \R$; $\sden{c}$ as in Section~\ref{subsec:firstrewards}, for $c\type\reward$; and $\sden{+}$ and $\sden{\leq}$ the monoid operation and ordering on $\R$. Recall that different constants of the same type are required to receive different denotations and that the consistency condition~\ref{opdencon} is required to be satisfied.

Turning to the algebraic operation symbols,  for $\myor$ we use the algebraic operation family $\myor_X$  given by Equation~\ref{ordef},
so:
\begin{equation} (\myor)_X(G_0,G_1)(\gamma) = G_0\gamma\;  \max{X,\ER{\W}{-}{\gamma}} \; G_1\gamma%
\label{or1def} \end{equation}
where, for any $\gamma\type X \rightarrow \R$, $\max{X,\gamma}\type X^2\rightarrow X$ is defined by:
\[x \,\max{X,\gamma}\, y = \left \{\begin{array}{ll} x & (\mbox{if\ }\gamma(x) \geq \gamma(y))\\
                                                                        y & (\mbox{otherwise})
                                       \end{array}\right.\]
For  $\give$ we take the algebraic operation family $(\give_\S)_X$ induced by the
$(\give_{\W})_X$, so, using Equation~\ref{finalgalgops}:
 \begin{equation}(\give_\S)_X(r,G)(\gamma) \;\; = \;\; (\give_{\W})_X(r, G\gamma)\;\; = \;\; \tuple{r + \pi_1(G\gamma), \pi_2(G\gamma)}\label{rew1def} \end{equation}

\subsection{Adequacy}\label{subsec:first-adeq}

We next aim to prove that the selection operational semantics  essentially coincides with its denotational semantics. This coincidence is our \emph{selection adequacy theorem}.

We need some notation to connect the operational semantics of programs with their denotations.
We set $( \sem_\W)_\sigma =  \W(\Ssem_p)\type  \W(\Val_\sigma) \rightarrow \W(\den{\Ssem}{\sigma})$.
 So for $u = \tuple{r,V}$ in $ \R \times \Val_\sigma =  \W(\Val_\sigma)$ we have
\[(\sden{\tuple{r,V}}_\W)_\sigma = \tuple{r,\den{\Ssem_p}{V}}\]

\begin{lem}\label{eff-sem1} For any effect value $E\type \sigma $ we have:
\[\den{\Ssem_\W}{E}(0) =  \sden{\SOp(E)}_\W\]
\end{lem}
\cutproof{\begin{proof}
We proceed by structural induction on $E$, and cases according to its form.
\begin{enumerate}
\item

Suppose $E$ is a value $V$.
Using Theorem~\ref{op-comp1}.\ref{partone-1}, we calculate:
\[\sden{\SOp(V)}_\W  \; = \; \sden{\tuple{0,V}}_\W
                                        \; = \;  \tuple{0,\den{\Ssem_p}{V}}
                                         \; = \;  \eta_{\W}(\den{\Ssem_p}{V})
                                        \; = \; \eta_{\S}(\den{\Ssem_p}{V})(0)
                                        \; = \; \den{\Ssem}{V}(0)
\]
%
\item

Suppose next that $E = E_1\; \myor \;E_2$. Then:
\[\begin{array}{llll} \hspace{-10pt}
 \sden{\SOp(E_1 \;\myor \;E_2)}_\W  \\
\hspace{30pt} = \;   \sden{\SOp(E_1) \, \max{\ER{\W}{-}{0}} \, \SOp(E_2)}_\W
                               & (\mbox{by Theorem~\ref{op-comp1}.\ref{parttwo-1}})\\
\hspace{30pt} = \;   \sden{\SOp(E_1) \, \max{\alpha_\W\circ\W(0)} \, \SOp(E_2)}_\W \\
\hspace{30pt} = \; \sden{\SOp(E_1) \, \max{\alpha_\W\circ\W(0)\circ\W(\sden{\;}_\sigma)} \, \SOp(E_2)}_\W\\
\hspace{30pt} = \; \sden{\SOp(E_1)}_\W \, \max{\alpha_\W\circ\W(0)} \, \sden{\SOp(E_2)}_\W
                              & (\mbox{using Lemma~\ref{trivial}}) \\
\hspace{30pt} = \; \den{\Ssem}{E_1}(0) \, \max{\alpha_\W\circ\W(0)} \, \den{\Ssem}{E_2}(0)
                              & (\mbox{by induction hypothesis}) \\
\hspace{30pt} = \; \myor_{\sden{\sigma}}(\den{\Ssem}{E_1},\den{\Ssem}{E_2})(0)
                              & (\mbox{by Equation~\ref{or1def}})\\
\hspace{30pt} = \; \den{\Ssem}{ E_1 \,\myor \,E_2}(0)\\
\end{array}\]

\item
Suppose instead that $E = c\cdot E'$.  Then:
\begin{align*}
\sden{\SOp(c\cdot E')}_\W
                                                & = \W(\Ssem_p)(\SOp(c\cdot E') )   \\                                                                     & = \W(\Ssem_p)(\sden{c}\cdot_{\W} \SOp(E'))
                                                    &   (\mbox{by Theorem~\ref{op-comp1}.\ref{partthree-1}})\\
                                                  & = \sden{c}\cdot_{\W} \W(\Ssem_p)(\SOp(E'))
                                                        &  (\mbox{as $ \W(\Ssem_p)$ is a homomorphism})\\
                                                  & = \sden{c}\cdot_{\W} \den{\Ssem}{E'}(0)
                                                        &   (\mbox{by induction hypothesis})\\
                                                 & = (\sden{c}\cdot_\S \den{\Ssem}{E'} )(0)
                                                      &                                   (\mbox{by Equation~\ref{rew1def}})\\
                                                 & = \den{\Ssem}{c\cdot E'}(0) \tag*{\qedhere}
\end{align*}
\end{enumerate}
\end{proof}}

\begin{thm}[Selection adequacy]\label{sel-ad1} For any program $M:\sigma$ we have:
\[\den{\Ssem_\W}{M}(0) =  \sden{\SOp(M)}_\W\]
\end{thm}
\begin{proof} We have:
\begin{align*}
\den{\Ssem_\W}{M}(0) & = \den{\Ssem_\W}{\Op(M)}(0)& (\mbox{by Theorem~\ref{basic-ad}})\\
             & = \sden{\SOp(\Op(M))}_\W & (\mbox{by Lemma~\ref{eff-sem1}})\\
                  & = \sden{\SOp(M)}_\W \tag*{\qedhere}
\end{align*}
\end{proof}

This theorem relates the compositional denotational semantics to the globally optimizing operational semantics. In particular, the latter determines the former at the zero-reward continuation. Whereas the denotational semantics optimizes only locally, as witnessed by the semantics of $\myor$, the latter optimizes over all possible Player strategies.
The use of the zero-reward continuation  is reasonable as the operational semantics of a program does not consider any continuation, and so,  as rewards mount up additively, the zero-reward continuation is appropriate at the top level.

In more detail, setting $\tuple{r,V} = \SOp(M)$, the theorem states that
$\den{\Ssem_\T}{M}(0) = \tuple{r,\den{\Ssem_p}{V}} $.
So the rewards according to both semantics agree, and the denotation of the value returned by the globally optimizing operational  semantics is given by the denotational semantics. In the case of base types  (or, more generally, products of base types)  the globally optimizing operational semantics is determined by the denotational semantics as the denotations of values of  base types determine the values (see Section~\ref{gen-ad}), and so, in that case, there is complete agreement between the operational semantics and the denotational semantics at the zero-reward continuation.

\subsection{Full abstraction, program equivalences, and purity}\label{subsec:equiv1}

 Given a notion of observations $\Ob(M)$ of programs $M\type b$ of a base type $b$, one  can define a notion of \emph{observational} or \emph{behavioural}  equivalence in a standard contextual manner; such notions are usually syntactical, being derived from operational semantics, though that is not necessary. Observational equivalence is generally robust against variations in the notion of observation, and we explore such variations in the context of our decision-making languages.

 So, for such a notion of observations $\Ob(M)$ of programs of base type $b$,   for programs $M,N\type\sigma$, define operational equivalence $M (\approx_{b,\Ob})_{\sigma} N$ between them  by:
\[M (\approx_{b,\Ob})_{\sigma} N \iff \forall C[\;\;]\type \sigma \rightarrow b.\, \Ob(C[M]) = \Ob(C[N]) \]
 (Here $C[\;\;]$ ranges over contexts with a single hole, defined in a standard way, and by  $C[\;\;]\type \sigma \rightarrow \tau$ we mean that for any $L\type\sigma$ we have $C[L]\type \tau$.)
We generally drop the type subscript $\sigma$ below.  Observational equivalence is an equivalence relation at any type, and it is closed under contexts, in the sense that for all programs  $M\type \sigma$, $N\type \sigma$ and contexts
$C[\;\;]\type \sigma \rightarrow \tau$ we have:
\[M \approx_{b,\Ob} N \implies C[M] \approx_{b,\Ob} C[N] \]

 Operational adequacy generally yields the implication:
\begin{equation} \models_\Msem M = N\type b   \implies \Ob(M) = \Ob(N)\label{adimp} \end{equation}
and it then follows that
\begin{equation} \models_\Msem M = N\type \sigma   \implies M \approx_{b,\Ob} N\label{moimp} \end{equation}
 As a particular case of this implication we have $M \approx_{b,\Ob} \Op(M)$ for programs $M\type \sigma$.
 The converse of the implication~\ref{moimp} is \emph{full abstraction} (of $\Msem$ with respect to $\approx_{b,\Ob}$) at type $\sigma$.

In the case of our language of choice and rewards, we work with observational equivalence at boolean type, and take the notion of observation to be simply the optimizing operational semantics $\SOp$, and write $\approx_{b}$ for $\approx_{b,\SOp}$, and $\approx$ for $\approx_\bool$. Note that the selection adequacy theorem (Theorem~\ref{sel-ad1}) immediately yields the implication~\ref{adimp} (and so also implication~\ref{moimp}) for $\Ssem_\W$ and $\SOp$, as expected, and we  then also have $M \approx_{b} \Op(M)$ for base types $b$ and programs $M\type \sigma$.

We next see  that, with this notion of observation, observational equivalence is robust against changes in choice of base type (Proposition~\ref{obvar1}). We investigate the robustness of observational equivalence against
 weakenings of the notion of observation later, observing either only values (Theorem~\ref{obweak1}) or only rewards  (Corollary~\ref{notrab}).

\begin{lem}\label{oevb} Suppose that $b$ is a base type with at least two constants. Then for any base type $b'$ and programs $M_1, M_2\type b'$ we have:
\[M_1 \approx_b  M_2 \implies \SOp(M_1)  = \SOp(M_2)\]
\end{lem}
\begin{proof} Let $E_i$ be $\Op(M_i)$ for $i=1,2$. Then $E_1 \approx_b  E_2$ (as $M_i \approx_{b} \Op(M_i)$ for $i=1,2$) and it suffices to prove that $\SOp(E_1)  = \SOp(E_2)$.
Suppose that  $\SOp(E_i) = \tuple{r_i,c_i}$ for $i = 1,2$. Let $f\type \Val_{b'} \rightarrow \Val_b$ be such that $f(c_1)$ and $f(c_2)$ are distinct, in case $c_1$ and $c_2$ are, and let $g$ be its restriction to the constants of the $E_i$ of type $b'$. For $i = 1,2$, we have:
\[\begin{array}{lcll} \SOp(\myF_gE_i) &=& \SOp(E_i[g])& (\mbox{by Lemma~\ref{consub}})\\
                                                     &=&  \W(f)(\SOp(E_i)) & (\mbox{by Lemma~\ref{sopsub1}})  \\
                                                     &=&  \tuple{r_i,f(c_i)}
\end{array}\]
As $E_1 \approx_b  E_2$, we have $\myF_gE_1 \approx_b  \myF_gE_2$ and so $\SOp(\myF_gE_1) = \SOp(\myF_gE_2)$ and so, from the above equations for the $\SOp(\myF_gE_i)$, that
$\tuple{r_1,f(c_1)} = \tuple{r_2,f(c_2)}$. So, as $f$ is 1--1 on $\{c_1,c_2\}$,
$\tuple{r_1, c_1} = \tuple{r_2,c_2}$ as required.
\end{proof}

As an immediate consequence of this lemma we have the following proposition that change of non-trivial base type does not affect observational equivalence:
\begin{prop}\label{obvar1} For all base types $b$ and programs $M,N\type \sigma$, we have
\[M \approx N \implies \; M \approx_b N\]
with the converse  holding if there are at least two constants of type $b$.
\end{prop}

\newcommand{\EVal}{\mathrm{EVal}}
\newcommand{\WVal}{\mathrm{WVal}}
\newcommand{\Obv}{\Ob_\mathrm{v}}
\newcommand{\Obr}{\Ob_\mathrm{r}}

Because the denotational semantics is compositional,
it facilitates proofs of program equivalences, including ones
that justify program transformations, and more broadly can be convenient for certain arguments about programs. For this purpose,
we rely on the equivalence relation
$\Gamma \vdash_{\Ax} M = N\type \sigma$
described in Section~\ref{gen-equations}. As remarked there, our general semantics
is equationally consistent.
We interest ourselves in a limited converse, with $\sigma$  a base type and $M$ and $N$  programs; we call this \emph{base type program completeness}.

Our system of axioms, $\Ax$,  is given in Figure~\ref{termequivs}.
As shown in Theorem~\ref{genax},  the choice operation is associative and idempotent;
from Corollary~\ref{auxeq} we have that the reward operation is an $\R$-action on the $\S(X)$ since it is on the $\W(X)$;
and we see from Theorem~\ref{distributes} that the reward operation commutes with the choice operation as the monoid addition preserves and reflects the order.
This justifies the first five of our axioms.
A  pointwise argument then shows that the following
equality holds for $r,s \in \R$ and $F,G \in \S(X)$, for any set $X$:
\begin{equation} r\!\cdot\! F \,\myor\, s\!\cdot\! F = t\!\cdot\! F \quad (t = \mathrm{max}(r,s))\label{r1} \end{equation}
Using this equality, the left-bias of the choice operation (shown in Theorem~\ref{genax}), and associativity, we have:
\begin{equation} (r\!\cdot\! F \,\myor\, G) \,\myor\, s\!\cdot\! F = r\!\cdot\! F \,\myor\, G  \quad (r \geq s)\label{r2} \end{equation}
and another pointwise argument establishes the equation:
\begin{equation} (r\!\cdot\! F \,\myor\, G) \,\myor\, s\!\cdot\! F =  G \,\myor\, s\!\cdot\! F \quad  (r < s)\label{r3}\end{equation}
These remarks justify our last two axioms.

\begin{figure}[h]
  \[(L \,\myor\, M) \,\myor\, N \;=\; L \,\myor\, (M \,\myor\, N) \qquad
M \,\myor\, M \;=\; M  \]
\[0\cdot N \;=\; N \qquad x \cdot (y \cdot N) \; =\;  (x + y)\cdot N\]
\[x \cdot (M\, \myor\, N) \; =\;  (x\cdot M)\, \myor\, (x \cdot N)\]
\[\myif x \geq y \mythen x\cdot M \myelse y\cdot M \;=\; x\cdot M\, \myor\, y\cdot M \]
\[\myif x \geq z \mythen (x\cdot M\, \myor\, N) \myelse ( N \,\myor\,z\cdot M)\,\;=\;\,
(x\cdot M\, \myor\, N) \,\myor\, z\cdot M \]
\caption{Equations for choices and rewards}\label{termequivs}
\end{figure}

Some useful consequences of these equations, mirroring the equalities~\ref{r1}--\ref{r3}, are:
\begin{equation}
  \tag{R$_1$}
c\cdot M \,\myor\, c' \cdot M \;\; =  \;\; c''\cdot M \quad (\mbox{where $\sden{c''}\; = \; \mathrm{max}(\sden{c},\sden{c'})$})%
  \label{R1}
\end{equation}
\begin{equation}
  \tag{R$_2$}
(c\cdot M \, \myor\,  N) \,\myor\,  c'\cdot M\; = \; c\cdot M \, \myor\, N \quad (\mbox{if $\sden{c} \geq \sden{c'}$})\%
  \label{R2}
\end{equation}
\begin{equation}
  \tag{R$_3$}
(c\cdot M \, \myor\, N) \,\myor\,  c'\cdot M\; = \;   N \,\myor\, c'\cdot M \quad (\mbox{if $\sden{c} < \sden{c'}$})\%
  \label{R3}
\end{equation}

Our equational system allows programs to be put
into a canonical
form. We say that  a  \emph{canonical form}  (ignoring bracketing of $\myor$) is an effect value of the form
\[(c_1\cdot V_1) \,\myor\, \ldots \,\myor\,   (c_n\cdot V_n)\]
with $n >0$ and no $V_i$ occurring twice.
\begin{lem}\label{norm} Every program $M$ 
is provably equal to a canonical 
form $\CF(M)$.
\end{lem}
\cutproof{\begin{proof} By the ordinary adequacy theorem (Theorem~\ref{basic-ad}), $M$ can be proved equal to an effect value $E$. Using the associativity equations, the fact that $\give$ and $\myor$ commute, and the $\R$-action equations, $E$ can be proved equal to a term of the form $c_1\cdot V_1 \,\myor\, \ldots \,\myor\,   V_n\cdot d_n$, possibly with some $V$'s occurring more than once. Such duplications can be removed using equations~\ref{R1},~\ref{R2}, and~\ref{R3} and associativity.
\end{proof}}

The next theorem shows that, for programs of base type, 
four equivalence relations coincide, and thereby simultaneously establishes for them: a normal form
for provable equality; completeness of our proof system for equations between such programs; and full abstraction.

\newcommand{\ov}[1]{\overline{#1}}
\begin{thm}\label{theorem:equivalences}
For any two programs $M$ and $N$ of base type $b$, the following equivalences hold:
\[\CF(M) = \CF(N) \iff \vdash_{\Ax} M=N\type b \iff \models_\Ssem M = N \type b \iff M \approx N\]
\end{thm}
\begin{proof} We already  know the implications from left-to-right hold.
So it suffices to show that:
\[\CF(M) \neq \CF(N) \implies M \not\approx N\]
First fix $l,r\type \Rew$ with $l < r$ (possible as $\R$ is expressively non-trivial).
We remark that, in general, to prove $A \not\approx B$ for $A,B\type \sigma$ it suffices to
to prove $A' \not\approx B'$ if we have $\vdash_{\Ax} A = A'\type \sigma$ and $\vdash_{\Ax} B = B'\type \sigma$.
We use this fact freely below. We also find it convenient to confuse sums of $\Rew$ constants with their denotations.

Let the canonical forms of $M$ and $N$ be
\[A = (c_1\cdot d_1) \,\myor\, \ldots \,\myor\, (c_n \cdot d_n)
\quad \mbox{and} \quad
B = (c'_1\cdot d'_1) \,\myor\, \ldots \,\myor\, (c'_{n'} \cdot d'_{n'})\]
and suppose they are different. It suffices to prove that $A \not\approx B$. Suppose, first, that for some $i_0$, $d_{i_0}$ is no $d'_j$.  Choose $c$ to be the maximum of the $c_i$ and the $c'_j$, other than $c_{i_0}$. Consider the context:
\[C_1[-] \eqdef \myif [-] =  d_{i_0} \mythen (c + r)\cdot \true \myelse  (c_{i_0} + l) \cdot \true\]
As $c_i + (c_{i_0} + l) < c_{i_0} + (c + r)$, we have $\pi_1(\SOp(C_1[A])) = c_{i_0} + c + r$. Further
$\pi_1(\SOp(C_1[B]))$ is  the maximum of the $c'_j + c_{i_0} + l$ and so $< c_{i_0} + c + r = \SOp(C_1[A]))$.
So we see that $A \not \approx B$ in this case.

Suppose, instead, that for some $i_0$, $d_{i_0}$ is $d'_{j_0}$ for some $j_0$ but that $c_{i_0} < c'_{j_0}$. Then we find that $\pi_1(\SOp(C_1[A])) = c_{i_0} + c + r$, as before, and that $\pi_1(\SOp(C_1[B]))$ is the maximum of the  $c'_j + (c_{i_0} + l)$, for $j \neq j_0$ and $c'_{j_0} + (c + r)$, which is  $c'_{j_0} + (c + r)$, and so we have again distinguished $A$ and $B$.

So, we may assume that for every $1\leq i \leq n$ there is a $1\leq j \leq n'$ such that $d_i = d'_j$ and $c_{i}\geq c'_{j}$. Arguing symmetrically, and recalling that none of the $d_i$ are repeated, and neither are any of the $d'_j$, we see that we may assume that $n = n'$ and that $(c_1\cdot d_1), \ldots ,(c_n \cdot d_n)$ and
$(c'_1\cdot d'_1) ,\ldots, (c'_n \cdot d'_n)$ are permutations of each other.

For the last case, suppose there is a first point $i_0$ at which $A$ and $B$ differ. We can then write them as:
\[A =  A_0\,\myor\, c_{i_0} \cdot d_{i_0} \,\myor\,  A_1 \,\myor\, c_{i_1} \cdot d_{i_1} \,\myor\, A_2\]
and
\[B =  A_0\,\myor\, c'_{i_0} \cdot d'_{i_0} \,\myor\,  B_1 \,\myor\, c'_{i_2} \cdot d'_{i_2} \,\myor\, B_2\]
with $d_{i_0} \neq d'_{i_0}$, $c_{i_1} \cdot d_{i_1} = c'_{i_0} \cdot d'_{i_0}$, and $c'_{i_2} \cdot d'_{i_2} = c_{i_0} \cdot d_{i_0}$, and where we allow any of $A_0,A_1,A_2, B_1$ or $B_2$ to be either a canonical form or the empty sequence, and, continuing to ignore parentheses,  interpret $A \,\myor\, B$ and
$B \,\myor\, A$ as $B$ when $A$ is empty and $B$ is not.

Let $c$ be the maximum of the $c_i$ and the $c'_i$, except for $c_{i_0}$ and $c'_{i_0}$, and consider the context
\[C_2[-] \eqdef \begin{array}{l} \mylet x\type b \mybe [ - ] \myin\\
                                               \myif x = d_{i_0} \mythen (c + c'_{i_0} + r) \cdot \true \myelse \\
                                               \myif x = d'_{i_0} \mythen (c + c_{i_0} + r) \cdot \false \myelse\\
                                                 \quad(c_{i_0} + c'_{i_0} + l)\cdot \false
                      \end{array}
\]
Then $C_2[A]$ is provably equal to
\[(\ov{c}_1\cdot \ov{d}_1) \,\myor\, \ldots \,\myor\, (\ov{c}_n \cdot \ov{d}_n)\]
where \[\begin{array}{lcll}\ov{c}_{i_0}\cdot \ov{d}_{i_0} &=& (c_{i_0} + c + c'_{i_0} + r)\cdot \true\\
            \ov{c}_{i_1}\cdot \ov{d}_{i_1} &=& (c'_{i_0} + c + c_{i_0} + r)\cdot \false\\
           \ov{c}_{i}\cdot \ov{d}_{i} & = & (c_{i} + c_{i_0} + c'_{i_0} + l)\cdot \false & (i \neq i_0,i_1)
            \end{array}\]
and  we see that $\SOp(C_2[A]) = \tuple{c_{i_0} + c + c'_{i_0} + r, \true}$.

Further, $C_2[B]$ is provably equal to
\[(\ov{c}'_1\cdot \ov{d}'_1) \,\myor\, \ldots \,\myor\, (\ov{c}'_n \cdot \ov{d}'_n)\]
where \[\begin{array}{lcll}\ov{c}'_{i_0}\cdot \ov{d}'_{i_0} &=& (c'_{i_0} + c + c_{i_0} + r)\cdot \false\\
            \ov{c}'_{i_2}\cdot \ov{d}'_{i_2} &=& (c_{i_0} + c + c'_{i_0} + r)\cdot \true\\
           \ov{c}'_{i}\cdot \ov{d}'_{i} & = & (c'_{i} + c_{i_0} + c'_{i_0} + l)\cdot \false & (i \neq i_0,i_2)
            \end{array}\]
and  we see that $\SOp(C_2[B]) = \tuple{c'_{i_0} + c + c_{i_0} + r, \false}$. So $C_2[-]$ distinguishes $A$ and $B$, concluding this final case.
\end{proof}

Theorem~\ref{theorem:equivalences} is in the spirit of~\cite{LS18} in giving axiomatic and denotational  accounts of observational equivalence at base types, though here at the level of terms rather than, as there, only effect values. (A natural axiomatic account of the observational equivalence of effect values at base types can be given by specializing the above axioms to them, including~\ref{R1},~\ref{R2}, and~\ref{R3}, but deleting the last two in Figure~\ref{termequivs}.)

Theorem~\ref{theorem:equivalences} holds a little more generally: for products of base types. The proof remains the same, using the fact that equality at any product of base types can be programmed using equality at base types. It follows that we have full abstraction at products of base types, i.e., for all programs of types of order 0.  A standard argument then shows that full abstraction holds for values of types of order 1; whether or not it holds for programs of types of order 1 is, however, open.

As a corollary of Theorem~\ref{theorem:equivalences} we have completeness for purity (i.e., effect-freeness) assertions at base types. Indeed we have it in a strong form:
\begin{cor}\label{purity1} For any program $M\type b$, we have:
\[\models_\Ssem M\downarrow b \implies \exists\, c\!\type \! b \vdash_\Ax M = c\]
\end{cor}
\begin{proof}
Suppose $\models_\Ssem M\!\downarrow \!b$. That is,  for some $x \in \sden{b}$,
$\den{\Ssem}{M} = \eta_{\S_\W}(x) = \lambda \gamma.\, \tuple{0,x}$.
For some $r\!\in \!\R$ and $c\!\type\! b$, $\SOp(M) = \tuple{r,c}$. So,  by adequacy we have:
\[\den{\Ssem}{M}(0) = \sden{\SOp(M)}_\W = \sden{\tuple{r,c}}_\W = \tuple{r,\sden{c}}\]
As $\den{\Ssem}{M} = \lambda \gamma.\, \tuple{0,x}$ we therefore have $r = 0$ and $\sden{c} = x$ and so $\models_\W M = c$.
It then follows from Theorem~\ref{theorem:equivalences} that $\vdash_\Ax M = c$.
\end{proof}

As may be expected, more generally we have strong purity completeness for products of base types, i.e., for any
$M\type \sigma$ where $\sigma$ is a product of base types we have:
\[\models_\Ssem M\downarrow \sigma \implies \exists\, V\!\type \! \sigma \vdash_\Ax M = V\]
and, indeed, this is a straightforward consequence of the corollary.

 \newcommand{\C}{\mathrm{C}}

 A natural question is whether, instead of using the selection monad $\S_\T$, we can treat  the choice operator at the same level as the reward one, say using a suitable free-algebra monad. This can be done, to some extent, by making use of Theorem~\ref{theorem:equivalences} and the equations we have established for these operations at the term level. Consider an equational system with a binary (infix) operation symbol  $- \,\myor\, -$ and an $\R$-indexed family of unary  operation symbols $r\cdot - \; (r \in \R)$, and impose Equations~\ref{reward-eqn}, associativity and commutativity equations:
 \[x \,\myor\, (y\,\myor\,z) = (x \,\myor\, y) \,\myor\,z \qquad r\cdot (x \,\myor\, y ) = r\cdot x \,\myor\, r\cdot y\]
and equations corresponding to Equations~\ref{R1},~\ref{R2}, and~\ref{R3}:
\[\begin{array}{cccl} r\cdot x \,\myor\, r' \cdot x   & = & \mathrm{max}( r, r')\cdot x\\[0.25em]
(r\cdot x \, \myor\, r'\cdot y) \,\myor\,  r''\cdot x & = & r\cdot x \, \myor\, r'\cdot y & (\mbox{if $ r  \geq  r''$})\\[0.25em]
(r\cdot x \, \myor\, r'\cdot y) \,\myor\,  r''\cdot x & = & r'\cdot y \,\myor\,  r''\cdot x & (\mbox{if $ r  <  r''$})\
\end{array}\]
Let $\C$ be the resulting free-algebra monad, and let  $\mathcal{C}$ be the  corresponding denotational semantics. One can  show that for all effect values $E,E'\type b$ of a base type $b$ we have:
\[\models_\mathcal{C} E = E' \;\; \iff \;\; \vdash_\Ax E = E'\type b\]
Using Theorems~\ref{basic-ad} and~\ref{theorem:equivalences} we then obtain a version of Theorem~\ref{theorem:equivalences} for $\mathcal{C}$, that, for any two programs $M$ and $N$ of base type $b$:
\[\vdash_\Ax M = N \type b \;\iff\; \models_\mathcal{C} M = N \type b\; \iff \; M \approx N \]

However we do not obtain an adequacy theorem analogous to the adequacy theorem (Theorem~\ref{sel-ad1}) which relates the operational semantics to the selection monad semantics at the zero-reward continuation. Consider, for example, the two boolean effect values $\true$ and $\true \,\myor\,  \false$. Operationally they both evaluate to $\true$. But they have different $\Ssem$-semantics as the second value is sensitive to the choice of reward continuation. They therefore have different $\mathcal{C}$-semantics, i.e., in this sense the $\mathcal{C}$-semantics is not sound.  An alternative would be to extend the operational semantics of programs to take a reward continuation into account, as done in~\cite{LS18}; however such an extension would be in tension with the idea that programs should be executable without additional information.

Turning to weakening the notion of observation, we may observe either just the reward or just the final value, giving two weakened notions of observation $\Obr = \pi_1\circ \SOp$, for the first, and $\Obv = \pi_2\circ \SOp$, for the second.
We begin by  investigating observing only values.
\begin{lem}\label{vequalss1} For programs $M,N\type \bool$ we have:
 \[M_1 \approx_{\Obv} M_2 \implies \SOp(M_1) = \SOp(M_2)\]
\end{lem}
\begin{proof} As $\Obv$ is weaker than $\SOp$ the implication~\ref{adimp} holds for $\Ssem_\W$ and it. We can therefore assume  without loss of generality that $M_1$ and $M_2$ are effect values, $E_1$ and $E_2$, say.
Suppose $\SOp(E_i) = \tuple{r_i,c_i}$ ($i = 1,2$).

Assume $E_1 \approx_{\Obv}E_2$. We then have $c_1 = c_2 = \true$, say. Suppose, for the sake of contradiction, that $r_1 \neq r_2$, and then,  without loss of generality,  that $r_1 < r_2$.  Define $f\type \Val_\bool \rightarrow \Val_\bool$ to be constantly
$\false$.
Then we have
\[\begin{array}{lcll}
\Obv(E_1 \,\myor\, \myF_fE_2) & = & \pi_2(\SOp(E_1 \,\myor\, \myF_fE_2))\\
                                                             & = & \pi_2(\SOp(E_1) \,\max{\pi_1}\, \SOp(\myF_fE_2))&
                                                                        (\mbox{by Theorem~\ref{op-comp1}.\ref{parttwo-1}})\\
                                                             & = & \pi_2(\tuple{r_1,\true} \,\max{\pi_1}\, \SOp(E_2[f]))&
                                                                        (\mbox{by Lemma~\ref{consub}})\\
                                                             & = & \pi_2(\tuple{r_1,\true} \,\max{\pi_1}\, \W(f)(\SOp(E_2)))&
                                                                        (\mbox{by Lemma~\ref{sopsub1}})\\
                                                             & = & \pi_2(\tuple{r_1,\true} \,\max{\pi_1}\, \tuple{r_2,\false})\\
                                                             & = & \false
\end{array}\]
and, similarly,
\[\begin{array}{lcll}
\Obv(E_2 \,\myor\,  \myF_fE_2) & = &  \pi_2(\tuple{r_2,\true} \,\max{\pi_1}\, \tuple{r_2,\false})\\
                                                             & = & \true
\end{array}\]
yielding the required contradiction, as $E_1 \approx_{\Obv}E_2$.
\end{proof}
It immediately follows that observing only values does not weaken the notion of observational equivalence.
\begin{thm}\label{obweak1} For programs $M,N\type \bool$ we have:
\[M \approx N \iff M \approx_{\Obv} N\]
\end{thm}

\newcommand{\Mr}{\M_{\mathrm{r}}}
\newcommand{\Mrsem}{\mathcal{M}_\mathrm{r}}

To investigate observing only rewards, we consider another free algebra monad, $\Mr$. It is the free algebra monad for the equational system with a binary (infix) associative, commutative, absorptive binary operation $- \,\myor\, -$ which forms a module relative to  the max-plus structure of $\R$, meaning that there is an $\R$-indexed family of unary  operation symbols $r\cdot - \; (r \in \R)$ forming an $\R$-action and with the following two equations holding:%
\[r\cdot (x \,\myor\, y ) = r\cdot x \,\myor\, r\cdot y \qquad r\cdot x \;\myor\; r' \cdot x    =  \mathrm{max}( r, r')\cdot x\]
We write $\Mrsem$ for the associated denotational semantics of our language with rewards.
\begin{lem}\label{rob}
For any programs $M_1,M_2\type b$ of base type, we have:
\[\Mrsem(M_1) = \Mrsem(M_2) \implies \Obr(M_1) =  \Obr(M_2)\]
\end{lem}
\begin{proof}
Assume $\Mrsem(M_1) = \Mrsem(M_2)$.  We can assume $M_1$ and $M_2$ are effect values, say $E_1$ and $E_2$.  These effect values take their denotations in the free algebra $\Mr(\sden{b})$. They can be considered as algebra terms  if we add  the constants $c\type b$ in $E_1$ and $E_2$ to the signature and
identify the constants $c\type \Rew$ occurring in subterms of the form $c\cdot E$ with their denotations. With that, their denotations are the same as their denotations in the free algebra extended so that the two denotations of the constants agree. So, as their denotations are equal, they can be proved equal in equational logic using closed instances of the axioms. We  show by induction on the size of proof that if $E = E'$ is so provable, then $\Obr(E) =  \Obr(E')$.

Other than commutativity, all closed instances $E = E'$ of the axioms hold in  $\Ssem$ and so $\SOp(E) = \SOp(E')$ for such instances. By Theorem~\ref{op-comp1}.\ref{parttwo-1}, for any effect values $E, E'\type b$ we have $\Obr(E \,\myor\, E') = \Obr(E)\, \max{} \, \Obr(E')$,
and so $\Obr(E \,\myor\, E') = \Obr(E \,\myor\, E')$ for all closed instances $E \,\myor\, E' = E \,\myor\, E'$ of commutativity. The only remaining non-trivial cases are the congruence rules. For that for choice we again use Theorem~\ref{op-comp1}.\ref{parttwo-1}; for that for rewards we use Theorem~\ref{op-comp1}.\ref{partthree-1}, which implies $\Obr(r\cdot E) = r + \Obr(E)$, for any effect value $E\type b$.
\end{proof}

\begin{thm}\label{vfullab} For any programs $M_1,M_2\type b$ of base type, we have:
\[\Mrsem(M_1) = \Mrsem(M_2) \iff M_1 \approx_{\Obr} M_2\]
\end{thm}
\begin{proof}
The implication from left to right follows immediately from Lemma~\ref{rob}. For the converse, suppose that
$M_1 \approx_{\Obr} M_2$. We can assume $M_1$ and $M_2$ are effect values, say $E_1$ and $E_2$.  Let $A_1 =  (c_1\cdot d_1) \,\myor\, \ldots \,\myor\, (c_n \cdot d_n)$ and
$A_2 = (c'_1\cdot d'_1) \,\myor\, \ldots \,\myor\, (c'_{n'} \cdot d'_{n'})$ be their normal forms.
 As the program equivalences used to put effect values of base type in normal form follow from those true in
 $\Mrsem$, we have $\Mrsem(E_i) = \Mrsem(A_i)\; (i = 1,2)$.
So, as $E_1 \approx_{\Obr} E_2$ we have $A_1 \approx_{\Obr} A_2$, using the implication from left to right.
The first part of the proof of Theorem~\ref{theorem:equivalences} that the two normal forms considered there are identical up to a permutation only uses the reward  part of the observation notion $\SOp$. So, reasoning as there, but now with $\Obr$ replacing $\SOp$, we see that
$(c_1\cdot d_1),\ldots , (c_n \cdot d_n)$ is a permutation of $(c'_1\cdot d'_1),\ldots, (c'_{n'} \cdot d'_{n'})$. As the commutativity program equivalence holds in $\Mrsem$, we therefore have
$\Mrsem(A_1) = \Mrsem(A_2)$ and so $\Mrsem(E_1) = \Mrsem(E_2)$, concluding the proof.
\end{proof}

\begin{cor}\label{notrab} The selection monad semantics augmented with auxiliary monad the writer monad $\W$ is not fully abstract at base types for $\approx_{\Obr}$ (and so $\approx_{\Obr}$ is strictly weaker than $\approx$). Indeed for programs $M,N\type \sigma$ of any type we have:
\[M \,\myor \, N \approx_{\Obr} N \,\myor \, M\]
\end{cor}
So, if we only care about optimizing rewards, we may even assume that $\myor$ is commutative.
\section{Adding probabilities}\label{second}

We next extend the language of choices and rewards by probabilistic nondeterminism. Thus, we have the three main ingredients
of MDPs, though in the setting of a higher-order $\lambda$-calculus rather than the more usual
state machines. We proceed as in
the previous section, often reusing
notation.

\subsection{Syntax}
For the syntax of our language,  in addition to the basic vocabulary and algebraic operations of the
language of  Section~\ref{subsec:first-syntax}, we assume available algebraic operation symbols $+_p\type \varepsilon,2\; (p \in [0,1])$ and function symbols
$\oplus_p \type \reward\,\reward\rightarrow \reward\; (p \in [0,1])$.
We use infix notation for both the $+_p$ and the $\oplus_p$. The former are for binary probabilistic choice. The latter are  for the convex combination of rewards; they prove useful for the equational logics given in Section~\ref{probeqpu}. (For example, see Equations~\ref{gather-eqn2-syntax} and~\ref{gather-eqn3-syntax}.) As before, we leave any other base type symbols, constants, or function symbols unspecified.

For example (continuing an example from Section~\ref{subsec:first-syntax}),
we may write the tiny program:
\[(5 \cdot \true) \hspace{0.5pt}\,\myor\,\hspace{0.5pt} ((5 \cdot \true) +_{.5}  (6 \cdot \false))\]
Intuitively,
like the program of Section~\ref{subsec:first-syntax},
this program could
return either $\true$ or $\false$, with respective rewards $5$ and $6$.
Both outcomes are possible on the right branch of its choice, each with probability $.5$.
The intended semantics aims to maximize expected rewards, so that branch is selected.

This example illustrates how the language can express MDP-like
transitions. In MDPs, at each time step, the decision-maker chooses an
action, and the process randomly moves to a new state and yields
rewards; the distribution over the new states depends on the current
state and the action. In our language, all decisions are binary,
but bigger decisions can be programmed
from them. Moreover, the decisions are separate from the probabilistic
choices and the rewards, but as in this example it is a simple matter of programming to combine them.
A more complete encoding of MDPs can be done by adding primitive recursion to the language, as suggested in the Introduction.

\subsection{Rewards and additional effects}
As in Section~\ref{subsec:firstrewards}  for both the operational and denotational semantics of our language we need a set of rewards $\R$ with appropriate structure and a monad employing it. To specify the structure we require on $\R$, we employ  the notion of a barycentric commutative monoid.
\emph{Barycentric algebras}  (also called convex algebras) are equipped with binary \emph{probabilistic choice} functions $+_p\type \R^2 \rightarrow \R \;\; (p \in [0,1])$ such that the following four equations hold:
\[\begin{array}{lcll} x +_1 y & = & x\\
x +_p x & = & x \\
x +_p y & =&  y +_{1- p} x \\
(x +_p y) +_q z & = & x +_{pq} (y +_{\frac{(1 - p)q} {1 - pq}}  z) & (p,q < 1)
\end{array}\]
 \emph{Barycentric commutative monoids} are barycentric algebras further equipped with a commutative monoid structure such that the monoid operation distributes over probabilistic choice, i.e., writing additively:
\[r + (s +_p s') = (r + s) +_p (r + s')\]

Barycentric algebras, introduced by Stone in~\cite{stone1949postulates}, provide a suitable algebraic structure for probability. They are equivalent to \emph{convex spaces} (also called convex algebras), which  are algebras equipped with 
operations $\sum_{i =1}^n p_ix_i$ (where the $p_i$ are in $[0,1]$, and $\sum_{i=1}^n p_i = 1$), subject to natural axioms
~\cite{pumplun1995convexity};  we  use the two notations interchangeably. Any mathematical expression built up using the operations of convex spaces from  mathematical expressions $e_i \; (n> 0, i=1,n)$ can be rewritten in the form $\sum_{i =1}^n p_ie_i$ using the axioms of convex spaces (and uniquely so if the $e_i$ do not involve the operations of convex spaces). For information on the extensive history of these concepts see~\cite{SW15,KeimelP16}.

Barycentric commutative monoids appear in the semantics of programming languages with probabilistic choice and
nondeterminism and  in categorical treatments of probability (for example, see~\cite{VW06,KeimelP16,DS21,Jac21,DPS18}).

 Turning to our assumptions on rewards, we assume  a set $\R$ of rewards is available, and that it is  equipped with:
 \begin{itemize}
 \item a barycentric commutative monoid structure, and
 \item  a total order with probabilistic choice and addition preserving and reflecting the order in their first argument (and so too in their second), in that, for all $r,s,t \in \R$:
\[r \leq s \iff r +_p t \leq s +_p t\quad (p \in (0,1))\]
and
\[r \leq s \iff r + t \leq s + t\]
 \end{itemize}
(Note the restriction on $p$ in the above condition on probabilistic choice.)
In the three examples of Section~\ref{subsec:firstrewards} (where the domain of R is the set of reals, nonnegative reals, or positive reals, respectively), probabilistic choice can be defined using the usual convex combination of real numbers:
$r +_p s = pr + (1-p)s$.
As in Section~\ref{subsec:firstrewards} we further assume that there is  an element $\sden{c}$ of $\R$ for each $c\type \Rew$ (with, in particular, $\sden{0} = 0$), and that  $\R$ is expressively non-trivial.

\newcommand{\DW}{\mathrm{DW}}

Our monad is the combination \[\DW(X) \eqdef \Pr(\R \times X)\] of the finite probability distribution monad with the writer monad for both  operational and denotational semantics. Our selection operational semantics, defined below,  evaluates programs $M$ of type $\sigma$ to finite distributions of pairs $\tuple{r,V}$, with $r \in \R$ and $V\type\sigma$, that is to elements of $\DW(\Val_\sigma)$.
 The monad is the free-algebra monad for \emph{barycentric $\R$-modules}. These are algebras with: an $\R$-indexed family of unary operations, written as $\give(r, -)$ or $r \cdot -$, forming an $\R$-action (Equation~\ref{reward-eqn}); and  a $[0,1]$-indexed family $- +_p -$ of binary operations forming a barycentric algebra over which the $\R$-action distributes, i.e., with the following equation holding:
\begin{equation} \label {rew-conv-eqn} r\cdot (x +_p y) = r\cdot x +_p  r\cdot y\end{equation}
The resulting monad has unit $(\eta_{\DW})_X (x) = \delta_{\tuple{0,x}}$;
the extension to $\DW(X)$ of a map $f\type X \rightarrow A$ to an algebra $A$ is given by
\[f^{\dagger_{\DW}}(\sum_{i = 1}^n p_i \tuple{r_i,x_i}) = \sum_{i = 1}^n p_i (r_i\cdot f(x_i))\]
(We used the Dirac distribution $\delta_z$ here; below, as is common, we just write $z$.)
With the assumptions made on $\R$, it forms a barycentric $\R$-module. Viewing $\R$ as a $\DW$-algebra, $\alpha_{\DW}\!\type\! \DW(\R) \rightarrow \R$, we have $\alpha_{\DW} =  (\id_\R)^{\dagger_{\DW}}$; explicitly:
\[\alpha_{\DW}(\sum_{i = 1}^n p_i \tuple{r_i,s_i}) = \sum_{i = 1}^n p_i (r_i + s_i)\]
The two $\DW$-algebraic operations  are:
\[(\give_{\DW})_X(r, \sum_{i = 1}^n p_i (r_i,x_i)) = \sum_{i = 1}^n p_i (r+ r_i,x_i)
\qquad
 ({+_p}_{\DW})_X(\mu,\nu) = p\mu + (1-p)\nu\]
 They are induced by the generic effects
 \[(g_{\DW})_{\give} \type \R \rightarrow \DW([1])\qquad
(g_{\DW})_{+} \type [0,1] \rightarrow \DW([2])\]
where  $(g_{\DW})_{\give}(r) = \tuple{r,\ast}$ and $(g_{\DW})_{+}(p) = p\tuple{0,0} + (1-p)\tuple{0,1}$.  We generally write $(\give_{\DW})_X$ using an infix operator $(\cdot_{\DW})_X$, as in Section~\ref{subsec:firstrewards}.

\subsection{Operational semantics}\label{sopsem2}

   For the ordinary operational semantics, as in Section~\ref{sopsem1} we assume available functions $\val_f$ for the function symbols $f$ of the basic vocabulary, as discussed in Section~\ref{gen-opsem}. For the selection operational semantics, we again take a game-theoretic point of view, with Player now playing a game against Nature, assumed to make probabilistic choices. Player therefore seeks to optimize their expected rewards. Effect values $E\type \sigma $ are regarded as games as before, but with one additional clause:
  \begin{itemize}
  \item[-] if $E = E_1 +_p E_2$, it is Nature's turn to move. Nature picks $E_1$ with probability $p$, and $E_2$ with probability $1-p$.
  \end{itemize}
To account for probabilistic choice we add a rule to the definition of strategies:
\[\frac{s_1\type E_1 \quad s_2 \type E_2}{(s_1,s_2)\type E_1 +_p E_2}\]
(Player will need a strategy for whichever move Nature chooses) and a case to the definition of the  total orders on strategies:
\begin{itemize}
\item Game is $E_1 +_p E_2$:
\[(s_1,s_2) \leq_{E_1 \,+_p\, E_2} (s'_1,s'_2) \iff  s_1 <_{E_1} s'_1 \quad  \vee \quad (s_1 = s'_1 \;\wedge\; s_2 \leq_{E_2} s'_2)\]
\end{itemize}

\noindent
For any effect value $E\type\sigma$, the outcome
$\Out(s,E)$
of  a strategy $s\type E$  is a finite probability distribution over $\R \times \Val_{\sigma}$, i.e., an element of $\Pr(\R \times \Val_{\sigma})$.
It is defined by:
\[\begin{array}{lcl}
\Out(\ast, V) & = & \tuple{0,V}\\
\Out(1s, E_1 \,\myor\, E_2) & = & \Out(s, E_1)\\
\Out(2s, E_1 \,\myor\, E_2) & = & \Out(s, E_2)\\
\Out(s, c\cdot E) & = & \sden{c} \cdot_{\Val_{\sigma}}\Out(s,E)\\
\Out((s_1,s_2), E_1 \,+_p\, E_2) & = & p \,\Out(s_1,E_1) \;+\; (1-p)\,\Out(s_2,E_2)\\
\end{array}\]

 The expected reward of a finite probability distribution on $\R \times X$, for a set $X$, is
\[ \Ex_X\left (\sum_{i = 1}^n p_i (r_i,x_i) \right ) \eqdef
\sum_{i = 1}^n p_ir_i\]
Note that $\Ex\type \Pr(\R \times X) \rightarrow \R$ can be written as $\ER{\DW}{-}{0}\;  (= \alpha_{\DW}\circ\DW(0))$, similarly to how  $\pi_1\type \R \times X \rightarrow \R$  could be in Section~\ref{sopsem1}.
 The expected reward of a strategy is:
\[ \Rew(s, E) \eqdef \Ex(\Out(s,E))\]

Our selection operational semantics, $\SOp(M)\in \R \times \Val_{\sigma}$ for $M:\sigma$, is defined as before by:
\[\SOp(M) = \Out(\argmax \, s\type \Op(M).\,\Rew(s,\Op(M)), \Op(M))\]
where we are now, as anticipated, maximizing  expected rewards.

We remark that, now that probabilistic choice is available, we could change our strategies to make a probabilistic choice for effect values $E_1\,\myor\, E_2$. However, as with Markov decision processes~\cite{Fel08}, that would make no change to the optimal  expected reward. It would, however, make a difference to the equational logic of choice if we chose with equal probability between effect values with equal expected reward: choice would then be commutative, but not associative.

Much as in Section~\ref{first}, we now develop a local characterization of the globally optimizing selection operational semantics.
We give this characterization in Theorem~\ref{op-comp2}, below; it is analogous to Theorem~\ref{op-comp1} in Section~\ref{first}. Some auxiliary lemmas are required. The first of them is
another $\argmax$ lemma  enabling us to deal with strategies for probabilistic choice.
\begin{lem} (Second argmax lemma)\label{argmax2} Let $P$ and $Q$ be finite  total orders,
let $P \times Q$ be given the lexicographic ordering, and suppose $\gamma\type P \times Q \rightarrow \R$. Define $g\type P \rightarrow Q$, $\un{u} \in P$ and $\un{v} \in Q$ by:
\[\begin{array}{lcl}
g(u) & = & \argmax\, v\type Q.\, \gamma(u,v)\\
\un{u}& = & \argmax\, u\type P.\, \gamma(u,g(u)) \\
\un{v}& = & g(\un{u})\\
\end{array}\]
Then:
\[(\un{u},\un{v}) = \argmax\, (u,v)\type P \times Q.\, \gamma(u,v)\]
\end{lem}
\cutproof{\begin{proof} Consider any pair $(u_0,v_0)$. By the definition of $g$ we have $g(u_0) \preceq v_0$ in the sense that:
\[\gamma(u_0, g(u_0)) > \gamma(u_0,v_0) \vee (\gamma(u_0,g(u_0)) = \gamma(u_0,v_0) \wedge g(u_0) \leq v_0)\]
and it follows that $(u_0,g(u_0)) \preceq_\gamma (u_0,v_0)$.

Next, by the definition of $\un{u}$ we have $\un{u} \preceq u_0$ in the sense that:
\[ \gamma(\un{u}, g(\un{u})) > \gamma(u_0, g(u_0)) \vee (\gamma(\un{u}, g(\un{u})) = \gamma(u_0, g(u_0)) \wedge \un{u} \leq u_0)\]
and it follows that $ (\un{u}, g(\un{u})) \preceq_\gamma (u_0,g(u_0))$. (The only non-obvious point may be that in the case where $\gamma(\un{u}, g(\un{u})) = \gamma(u_0, g(u_0))$, we have $\un{u} \leq u_0$, so either
$\un{u} < u_0$, when  $ (\un{u}, g(\un{u})) <_\gamma (u_0,g(u_0))$ or else
$\un{u} = u_0$, when  $ (\un{u}, g(\un{u})) = (u_0,g(u_0))$.)

So, as $\un{v} = g(\un{u})$, we have
\[ (\un{u}, \un{v}) = (\un{u}, g(\un{u})) \preceq_\gamma (u_0,g(u_0)) \preceq_\gamma  (u_0, v_0)\]
establishing the required minimality of $(\un{u}, \un{v})$.
\end{proof}}

The next lemma concerns expectations for probability distributions constructed by the reward and convex combination operations.
\begin{lem}\label{rewE}
We have:
\begin{enumerate}
\item $\Rew(s, c\cdot E)= \sden{c} + \Rew(s, E) $
\item  $\Rew(is_i, E_1 \,\myor\, E_2) = \Rew(s_i, E_i) \quad (i =1,2)$
\item $\Rew((s_1,s_2), E_1 \,+_p\, E_2)= p \Rew(s_1, E_1) + (1-p)\Rew(s_2, E_2) $
\end{enumerate}
\end{lem}
\begin{proof}
The second part is evident. For the other two, using the fact that $\Ex$ is a homomorphism, we calculate:
\[\begin{array}{lcl}\Rew(s, c\cdot E)& = & \Ex(\Out(s,c\cdot E))\\
& = & \Ex(\sden{c}\cdot\Out(s,\cdot E))\\
& = & \sden{c} + \Ex(\Out(s,\cdot E))\\
& = & \sden{c} + \Rew(s, E)\end{array} \]
and
\begin{align*}
\Rew((s_1,s_2), E_1 \,+_p\, E_2) & = \Ex(\Out((s_1,s_2), E_1 \,+_p\, E_2))\\
                                                                                 & = \Ex(p \,\Out(s_1,E_1) \;+\; (1-p)\,\Out(s_2,E_2))\\
                                                                    & = p\Ex(\Out(s_1,E_1)) +(1-p)\Ex(\Out(s_2,E_2))\\
                                                                    & = p\Rew(s_1,E_1) +(1-p)\Rew(s_2,E_2) \tag*{\qedhere}
\end{align*}
\end{proof}

\begin{thm}\label{op-comp2} The following hold for well-typed effect values:
\begin{enumerate}
\item $\SOp(V) = \tuple{0,V}\; (= \eta_{\DW}(V))$\label{partone-2}
\item
 $\SOp(E_1 \,\myor\, E_2)  =   \SOp(E_1) \,\max{\Ex}\, \SOp(E_2) \;\, (=  \SOp(E_1) \,\max{\ER{\DW}{-}{0}}\, \SOp(E_2))$\label{parttwo-2}

\item  $\SOp(c\cdot E) = \sden{c} \cdot \SOp(E)$\label{partthree-2}
\item  $\SOp(E_1 \,+_p\, E_2) =  p\SOp(E_1) + (1-p)\SOp(E_2)$\label{partfour-2}
\end{enumerate}
\end{thm}
\cutproof{\begin{proof}
\POPLomit{
\begin{enumerate}
\item The proof here is the same as the corresponding case of Theorem~\ref{op-comp1}.
\item The proof here is the same as that of the corresponding case of Theorem~\ref{op-comp1}, except that $\pi_1$ is replaced by $\Ex$.
\item
The proof here is again the same as that of the corresponding case of Theorem~\ref{op-comp1}, except that we use Lemma~\ref{rewE} to show that $\Rew(s, c\cdot E) =\sden{c} + \Rew(s, E)$.

\item
}
We just consider the fourth case. We have:
\[\begin{array}{lcl} \SOp(E_1 \,+_p\, E_2)  = \\ \qquad\Out(\argmax\, (s_1,s_2) \type E_1 \,+_p\, E_2.\,\Rew((s_1,s_2),E_1 \,+_p\, E_2), E_1 \,+_p\, E_2)\\
           \end{array}\]
So, following the second argmax lemma (Lemma~\ref{argmax2}), we first consider the function
\[\begin{array}{lclcl}
  f(s_1,s_2) & \eqdef & \Rew((s_1,s_2),E_1 \,+_p\, E_2) & = &  p\Rew(s_1,E_1) + (1-p) \Rew(s_2,E_2)
\end{array}\]
where the second equality holds by Lemma~\ref{rewE}.
We  then consider the function:
\[\begin{array}{lcl}
g(s_1) & \eqdef & \argmax\, s_2 \type E_2.\, f(s_1,s_2)\\
           & = & \argmax\, s_2 \type E_2.\, p\Rew(s_1,E_1) + (1-p) \Rew(s_2,E_2)\\
           & = & \argmax\, s_2 \type E_2.\,  \Rew(s_2,E_2)\\
\end{array}\]
where the second equality holds  as convex combinations are order-preserving and reflecting in their second argument.
Finally we consider
\[\begin{array}{lcl}
\un{s}_1 & \eqdef & \argmax\, s_1 \type E_1.\, f(s_1,g(s_1))\\
           & = & \argmax\, s_1 \type E_1.\, p\Rew(s_1,E_1) + (1-p)\Rew(g(s_1),E_2) \\
           & = & \argmax\, s_1 \type E_1.\,  \Rew(s_1,E_1)\\
\end{array}\]
where the last equality holds as convex combinations are order-preserving and reflecting in their first argument, and as $g(s_1)$ is independent of $s_1$.

So setting
\[\un{s}_2 = g(\un{s}_1) = \argmax\, s_2 \type E_2.\,  \Rew(s_2,E_2)\]\\[-1em]
by the second argmax lemma (Lemma~\ref{argmax2}) we have:
\[(\un{s}_1,\un{s}_2) = \argmax\, (s_1,s_2) \type E_1 \,+_p\, E_2.\,\Rew((s_1,s_2),E_1 \,+_p\, E_2)\]
so we finally have:
\begin{align*}
 \SOp(E_1 \,+_p\, E_2) & = \Out((\un{s}_1,\un{s}_2), E_1 \,+_p\, E_2)\\
                                             & = p\Out(\un{s}_1, E_1)  + (1-p)\Out(\un{s}_2, E_2)\\
                                             & = p\Out( \argmax\, s_1 \type E_1.\,  \Rew(s_1,E_1), E_1) \\
                                             & \; + (1-p)\Out( \argmax\, s_2 \type E_2.\,  \Rew(s_2,E_2), E_2)\\
                                             & = p \SOp(E_1) + (1-p) \SOp(E_2) \tag*{\qedhere}
\end{align*}
\POPLomit{
\end{enumerate}
}
\end{proof}}

\noindent
There is an analogous lemma to Lemma~\ref{sopsub1}, that  substitutions of constants for constants can equivalently be done via $\DW$.
\begin{lem}\label{sopsub2} Suppose $E\type b$ is an effect value, and that $f\type \Val_b \rightarrow \Val_{b'}$.
\begin{enumerate}
\item $\Ex_{\Val_{b'}} = \Ex_{\Val_{b}}\circ \DW(f)$\label{sopsub2part1}
\item Let
$g\type u \subseteq \Con_b$ be the restriction of $f$ to a finite set that includes all the constants of type $b$ in $E$. Then:
\[\SOp(E[g]) = \DW(f)(\SOp(E))\]\\[-2.5em]\label{sopsub2part2}
\end{enumerate}
\end{lem}
\begin{proof}
The first part is a straightforward calculation. The proof of the second part is by structural induction on $E$. The only non-trivial case is where  $E = E_1 \,\myor\, E_2$.
We calculate:
\begin{align*}
\SOp((E_1 \,\myor\, E_2)[g]) & = \SOp(E_1[g] \,\myor\, E_2[g])\\
                                              & = \SOp(E_1[g]) \,\max{\Ex_{\Val_{b'}}}\, \SOp(E_2[g])
                             \tag{By Theorem~\ref{op-comp2}.\ref{parttwo-2}} \\
                             & = \DW(f)(\SOp(E_1)) \,\max{\Ex_{\Val_{b'}}}\, \DW(f)(\SOp(E_2))
                             \tag{by induction hypothesis} \\
                             & = \DW(f)(\SOp(E_1) \,\max{\Ex_{\Val_{b'}}\circ\DW(f)}\, \SOp(E_2))\\
                             & = \DW(f)(\SOp(E_1) \,\max{\Ex_{\Val_{b}}}\, \SOp(E_2))
                             \tag{by Lemma~\ref{trivial}}\\
                             & = \DW(f)(\SOp(E_1\,\myor\,E_2))
                             \tag*{(By Theorem~\ref{op-comp2}.\ref{parttwo-2}) \qedhere}\\
\end{align*}
\end{proof}

\subsection{Denotational semantics}\label{den-sem2}

For the  denotational semantics we  consider three auxiliary monads $\T_1,\T_2$, and $\T_3$,  corresponding to three notions of observation
with varying degrees of correlation between possible program values and expected rewards.  Consider, for example, the effect value $1 \cdot \true +_{0.5} (2 \cdot \false +_{0.4} 3\cdot \true)$. With probability $0.5$ this returns $\true$ with reward $1$, with probability $0.2$ it returns $\false$ with reward $2$, and with probability $0.3$ it returns  $\true$ with reward $3$. This level of
detail is recorded using $\DW$ as our first monad. At a much coarser grain, we may simply record that $\true$ and $\false$ are  returned with respective probabilities $0.8$ and $0.2$, and that the overall expected reward is  $1.8$. This level of
detail is recorded using our third monad $\T_3$. At an intermediate level we may record the same outcome distribution and the expected reward given a particular outcome (in the example, the expected reward is $1.4$, given outcome $\true$,
and $2$, given outcome $\false$).
This level of detail is recorded using our second monad $\T_2$.

We  work with a general auxiliary monad, and then specialize our results to the $\T_i$. Specifically, we assume we have: a monad $\T$;
 $\T$-generic effects $(g_{\T})_{\give}: \R \rightarrow \T([1])$
 and $(g_{\T})_{+_p} \in \T([2])$,
with corresponding algebraic operations  \[(\give_\T)_X\type \R \times \T(X) \rightarrow \T(X) \qquad
 ({+_p}_{\T})_X\type \T(X)^2 \rightarrow \T(X)\]
 together with a $\T$-algebra $\alpha_{\T}\type\T(\R) \rightarrow \R$, such that, using  evident infix notations:
\begin{itemize}
\item[] (A1) for any set $X$, $(\give_\T)_X\type \R \times \T(X) \rightarrow \T(X)$
and $({+_p}_{\T})_X\type \T(X)^2 \rightarrow \T(X)$ form a barycentric $\R$-module, and
\item[]
(A2) the algebra map $\alpha_\T\type \T(\R) \rightarrow \R$ is a barycentric $\R$-module homomorphism, i.e., for $x,y \in \T(\R)$ we have:
\[\alpha_\T(r\cdot x) = r + \alpha_\T(x) \qquad \alpha_\T(x\, +_p\, y) = \alpha_\T(x)\, +_p\, \alpha_\T(y)\]
\end{itemize}
So we have the anticipated strong monad
\[\S(X) =  (X \rightarrow \R) \rightarrow \T(X)\]
We assume available semantics of base types, constants, and function symbols, as discussed for the language without probability in Section~\ref{den-sem1} with, additionally, the function symbols $\oplus_p$ denoting the corresponding convex combination operations on $\R$. As before, different constants of the same type are required to receive different denotations and the consistency condition~\ref{opdencon} is required to be satisfied.

As regards the algebraic operation symbols,  for $\myor$ we use the algebraic operation $\myor$ given by Equation~\ref{ordef}, so
\begin{equation} (\myor)_X(G_0,G_1)(\gamma) = G_0\gamma\;  \max{X,\ER{\T}{-}{\gamma}} \; G_1\gamma%
\label{or2def}\end{equation}
For  $\give$ and $+_p$
we take the algebraic operations $(\give_\S)_X$ and $({+_p}_\S)_X$ induced by the
$(\give_{\T})_X$ and $({+_p}_{\T})_X$, so:
\[  (\give_\S)_X(r,G)(\gamma)  =  (\give_{\T})_X(r, G\gamma)\]
and
 \[ ({+_p}_\S)_X(F,G)(\gamma)  =  ({+_p}_{\T})_X(F(\gamma), G(\gamma))
\]

As mentioned above, our first monad is $\T_1 \eqdef \DW$. With its associated generics for reward and probabilistic choice and $\T$-algebra it evidently satisfies the two assumptions (A1) and (A2).

 Writing $\supp(\nu)$ for the support of a probability distribution $\nu$, our second monad  is
 \[\T_2(X)  =  \{\tuple{\mu, \rho} \mid \mu \in \Pr(X), \rho\type \supp(\mu) \rightarrow \R\}\]
It  is the free-algebra monad for algebras with an $\R$-indexed family of unary operations, written as $\give(r, -)$ or $r \cdot -$, and a $[0,1]$-indexed family $- +_p -$ of binary operations satisfying the equations for $\DW$ together with the equation:
\begin{equation} \label {gather-eqn2} r\cdot x \,+_p \, s\cdot x = (r +_p s) \cdot x \end{equation}
The two $\T_2(X)$-algebraic operations are:
\[(\give_{\T_2})_X(r, \tuple{\mu, \rho}) = \tuple{\mu, x \mapsto r + \rho(x) }\]
and
\[ ({+_p}_{\T_2})_X(\tuple{\mu, \rho },\tuple{\mu',\rho'}) = \tuple{p\mu + (1-p)\nu, \rho''}\]
 where:
 \[\rho''(x) = \left \{\begin{array}{ll}
                                 \rho(x) +_q \rho'(x) & (x \in \supp(\mu) \cap \supp(\mu'))\\
                                 \rho(x) & (x \in \supp(\mu) \backslash \supp(\mu') )\\
                                 \rho'(x) & (x \in \supp(\mu') \backslash \supp(\mu) )\\
                            \end{array}
                   \right.\]
 where $q = p\mu(x)/(\mu(x) +_p \mu'(x))$. One can then show that:
 \begin{equation} \sum_{i = 1}^np_i \tuple{\mu_i,\rho_i} = \tuple{\mu,\rho}\label{bigavT2} \end{equation}
 where $\mu = \sum_{i = 1}^np_i\mu_i$ and, for $x \in \supp(\mu)$:
 \[\rho(x) = \sum_{\supp(\mu_i) \,\ni\, x} \frac{p_i\mu_i(x)}{\mu(x)}\rho_i(x)\]

The resulting monad has unit $(\eta_{\T_2})_X (x) = \tuple{x,x \mapsto 0\}}$;
the extension to $\T_2(X)$ of a map $f\type X \rightarrow A$ to an algebra $A$ is given by
\[f^{\dagger_{\T_2}}(\mu, \rho) = \sum_{x \in \supp(\mu)} \mu(x)(\rho(x)\cdot f(x))\]
 and  for any $f\type X \rightarrow Y$ we have:
\[\T_2(f)(\mu, \rho) = \tuple{\Pr(\mu), \rho'}\]
where
\[\rho'(y) = \frac{\sum_{f(x) = y} \mu(x) \rho(x)}{\Pr(\mu)(y)}\quad (y \in \supp(\Pr(\mu)))\]
Equation~\ref{gather-eqn2} holds for $\R$, using commutativity and homogeneity, so we can take the algebra map $\alpha_{\T_2}$  to be $\id_\R^{\dagger_{\T_2}}$; explicitly we find:
\[\alpha_{\T_2} \left ( \sum_{i = 1}^n p_i r_i, \rho \right) = \sum_{i = 1}^n p_i (\rho(r_i) + r_i)\]

 Our third monad $\T_3(X)  =  \Pr(X) \times \R$ is the free-algebra monad for algebras with an $\R$-indexed family of unary operations, written as $\give(r, -)$ or $r \cdot -$, and a $[0,1]$-indexed family $- +_p -$ of binary operations satisfying the equations for $\DW$ and the equation:
\begin{equation} \label {gather-eqn3} r\cdot x \,+_p \, s\cdot y = (r +_p s) \cdot x \,+_p\,  (r +_p s) \cdot y \end{equation}
The two $\T_3(X)$-algebraic operations are:
\[(\give_{\T_3})_X(r, \tuple{\mu, s}) = \tuple{\mu, r +s}\]
and
\[ ({+_p}_{\T_3})_X(\tuple{\mu,r},\tuple{\nu,s}) = \tuple{p\mu + (1-p)\nu, pr + (1-p)s}\]
One can then show that:
\begin{equation} \sum_{i = 1}^n p_i \tuple{\mu_i,r_i} = \tuple{\sum_{i = 1}^n p_i, \sum_{i = 1}^n\mu_i,r_i}%
\label{bigavT3}\end{equation}

The resulting monad has unit \[(\eta_{\T_3})_X (x) = \tuple{x,0}\]
the extension to $\T_3(X)$ of a map $f\type X \rightarrow A$ to an algebra $A$ is given by
\[f^{\dagger_{\T_3}}(\sum p_i x_i, r) = r\cdot \sum p_i f(x_i)\]
and  \[\T_3(f)(\tuple{ \sum p_i x_i, r}) =  \tuple{ \sum p_i f(x_i), r}\]
Unfortunately Equation~\ref{gather-eqn3} need not hold for $\R$ with the assumptions made on it so far; indeed, while it does hold for the two examples with the reals and addition, it does not hold for the example of the positive reals and multiplication.
When dealing with $\T_3$ we therefore assume additionally that $\R$ satisfies Equation~\ref{gather-eqn3}, and
so we can take $\alpha_{\T_3}$ to be $\id_\R^{\dagger_{\T_3}}$; explicitly we find:
\[\alpha_{\T_3} \left ( \sum_{i = 1}^n p_i r_i, r \right) = r + \left (\sum_{i = 1}^n p_i r_i \right )\]

Define comparison maps:
\[ (\theta_{\DW,\T})_X\type \DW(X) \rightarrow \T(X) \; \eqdef \; (\eta_{\T})_X^{\dagger_{\DW}}\]
These functions are useful when discussing adequacy and full abstraction.
Explicitly we have:
\begin{equation}\label{thetafor} (\theta_{\DW,\T})_X(\sum_{i = 1}^n p_i\tuple{r_i,x_i}) =
\sum_{i = 1}^n p_i (r_i(\cdot_\T)_{\T(X)} \eta_\T(x_i))\end{equation}
\begin{lem}  $\theta_{\DW,\T}$ is a monad morphism.\label{thetamor}
\end{lem}
\begin{proof} We have to show that $\theta$ is natural and preserves the unit and multiplication maps. We make use of Equation~\ref{thetafor} throughout the proof.

For naturality we need to show that for $f\type X \rightarrow Y$ we have
$\T(f)\circ \theta_X = \theta_Y \circ \DW(f)$. Choosing
$\sum_{i = 1}^n p_i\tuple{r_i,x_i} \in \DW(X)$ we have:
\[\begin{array}{lcl}\T(f)(\theta_X(\sum_{i = 1}^n p_i\tuple{r_i,x_i}))
   & = & \T(f)(\sum_{i = 1}^n p_i (r_i\cdot_\T \eta_\T(x_i))) \\
      & = & \sum_{i = 1}^n p_i (r_i\cdot_\T \T(f)(\eta_\T(x_i)))\\
         & = & \sum_{i = 1}^n p_i (r_i\cdot_\T \eta_\T(f(x_i)))
    \end{array}  \]
    using the fact that maps of the form $\T(f)$ act homomorphically on algebraic operations, and:
    \[\begin{array}{lcl} \theta_Y(\DW(f)(\sum_{i = 1}^n p_i\tuple{r_i,x_i}))
    & = & \theta_Y(\sum_{i = 1}^n p_i\tuple{r_i,f(x_i)})\\
             & = & \sum_{i = 1}^n p_i (r_i\cdot_\T \eta_\T(f(x_i)))
    \end{array}\]

    For preservation of the unit we have to show that $(\eta_T)_X = \theta_X\circ \eta_\DW$. This is immediate from the definition of $\theta$.

    For preservation of multiplication we have to show that
    \[\theta_X\circ (\mu_\DW)_X = (\mu_\T)_X\circ \theta_{\T(X)}\circ \DW(\theta_X)\]
    To this end, choose $\sum_{i = 1}^n p_i\tuple{r_i,u_i} \in \DW(\DW(X))$ where
    $u_i = \sum_{j = 1}^{m_i} q_{ij}\tuple{s_{i,j},x_{i,j}}$, for $i = 1,\ldots,n$.
    Then, using the fact that monad multiplications act homomorphically on algebraic operations,  we have:
    \[\begin{array}{lcl}
    \theta_X((\mu_\DW)_X( \sum_{i = 1}^n p_i\tuple{r_i,u_i}))
    & = &  \theta_X(\sum_{i = 1}^n p_i \sum_{j = 1}^{m_i} q_{ij} \tuple{r_i + s_{i,j},x_{i,j}})\\
    & = &
      \sum_{i = 1}^n p_i \sum_{j = 1}^{m_i} q_{ij} (r_i + s_{i,j})\cdot_\T \eta_\T(x_{i,j})
       \end{array}\]
    and:
    \begin{align*}
&        (\mu_\T)_X(\theta_{\T(X)}( \DW(\theta_X) ( \textstyle\sum\nolimits_{i = 1}^n p_i\tuple{r_i,u_i}))) \\
&        \hspace{100  pt} =         (\mu_\T)_X(\theta_{\T(X)}  ( \textstyle\sum\nolimits_{i = 1}^n p_i\tuple{r_i,\theta_X(u_i)}))\\
& \hspace{100  pt} =  (\mu_\T)_X  ( \textstyle\sum\nolimits_{i = 1}^n p_i (r_i (\cdot_T)_{\T(\T(X))} (\eta_T)_{\T(X)} \theta_X(u_i)))\\
& \hspace{100  pt} =     \textstyle\sum\nolimits_{i = 1}^n p_i (r_i (\cdot_T)_{\T(X)} (\mu_\T)_X((\eta_T)_{\T(X)} \theta_X(u_i)))\\
& \hspace{100  pt} =     \textstyle\sum\nolimits_{i = 1}^n p_i (r_i (\cdot_T)_{\T(X)}  \theta_X(u_i))\\
& \hspace{100  pt} =     \textstyle\sum\nolimits_{i = 1}^n p_i (r_i (\cdot_T)_{\T(X)}  \theta_X(\textstyle\sum\nolimits_{j = 1}^{m_i} q_{ij}\tuple{s_{i,j},x_{i,j}}))\\
& \hspace{100  pt} =     \textstyle\sum\nolimits_{i = 1}^n p_i (r_i \cdot_T  \textstyle\sum\nolimits_{j = 1}^{m_i} q_{ij}(s_{i,j} \cdot_\T\eta_\T(x_{i,j})))\\
& \hspace{100  pt} =     \textstyle\sum\nolimits_{i = 1}^n p_i (\textstyle\sum\nolimits_{j = 1}^{m_i} q_{ij}(s_{i,j}\cdot r_i \cdot_T   \cdot_\T\eta_\T(x_{i,j})))\\
& \hspace{100  pt} =     \textstyle\sum\nolimits_{i = 1}^n p_i (\textstyle\sum\nolimits_{j = 1}^{m_i} q_{ij}((s_{i,j} + r_i) \cdot_T   \cdot_\T\eta_\T(x_{i,j}))) \tag*{\qedhere}
\end{align*}
\end{proof}

In the case of $\T_1$, $\theta_{\T_1}$ is the identity.  In the case of $\T_2$,
first, given a distribution $\mu = \sum_{i = 1}^n p_i\tuple{r_i,x_i} \in \T_1(X)$, define its \emph{value distribution} $\VDis(\mu)$ in
$\Pr(X)$,  and  its \emph{value support} $\vsupp(\mu) \subseteq X$  by:
 \[\VDis(\mu) = \sum_{i = 1}^n p_i x_i \qquad  \vsupp(\mu) = \{x_i\}\]
and then define the \emph{conditional expected reward} of $\mu$ given $x \in \vsupp(\mu)$ by:
\[\Rew(\mu | x) =  \frac{\sum_{x_i = x} p_ir_i}{\sum_{x_i = x} p_i} \]
We then have:
\[(\theta_{\T_2})_X (\mu) = \tuple{\VDis(\mu), \Rew(\mu | - ) }\]
as, using Equation~\ref{bigavT2}, we can calculate:
\[\begin{array}{lcl}\theta_{\DW,\T_2}(\sum_{i = 1}^n p_i\tuple{r_i,x_i})
& = & \sum_{i = 1}^n p_i (r_i\cdot_{\T_2} \eta_{\T_2}(x_i))\\
& = & \sum_{i=1}^n p_i\tuple{x_i, x_i \mapsto r_i}\\
& = & \tuple{\mu,\rho}\\
\end{array}\]
 where $\mu = \sum_{i = 1}^np_ix_i$ and, for $x \in \supp(\mu)$:
 \[\rho(x) = \sum_{x_i = x} \frac{p_i}{\mu(x)}r_i\]

In the case of $\T_3$ we have:
\[(\theta_{\T_3})_X(\mu) = \tuple{ \VDis(\mu),  \Ex(\mu)}\]
as, using Equation~\ref{bigavT3},  we can calculate:
\[
 \sum_{i = 1}^n p_i (r_i\cdot_{\T_3} \eta_{\T_3}(x_i))
 \; = \;  \sum_{i = 1}^n p_i \tuple{x_i,r}\\
 \; =  \;\tuple{\sum_{i = 1}^n p_i, \sum_{i = 1}^n x_i,r} \\
\]

Two properties of the $\T_i$ are useful when we consider full abstraction below.
For the first property, say that $\B$ is \emph{characteristic} for $\T$ if, for any set $X$ and any two $u,v \in \T(X)$ we have:
\[u \neq v \implies \exists f\type X \rightarrow \B.\, \T(f)(u) \neq \T(f)(v)\]

\begin{lem}\label{characteristic} $\B$ is characteristic for each of the $\T_i$.
\end{lem}
\cutproof{\begin{proof}
Fix $X$, and, for any $x \in X$ let $f_x \type X \rightarrow \B$ be the map that sends $x$ to $0$ and everything else in $X$ to $1$.

For the case of $\T_1$ suppose  we have distinct elements of  $\T_1(X)$, viz.\
$\mu = \sum_{i = 1}^m p_i\tuple{r_i,x_i}$  and
$\nu = \sum_{j = 1}^n q_j\tuple{s_j,y_j}$. Then there is an $\tuple{r_i,x_i}$ in the support of (say) $\mu$ that is either not in the support of $\nu$ or has different probability there.
 Then $\tuple{r_i,0}$ is in the support of $\T_1(f_{x_i})(\mu)$ but is either not in the support of $\T_1(f_{x_i})(\nu)$ or has different probability there.

In the case of $\T_2$ suppose we have distinct elements  of $\T_2(X)$, viz.\
$a = \tuple{\sum_{i = 1}^m p_i x_i, \rho_0}$  and  $b = \tuple{\sum_{j = 1}^n q_jy_j, \rho_1}$. If $\sum_{i = 1}^m p_i x_i$ and $\sum_{j = 1}^n q_j y_j$ are distinct we proceed as in the case of $\T_1$. Otherwise there is an $x_i$, say $x_1$, such that $\rho_0(x_1) \neq \rho_1(x_1)$. Let $\rho'_0$ and $\rho'_1$ be the second components of
 $\T_2(f_{x_1})(a)$ and  $\T_2(f_{x_1})(b)$. Then $\rho'_0(0) =   \rho_0(x_i)$ and $\rho'_1(0) =    \rho_1(x_i)$ and these are different.

In the case of $\T_3$ suppose we have distinct elements of $\T_3(X)$, viz.\
$a = \tuple{\sum_{i = 1}^m p_i x_i, r}$  and  $b = \tuple{\sum_{j = 1}^n q_jy_j, s}$. If $\sum_{i = 1}^m p_i x_i$ and $\sum_{j = 1}^n q_j y_j$ are distinct we proceed as in the case of $\T_1$. Otherwise, $r \neq s$, and $\T(f)$  distinguishes $a$ and $b$.
\end{proof}}

 For the second property, for any $X$ and $\gamma \type X \rightarrow \R$, define the \emph{reward addition} function
\[\myk_\gamma \type \T(X) \rightarrow \T(X)\]
to be $f^{\dagger_{\T}}$, where $f(x) \eqdef \gamma(x)\cdot_\T \eta_\T(x)$. Then we say that \emph{reward addition is injective for $\T$} if such functions are always injective.
 \begin{lem}\label{addition} Reward addition is injective for each of the $\T_i$.
\end{lem}
\cutproof{\begin{proof} Fix $X$ and $\gamma \type X \rightarrow \R$. Beginning with $\T_1$ for any
$\mu = \sum_{i = 1}^m p_i\tuple{r_i,x_i}$, with no $\tuple{r_i,x_i}$ repeated, we have
\[\myk_\gamma(\mu) = \sum_{i=1}^m p_i\tuple{\gamma(x_i) + r_i,x_i}\]
 with no
repeated $\tuple{\gamma(x_i) + r_i,x_i}$ (since the monoid addition on $\R$ reflects the order). So, for any such
$\mu = \sum_{i=1}^m p_i \tuple{r_i,x_i}$ and $\nu = \sum_{j=1}^n q_j\tuple{s_j,y_j}$,
if we have $\myk_\gamma(\mu) = \myk_\gamma(\nu)$, i.e., if  we have $\sum_{i=1}^m p_i \tuple{r_i,x_i} = \sum_{j=1}^n
q_j\tuple{s_j,y_j}$, then $m = n$ and, for some permutation $\pi$ of the indices, we have
$p_i\tuple{\gamma(x_i) + r_i,x_i} = q_{\pi(j)}\tuple{\gamma(y_{\pi(j)}) + s_{\pi(j)},y_{\pi(j)}}$. So then $x_i =  y_{\pi(j)}$, and $\gamma(x_i) = \gamma(y_{\pi(j)})$ follows, and we see that
$p_i\tuple{r_i,x_i} = q_{\pi(j)}\tuple{s_{\pi(j)},y_{\pi(j)}}$. So $\mu = \nu$, as required.

The proofs for $\T_2$ and $\T_3$ are similar, using the respective formulas
\[\myk_\gamma(\tuple{\sum_{i = 1}^mp_i x_i,\rho}) = \tuple{\sum_{i = 1}^mp_i x_i, x_i \mapsto \rho(x_i) + \gamma(x_i) }\]
and
\[\myk_\gamma(\tuple{\sum_{i = 1}^mp_i x_i,r}) = \tuple{\sum_{i = 1}^mp_i x_i, r +  \sum_i p_i \gamma(x_i)}\qedhere \]
\end{proof}}

 \subsection{Adequacy}

As in Section~\ref{subsec:first-adeq}, we aim to prove a selection adequacy theorem connecting the globally defined selection operational semantics with the denotational semantics. We again need some notation.
Using assumption (A1) of Section~\ref{den-sem2}, we set
 \[( \sem_\T)_\sigma =  ((\eta_\T)_{\den{\Ssem}{\sigma}}\circ  \Ssem_p)^{\dagger_{\DW}}\type  \DW(\Val_\sigma) \rightarrow \T(\den{\Ssem}{\sigma})\]
  So, for
$\mu = \sum_{i = 1}^n p_i\tuple{r_i,V_i} \in 
\DW( \Val_{\sigma})$ we have:
\[\sden{\mu}_\T =  \sum_{i = 1}^n p_i (r_i \cdot_\T \eta_\T(\den{\Ssem_p}{V_i}))  \]

\begin{lem}\label{auxaux2}
For any $\mu \in \DW(\Val_\sigma)$ we have:
$\alpha_{\T}(\T(0_{\sden{\sigma}})(\sden{\mu}_\T)) = \alpha_{\DW}(\DW(0_{\Val_\sigma})(\mu))$
\end{lem}
\begin{proof}
Suppose $\mu =  \sum_{i = 1}^n p_i\tuple{r_i,V_i}$. We calculate:
\begin{align*}
& \alpha_{\T}(\T(0)(\sden{ \textstyle\sum\nolimits_{i = 1}^n p_i\tuple{r_i,V_i}}_\T))  \\
& \hspace{50pt} = \; \alpha_{\T}(\T(0)(\textstyle\sum\nolimits_{i = 1}^n p_i (r_i \cdot_\T \eta_\T(\den{\Ssem_p}{V_i})))) \\
& \hspace{50pt} = \; \alpha_{\T}(\textstyle\sum\nolimits_{i = 1}^n p_i (r_i \cdot_\T \T(0)(\eta_\T(\den{\Ssem_p}{V_i})))) \\
& \hspace{50pt} = \; \alpha_{\T}(\textstyle\sum\nolimits_{i = 1}^n p_i (r_i \cdot_\T \eta_\T(0))) \\
& \hspace{50pt} = \; \textstyle\sum\nolimits_{i = 1}^n p_i (r_i +  \alpha_{\T}(\eta_\T(0)))  \tag{by assumption (A2) of Section~\ref{den-sem2}}\\
& \hspace{50pt} = \; \textstyle\sum\nolimits_{i = 1}^n p_i r_i\\
& \hspace{50pt} = \; \alpha_{\DW}(\DW(0)( \textstyle\sum\nolimits_{i = 1}^n p_i\tuple{r_i,V_i})) \tag*{\qedhere}
\end{align*}
\end{proof}

\begin{lem}\label{eff-sem2} For any effect value $E\type \sigma $ we have:
$\den{\Ssem}{E}(0) =  \sden{\SOp(E)}_\T$
\end{lem}
\begin{proof}
We proceed by structural induction on $E$:
\begin{enumerate}
\item Suppose that $E$ is a value $V$. We calculate:
\[\sden{\SOp(V)}_\T  \; = \; \sden{\tuple{0,V}}_\T
                                         \; = \;  \eta_{\T}(\den{\Ssem_p}{V})
                                        \; = \; \eta_{\S}(\den{\Ssem_p}{V})(0)
                                        \; = \; \den{\Ssem}{V}(0)
\]

\item Suppose that $E$ has the form $E_1\,\myor\, E_2$. We calculate:
\[\begin{array}{llll} \hspace{-10pt}
 \sden{\SOp(E_1 \;\myor \;E_2)}_\T  \\
 \hspace{30pt} = \;   \sden{\SOp(E_1) \, \max{\ER{\DW}{-}{0}} \, \SOp(E_2)}_\T
                               & (\mbox{by Theorem~\ref{op-comp2}.\ref{parttwo-1}})\\
 \hspace{30pt} = \;   \sden{\SOp(E_1) \, \max{\alpha_\DW\circ\DW(0)} \, \SOp(E_2)}_\T \\
  \hspace{30pt} = \; \sden{\SOp(E_1) \, \max{\alpha_\T\circ\T(0)\circ\T(\sden{\;}_\sigma)} \, \SOp(E_2)}_\T
                                  & (\mbox{by Lemma~\ref{auxaux2}}) \\
\hspace{30pt} = \; \sden{\SOp(E_1)}_\T \, \max{\alpha_\T\circ\T(0)} \, \sden{\SOp(E_2)}_\T
                              & (\mbox{using Lemma~\ref{trivial}}) \\
\hspace{30pt} = \; \den{\Ssem}{E_1}(0) \, \max{\alpha_\T\circ\T(0)} \, \den{\Ssem}{E_2}(0)
                              & (\mbox{by induction hypothesis}) \\
\hspace{30pt} = \; \myor_{\sden{\sigma}}(\den{\Ssem}{E_1},\den{\Ssem}{E_2})(0)
                              & (\mbox{by Equation~\ref{or2def}})\\
\hspace{30pt} = \; \den{\Ssem}{ E_1 \myor \,E_2}(0)\\
\end{array}\]
\item    Suppose that $E$ has the form $c\cdot E'$. We calculate:
\[\begin{array}{lcll}\sden{\SOp(c\cdot E')}_\T
                                                & = & \sden{\sden{c}\cdot_{\T} \SOp(E')}_\T
                                                    &   (\mbox{by Theorem~\ref{op-comp2}.\ref{partthree-1}})\\
                                                  & = &   \sden{c}\cdot_{\T} \sden{\SOp(E')}_\T
                                                        &  (\mbox{as $ \sden{\mbox{-}}_\T$ is a homomorphism})\\
                                                  & = &   \sden{c}\cdot_{\T} \den{\Ssem}{E'}(0)
                                                        &   (\mbox{by induction hypothesis})\\
                                                 & = &   (\sden{c}\cdot_\S \den{\Ssem}{E'} )(0)
                                                      &                                   (\mbox{by Equation~\ref{rew1def}})\\
                                                 & = &   \den{\Ssem}{c\cdot E'}(0)
\end{array} \]

\item Suppose that $E$ has the form $E_1\, +_p \, E_2$. We calculate:
\begin{align*}
 \sden{\SOp(E_1 \, +_p\, E_2)}_\T  &= \sden{\SOp(E_1) \, +_p \, \SOp(E_2)}_\T
                \tag{by Theorem~\ref{op-comp2}.\ref{partfour-2}} \\
               &= \sden{\SOp(E_1)}_\T \, +_p \, \sden{\SOp(E_2)}_\T
                            \tag{as $ \sem_\sigma$ is a homomorphism} \\
               &= \den{\Ssem}{E_1}(0) \, +_p \, \den{\Ssem}{E_2}(0)
                            \tag{by induction hypothesis} \\
               &= ({+_p}_\S)_{\sden{\sigma}}(\den{\Ssem}{E_1},\den{\Ssem}{E_2})(0)\\
               &= \den{\Ssem}{E_1 +_p E_2}(0) \tag*{\qedhere}
\end{align*}
\end{enumerate}

\end{proof}
\noindent
We then have selection adequacy for our language with probabilities:
\begin{thm}[Selection adequacy]\label{sel-ad2} For any program $M\type\sigma$ we have:
\[\den{\Ssem}{M}(0) =  \sden{\SOp(M)}_\T\]
\end{thm}
The  proof of this theorem is the same as that of  Theorem~\ref{sel-ad1}.
As before, the adequacy theorem implies that the globally optimizing operational semantics determines the denotational semantics at the zero-reward continuation.

For the converse direction, noting that
 \[\begin{array}{lclll}( \sem_\T)_\sigma & = &  ((\eta_\T)_{\den{\Ssem}{\sigma}}\circ  \Ssem_p)^{\dagger_{\DW}}\\
                                     & = &  (\T( \Ssem_p)\circ (\eta_\T)_{\Val_\sigma})^{\dagger_{\DW}} & (\mbox{$\eta$ is natural})\\
                                     & = & \T( \Ssem_p)\circ ((\eta_\T)_{\Val_\sigma})^{\dagger_{\DW}}\\
                                     & = & \T( \Ssem_p)\circ ( \theta_{\DW,\T})_{\Val_\sigma} \\
\end{array}\]
we see from the adequacy theorem that, for $M\!\type\! \sigma$, the denotational semantics determines
$ (\theta_{\DW,\T})_{\Val_\sigma}(\SOp(M)) \in \T(\Val_\sigma)$ up to $\T( \Ssem_p)$. We view
$ (\theta_{\DW,\T})_{\Val_\sigma}(\SOp(M))$ as an observation of the selection operational semantics of $M$, and so, for $M\type \sigma$ we adopt the notation:
\[\Ob_{\sigma, \T}(M) = ( \theta_{\DW,\T})_{\Val_\sigma}(\SOp(M)) \in \T(\Val_\sigma)\]
Using this notation, we see that the adequacy theorem determines observations $\Ob_{\sigma, \T}(M)$ up to $\T( \Ssem_p)$.
With the aid of the above discussion of the monad morphism $\theta_{\DW,\T}$ we find for $M\type \sigma$ that:
\[\begin{array}{lcl}
\Ob_{\sigma,\T_1}(M) & = &\SOp(M)\\[0.2em]
\Ob_{\sigma,\T_2}(M) & = & \tuple{\VDis(\SOp(M)), \Rew(\SOp(M)| - )} \\[0.2em]
\Ob_{\sigma,\T_3}(M) & = &  \tuple{\VDis(\SOp(M)), \Ex(\SOp(M))}
\end{array} \]
In the case where $\sigma$ is a product of base types, $ \T(\Ssem_p)$ is an injection. (For $\Ssem_p$  is
then an injection and $\T$ preserves injections with nonempty domain, as do all functors on sets.) So in this case the denotational semantics  determines $\T$-observations $\Ob_{\sigma,\T}(M)$ of the selection operational semantics of terms $M\type\sigma$.

\subsection{Full abstraction}\label{sec:equiv2}

We continue to proceed generally, as above,  in terms of an auxiliary monad $\T$ and algebra $\alpha_\T\type \T(\R) \rightarrow \R$.
Having a general notion of observation $\Ob_{b,\T}$ at base types, we have  corresponding general observational  equivalence relations
$\approx_{b,\T}$, and so, instantiating, observational  equivalence relations   $\approx_{b,\T_i}$ for the $\T_i$.
We write $\Ob_\T$ and $\approx_\T$ for $\Ob_{\bool,\T}$ and
$\approx_{\bool,\T}$, respectively, and similarly for the $\T_i$.
From the discussion of the selection adequacy theorem (Theorem~\ref{sel-ad2}) at base types, we see that the implications~\ref{adimp} and~\ref{moimp} hold for $\Ssem_\T$ and all $\Ob_{b,\T}$ and $\approx_{b,\T}$; we  then also have $M \approx_{b,\T} \Op(M)$ for base types $b$ and programs $M\type \sigma$.

We next consider, as we did for our first language, whether observing at different base types makes a difference to contextual equivalence.
\begin{lem}\label{oeop} Suppose that $\B$ is characteristic for $\T$ and that $b$ is a base type with at least two constants. Then for any base type $b'$ and programs $M_1, M_2\type b'$ we have:
\[M_1 \approx_{b,\T}  M_2 \implies \Ob_{b',\T}(M_1)  = \Ob_{b',\T}(M_2)\]
\end{lem}
 \begin{proof}
 We can assume  without loss of generality that the $M_i$ are effect values, and write $E_i$ for them. We assume $E_1 \approx_{b,\T} E_2$, and
suppose, for the sake of contradiction,  that $\Ob_{b',\T}(E_1)  \neq  \Ob_{b',\T}(E_2)$.
As $\B$ is characteristic for $\T$,  there is a map  $f\!\!\type\!\! \Val_{b'}\!\! \hookrightarrow\!\! \Val_{\bool}$ such that
$\T(f)(\Ob_{b',\T}(E_1))  \neq \T(f)(\Ob_{b',\T}(E_2))$.
As $b$ has at least two constants, there is an injection $\iota\type \Val_{\bool} \rightarrow \Val_b$. Set $f' = f\circ \iota\type \Val_{b'} \rightarrow \Val_{b}$.
As $\T$ preserves injections with nonempty domain we have $\T(f')(\Ob_{b',\T}(E_1))  \neq \T(f')(\Ob_{b',\T}(E_2))$.
Let $g$ be the restriction of $f'$ to $g \type u \rightarrow \Val_b$, where $u$ is the set of constants of type $b$ occurring
in $E_1$ or $E_2$.

As $E_1 \approx_{b,\T} E_2$ we have $\myF_gE_1 \approx_{b,\T} \myF_gE_2$, so
$\Ob_{b,\T}(\myF_gE_1) = \Ob_{b,\T}(\myF_gE_2)$, and so, by adequacy, $\den{\Ssem}{\myF_gE_1}(0) = \den{\Ssem}{\myF_gE_2}(0)$.
For $i = 1,2$, we calculate:
\[\begin{array}{lcll}
 \den{\Ssem}{\myF_gE_i}(0)
                                     & = & \den{\Ssem}{E_i[g]}(0)  & (\mbox{by Lemma~\ref{consub}})\\
                                     & = & \sden{\SOp(E_i[g])}_\T & (\mbox{by adequacy})\\
                                     & = & \sden{\DW(f')\SOp(E_i)}_\T & (\mbox{by Lemma~\ref{sopsub2}}.\ref{sopsub2part2})\\
                                     & = & \T(\sem)( (\theta_{\DW,\T})_{\Val_{b}}(\DW(f')(\SOp(E_i))))&\\
                                     & = & \T(\sem)(\T(f') (\theta_{\Val_{b'}}(\SOp(E_i))))&
                                     (\mbox{as $\theta_{\DW,\T}$ is natural} \\
                                     &&& \mbox{\, by Lemma~\ref{thetamor}})\\
                                     & = & \T(\sem)(\T(f') (\Ob_{b',\T}(E_i)))&
 \end{array}\]
So, as $\T(\sem)$ is injective, $\T(f') (\Ob_{b,\T}(E_1)) = \T(f') (\Ob_{b,\T}(E_2))$, yielding the required contradiction.
\end{proof}

We then have the following analogue of Proposition~\ref{obvar1}:
\begin{prop}\label{obvar2} Suppose that $\B$ is characteristic for $\T$. Then, for all base types $b$ and programs $M,N\type \sigma$, we have
\[M \approx_\T N \implies \; M \approx_{b,\T} N\]
with the converse  holding if there are at least two constants of type $b$.
\end{prop}
 As $\B$ is characteristic for the $\T_i$ (Lemma~\ref{characteristic}), we have invariance of the observational equivalences $\approx_{b,\T_i}$ under changes of base type with at least two constants.
Modulo a reasonable definability assumption, each of our three semantics is 
fully abstract at base types with respect to their corresponding notion of observational equivalence. We establish this via general results for $\T$ and $\alpha_{\T}\type \T(\R) \rightarrow \R$, as above.

We first need a general result on reward continuations.
Suppose $u  \subseteq \Con_b$
and suppose too that $\gamma\type \sden{b} \rightarrow \R$
is \emph{definable on $u$} in the sense that there is a (necessarily unique)
$g\type \Con_b \rightarrow \Con_\Rew$ such that $\gamma(\sden{c}) = \sden{g(c)}$, for $c \in u$.
Set $\myK_{u,\gamma} = \myF_h\type b \rightarrow b$ where $h(c) = g(c)\cdot c\; (c \in u)$. We have:
\[\den{\Ssem}{\myK_{u,\gamma} c}(0) = \gamma(\sden{c})\cdot_\T \eta_\T(\sden{c}) \quad (c \in u)\]
This program $\myK_{u,\gamma}$ can be used to reduce calling definable reward continuations to calling the
zero-reward continuation, modulo reward addition:\\
\begin{lem}\label{reduce} Suppose $E\type b$ is an effect value, $u$ a finite set of constants of type $b$ including all those occurring in $E$,  and
$\gamma \type \sden{b} \rightarrow \R$ is a reward function definable on $u$. Then we have:
\[\myk_\gamma(\den{\Ssem}{E}(\gamma)) = \den{\Ssem}{\myK_{u,\gamma} E}(0)\]
\end{lem}
\cutproof{\begin{proof}
The proof is by structural induction on $E$. If $E$ is a constant $c$, then
\[\myk_\gamma(\den{\Ssem}{c}(\gamma)) = \myk_\gamma(\eta_\T(\sden{c}))
= \gamma(\sden{c})\cdot_\T \eta_\T(\sden{c})
= \den{\Ssem}{\myK_{u,\gamma} c}(0)\]

Suppose next that $E$ has the form $E_1 \,\myor\, E_2$. We first show that
\[\alpha_\T \circ\T(\gamma) = \alpha_\T \circ\T(0) \circ \myk_\gamma  \qquad (\ast)\]
We have
$\myk_\gamma = f^{\dagger_{\T}}\type \T(\sden{b}) \rightarrow \T(\sden{b})$, where $f(x) \eqdef \gamma(x)\cdot_\T \eta_\T(x)$, for $x \in \sden{b}$. Setting  $g(x) \eqdef \T(0)(\gamma(x)\cdot_\T \eta_\T(x)) \; (= \gamma(x)\cdot_\T \T(0)(\eta_\T(x)))$, for $x \in \sden{b}$, we then see that  $\T(0) \circ \myk_\gamma = g^{\dagger_{\T}}\type \T(\sden{b}) \rightarrow \T(\R)$.
 Making use of assumption (A2) of Section~\ref{den-sem2}, we next see that
 $\alpha_\T (g(x)) = \alpha_\T (\gamma(x)\cdot_\T \T(0)(\eta_\T(x))) = \alpha_\T (\gamma(x)\cdot_\T \eta_\T(0)) = \gamma(x) + 0 = \gamma(x)$.
  This, in turn, yields
 $\alpha_\T \circ\T(0) \circ \myk_\gamma = \alpha_\T \circ g^{\dagger_{\T}} = \alpha_\T \circ \T(\alpha_\T \circ g)= \alpha_\T \circ \T(\gamma)$
 as required. (The second equation in this chain holds generally for monad algebras.)

We then calculate:
\begin{align*}
& \myk_\gamma(\den{\Ssem}{E_1 \,\myor\, E_2}(\gamma)) \\
& \hspace{10pt} =
    \; \myk_\gamma(\den{\Ssem}{E_1}(\gamma) \, \max{\ER{\T}{-}{\gamma}} \, \den{\Ssem}{E_2}(\gamma)) \\
& \hspace{10pt} =
    \; \myk_\gamma(\den{\Ssem}{E_1}(\gamma) \, \max{\alpha_\T \circ\T(\gamma)} \, \den{\Ssem}{E_2}(\gamma)) \\
& \hspace{10pt} =
    \; \left \{\begin{array}{ll} \myk_\gamma(\den{\Ssem}{E_1}(\gamma)) &
(\mbox{if\ }(\alpha_\T \circ\T(\gamma))\den{\Ssem}{E_1}(\gamma) \geq (\alpha_\T \circ\T(\gamma))\den{\Ssem}{E_2}(\gamma))\\
                                                                        \myk_\gamma(\den{\Ssem}{E_2}(\gamma)) & (\mbox{otherwise})
   \end{array}\right. \\
& \hspace{10pt} =
    \; \left \{\begin{array}{ll}    \den{\Ssem}{\myK_{u,\gamma} (E_1)}(0) &
            (\mbox{if\ }(\alpha_\T \circ\T(0)\circ \myk_\gamma)\den{\Ssem}{ E_1}(\gamma) \geq (\alpha_\T \circ\T(0)\circ \myk_\gamma)\den{\Ssem}{E_2}(\gamma))\\
                                    \den{\Ssem}{\myK_{u,\gamma} (E_2)}(0)  & (\mbox{otherwise})
        \end{array}\right.\\
        \tag{by induction hypothesis and $(*)$} \\
& \hspace{10pt} = \;
\left \{\begin{array}{ll}    \den{\Ssem}{\myK_{u,\gamma} (E_1)}(0) &
           ( \mbox{if\ }(\alpha_\T \circ\T(0)) \myk_\gamma(\den{\Ssem}{ E_1}(\gamma)) \geq (\alpha_\T \circ\T(0))\myk_\gamma(\den{\Ssem}{E_2}(\gamma)))\\
                                    \den{\Ssem}{\myK_{u,\gamma} (E_2)}(0)  & (\mbox{otherwise})
        \end{array}\right. \\
& \hspace{10pt} =
    \; \left \{\begin{array}{ll}    \den{\Ssem}{\myK_{u,\gamma} (E_1)}(0) &
            (\mbox{if\ }(\alpha_\T \circ\T(0))  \den{\Ssem}{\myK_{u,\gamma} (E_1)}(0) \geq (\alpha_\T \circ\T(0)) \den{\Ssem}{\myK_{u,\gamma} (E_2)}(0))\\
                                    \den{\Ssem}{\myK_{u,\gamma} (E_2)}(0)  & (\mbox{otherwise})
        \end{array}\right. \tag{by Lemma~\ref{reduce}} \\
& \hspace{10pt} =
    \; \den{\Ssem}{(\myK_{u,\gamma} E_1) \,\myor\, (\myK_{u,\gamma}  E_2)}(0) \\ 
&\hspace{10pt} = \;  \den{\Ssem}{\myK_{u,\gamma} (E_1 \,\myor\, E_2)}(0) 
\tag{using Equation~\ref{scontext}}
\end{align*}

Suppose next that $E$ has the form $E_1 \,+_p \, E_2$. Then we calculate:
\[\begin{array}{lcll}
\myk_\gamma(\den{\Ssem}{E_1 +_p  E_2}(\gamma)) & = &  \myk_\gamma(\den{\Ssem}{E_1 }(\gamma) \, +_p\, \den{\Ssem}{E_2}(\gamma)) \\
& = & \myk_\gamma(\den{\Ssem}{E_1 }(\gamma)) \, +_p\,\myk_\gamma(\den{\Ssem}{E_2}(\gamma)) \\
& = &  \den{\Ssem}{\myK_{u,\gamma} E_1}(0)\, +_p\,\den{\Ssem}{\myK_{u,\gamma} E_2} (0) \\
& = &  \den{\Ssem}{\myK_{u,\gamma} E_1\, +_p\, \myK_{u,\gamma} E_2}(0)\\
& = &  \den{\Ssem}{\myK_{u,\gamma} (E_1\, +_p\, E_2)}(0)\\
\end{array}\]

Finally, suppose that $E$ has the form $c\cdot E'$. This case is handled similarly to the previous one:
\begin{align*}
\myk_\gamma(\den{\Ssem}{c \cdot E'}(\gamma)) & = \myk_\gamma(\sden{c}\cdot \den{\Ssem}{E'}(\gamma))\\
                                                                               & = \sden{c}\cdot \myk_\gamma(\den{\Ssem}{E'}(\gamma))\\
                                                                               & =  \sden{c}\cdot\den{\Ssem}{\myK_{u,\gamma} E'}(0)\\
                                                                               & = \den{\Ssem}{c \cdot\myK_{u,\gamma} E'}(0)\\
                                                                               & = \den{\Ssem}{\myK_{u,\gamma} (c \cdot E')}(0) \tag*{\qedhere}
\end{align*}
\end{proof}}

We can now demonstrate full abstraction for general $\T$, subject to three assumptions, Say that a type $b$ is \emph{numerable} if all elements of $\sden{b}$ are definable by a constant.\\
\begin{thm}\label{fullabT} Suppose $\B$ is characteristic for $\T$, reward addition is injective for $\T$, and $\Rew$ is numerable. Then $ \Ssem$ is fully abstract with respect to $\approx_\T$ at $b$.
\end{thm}
\cutproof{\begin{proof} Suppose $M_1 (\approx_\T)_b M_2$.  We wish to show that $\den{ \Ssem}{M_1} = \den{ \Ssem}{M_2}$. By the ordinary adequacy theorem there are effect values $E_1, E_2$ with $\den{ \Ssem}{M_1} = \den{ \Ssem}{E_1}$ and $\den{ \Ssem}{M_2} = \den{ \Ssem}{E_2}$.

Let $u$ be the set of constants  of type $b$ appearing in any one of these effect values.
Let $\gamma \type \sden{b} \rightarrow \R$  be a reward function. It is definable on $u$ by the numerability assumption.
Using Lemma~\ref{reduce} we see that

\[\myk_\gamma(\den{ \Ssem}{M_i}(\gamma))  =  \myk_\gamma(\den{ \Ssem}{E_i}(\gamma)) = \den{ \Ssem}{\myK_{u,\gamma} E_i}(0)\qquad (\ast) \]

As $M_1 (\approx_\T)_b M_2$, we have $E_1 (\approx_\T)_b E_2$ so $\myK_{u,\gamma} E_1 (\approx_\T)_b \myK_{u,\gamma} E_2$.
As $\B$ is characteristic for $\T$, we can then apply Lemma~\ref{oeop}, finding that $\Ob_{\T}(\myK_{u,\gamma} E_1) = \Ob_{\T}(\myK_{u,\gamma} E_2)$. So,  by adequacy, $\den{ \Ssem}{\myK_{u,\gamma} E_1}(0) = \den{ \Ssem}{\myK_{u,\gamma} E_2}(0)$.

With this, we see, using $(\ast)$, that
$\myk_\gamma(\den{ \Ssem}{M_1}(\gamma)) = \myk_\gamma(\den{ \Ssem}{M_2}(\gamma))$, so, as reward addition is injective for $\T$, that $\den{ \Ssem}{M_1}(\gamma) = \den{ \Ssem}{M_2}(\gamma)$.
As $\gamma$ is an arbitrary reward function, we finally have $\den{ \Ssem}{M_1} = \den{ \Ssem}{M_2}$ as required.
\end{proof}}

As, by Lemmas~\ref{characteristic} and~\ref{addition}, $\B$ is characteristic for all of three $\T_i$ and reward addition is injective for all of them, we immediately obtain:
\begin{cor}\label{concrete-case} Suppose that $\Rew$ is numerable. Then $ \Ssem_{\T_i}$ is fully abstract with respect to $\approx_{\T_i}$  at base types, for $i = 1,2,3$.
\end{cor}

Regarding full abstraction at other types, full abstraction for general $\T$ at products of base types and so, too, at values of types of order 1 is a consequence of Theorem~\ref{fullabT} (under the same assumptions as those of the theorem). We then obtain full abstraction for the $\T_i$ at products of base types and at values of types of order 1 (assuming $\Rew$  numerable). As in the case of the language of Section~\ref{first}, the question of full abstraction at other types is open.

There is a ``cheap'' version of the free-algebra monad $\C$ discussed in Section~\ref{subsec:equiv1} for general auxiliary monads $\T$. Take $\Ax_c$ to be the set of equations between effect values that hold in $\Ssem_\T$, and take $\mathrm{D}$ to be the corresponding free algebra monad, yielding a corresponding denotational semantics $\mathcal{D}$. Then we have:
\[\models_\mathcal{D} E = E'\type b\; \iff \; \vdash_{\Ax_c} E = E' \type b \; \iff \; \models_{\Ssem_\T} E = E'\type b\]
Assuming $\B$ characteristic for $\T$ and reward addition  injective for $\T$, using Theorems~\ref{basic-ad} and~\ref{fullabT} we then obtain a version of Theorem~\ref{theorem:equivalences} for $\mathcal{D}$ for numerable $b$:
\[\vdash_{\Ax_c} M = N \type b \;\iff\; \models_\mathcal{C} M = N \type b\; \iff \; M \approx_b N \]

Turning to weakening the notion of observation, analogously to Section~\ref{first} we could forget all reward information. We do this  by taking  our notion of observation $\Ob_{\VDis}$ to be
$\VDis\circ \SOp$, i.e., the distribution of final values. As we next show, the observational equivalence $\approx_{\Ob_{\VDis}}$ resulting from this notion coincides with $\approx_{\T_3}$.
 \begin{lem}\label{expectedreward}   For programs $M,N \type \bool$ we have
 \[M_1 \approx_{\Ob_{\VDis}}   M_2  \implies \Ob_{\T_3}(M_1) = \Ob_{\T_3}(M_2)\]
  \end{lem}
  \begin{proof}
  Set $E_i = \Op(M_i)$, for $i=1,2$. We have $M_i\approx_{\T_3}  E_i$ and so, as  $\Ob_{\VDis}$ is weaker than $ \Ob_{\T_3}$, we also have $M_i\approx_{\Ob_{\VDis}}  E_i$. It therefore suffices to prove that:
   \[E_1 \approx_{\Ob_{\VDis}}   E_2  \implies \Ob_{\T_3}(E_1) = \Ob_{\T_3}(E_2)\]
  So suppose that $E_1 \approx_{\Ob_{\VDis}}   E_2$ and,
for the sake of contradiction, that, for example,  $\Ex(\SOp(E_1)) < \Ex(\SOp(E_2))$.

  Since $E_1 \approx_{\Ob_{\VDis}}   E_2$, they return the same probability distribution $\true +_p \false$ on boolean values.
Suppose,  without loss of generality, that this distribution is not $\false$. (If it is, we can work with $\true$ instead.)
Define $f\type \Val_\bool \rightarrow \Val_\bool$ to be constantly
$\false$.
Then we have
\[\begin{array}{lcll}
\Ob_{\VDis}(E_1 \,\myor\, \myF_fE_2) & = & \VDis(\SOp(E_1 \,\myor\, \myF_fE_2))\\
                                                             & = & \VDis(\SOp(E_1) \,\max{\Ex}\, \SOp(\myF_fE_2))&
                                                                        (\mbox{by Theorem~\ref{op-comp2}.\ref{parttwo-2}})\\
                                                             & = & \VDis(\SOp(E_1) \,\max{\Ex}\, \SOp(E_2[f]))&
                                                                        (\mbox{by Lemma~\ref{consub}})\\
                                                             & = & \VDis(\SOp(E_1) \,\max{\Ex}\, \W(f)(\SOp(E_2)))&
                                                                        (\mbox{by Lemma~\ref{sopsub2}.\ref{sopsub2part2}})\\
                                                             & = & \VDis(\W(f)(\SOp(E_2)))\\
                                                             & = & \false
\end{array}\]
where the next to last equality holds as, using Lemma~\ref{sopsub2}.\ref{sopsub2part1}, we have:
\[\Ex(\SOp(E_1))  < \Ex(\SOp(E_2))= \Ex(\W(f)(\SOp(E_1)))\]

Similarly,
\[\begin{array}{lcll}
\Ob_{\VDis}(E_2 \,\myor\,  \myF_fE_2) & = & \VDis(\SOp(E_2) \,\max{\Ex}\, \W(f)(\SOp(E_2)))\\
                                                             & = & \VDis(\SOp(E_2))\\
                                                             & = & \true +_p \false
\end{array}\]
yielding the required contradiction.
\end{proof}
We then have the following analogue to  Theorem~\ref{obweak1}:
\begin{thm}\label{obweak2} For any programs $M,N \type \sigma$, we have
\[M \approx_{\T_3} N \iff M \approx_{\Ob_{\VDis}}   N\]

\end{thm}

\subsection{Program equivalences and purity}\label{probeqpu}

We begin by considering the equations holding in   $\Ssem_{\T}$ for a general $\T$ as above. We need some terminology and notation. Say that a term $M\type \sigma$ is in \emph{expectation PR-form over terms $L_1, \ldots,L_n$} if it has the form
\[\sum_{i=1}^m p_i\sum_{j=1}^{n_i}q_{ij}(M_{ij}\cdot L_i) \]
where the $M_{ij}$ are either variables or constants (and we say $M$ is an \emph{expectation PR-value} if  the $M_{ij}$ and $L_i$ are all constants).
 For such a term we write $\Ex_s(M)\type \Rew$ for the term:
\[\bigoplus_{i=1}^m p_i\bigoplus_{j=1}^{n_i}q_{ij}M_{ij}\]
(We write $\bigoplus M_i$ for iterated uses of the $\oplus_p$ to avoid confusion with iterated uses of the $+_p$.)
In case the $M_{ij}$ are constants $d_{ij}$, we set:
\[\left [\!\! \! \left [ \bigoplus_{i=1}^m p_i\bigoplus_{j=1}^{n_i}q_{ij}d_{ij}\right ]\!\! \! \right ] = \sum_{i=1}^m p_i\sum_{j=1}^{n_i}q_{ij}\sden{d_{ij}}\]

Our system $\Ax_1$ of equations is given in Figure~\ref{ptermequivs} (where we omit type information). In the last two equations it is assumed that $M$ and $N$ are in expectation PR-form over the same $L_1,\ldots,L_n$. The equations express at the term level, that:
choice is idempotent and associative;
rewards form an action for the commutative monoid structure on $\R$;
probabilistic choice forms a convex algebra;
the $\R$-action acts on both forms of choice; probabilistic choice distributes over choice; and, where this can be seen from the syntax, that choice is made according to the highest reward, with priority to the left for ties.

\begin{figure}[h]
  \[\begin{array}{c}
  M \,\myor\, M \; = \; M \qquad (L \,\myor\, M) \,\myor\, N \; = \; L \,\myor\, (M \,\myor\, N)  \\[0.5em]
0\cdot N \; = \; N \qquad x \cdot (y \cdot N) \; =\;  (x + y)\cdot N\\[0.5em]
M +_1 N = M \qquad M +_p N \; =\;  N +_{1-p} M\\[0.2em]
(M +_p N) +_q P \; =\;  M +_{pq} (N +_{\frac{(1 - p)q} {1 - pq}}  P) \quad (p,q < 1)\\[0.5em]
x\cdot (M+_p N) \; = \; x\cdot M +_p x\cdot N\\[0.2em]
x \cdot (M\, \myor\, N) \; =\;  (x\cdot M)\, \myor\, (x \cdot N)\\[0.2em]
 L +_p (M \,\myor\, N) \; = \; (L +_p M)\, \myor\, (L +_p N)\\[0.3em]
\myif \Ex_s(M) \geq \Ex_s(N) \mythen M \myelse N \; = \; M\, \myor\,  N \\[0.3em]
\myif \Ex_s(M) \geq \Ex_s(N) \mythen (M\, \myor\, P) \myelse ( P \,\myor\,N)\,\; = \;\,
(M\, \myor\, P) \,\myor\, N \\\end{array}\]

\caption{Equations for choices, probability, and rewards}\label{ptermequivs}
\end{figure}

 Below, for $u \in \T(X)$ and $\gamma\type X \rightarrow \R$, we set
 \[\Ex(u|\gamma) \eqdef \ER{\T}{u}{\gamma}\; (= \alpha_\T(\T(\gamma)(u)))\]
 This is  the expected reward of $u$, given $\gamma$.

\begin{prop}
The axioms hold for general $\Ssem_\T$.
\end{prop}
\begin{proof} Other than the last two axiom schemas, this follows from  Theorem~\ref{genax}, Corollary~\ref{auxeq}, and Theorem~\ref{distributes}.
The last two cases are straightforward pointwise arguments, although we need an observation. We calculate that for a PR-term
$M = \sum_{i=1}^m p_i\left(\sum_{j=1}^{n_i}q_{ij}(d_{ij}\cdot L_i)\right )$ of type $\sigma$ and a reward function
$\gamma\type \sden{\sigma} \rightarrow \R$
we have:
{\small \[\begin{array}{lcl} \Ex(\den{\Ssem}{\sum_{i=1}^m p_i\sum_{j=1}^{n_i}q_{ij}(d_{ij}\cdot L_i)}\gamma \mid\gamma)
&\! = \!\!&
\sum_{i=1}^m p_i (\sum_{j=1}^{n_i}q_{ij}(\sden{d_{ij}} +  \Ex(\den{\Ssem}{L_i} \gamma \mid\gamma) ))\\[0.5em]
& \!\! = \!\!&
\sum_{i=1}^m p_i (\sum_{j=1}^{n_i}q_{ij}\sden{d_{ij}}   )+  \\[0.5em]
&& \quad \sum_{i=1}^m p_i (\sum_{j=1}^{n_i}q_{ij}\Ex(\den{\Ssem}{L_i} \gamma \mid\gamma)  )\\[0.5em]
&\!\! = \!\!& \sum_{i=1}^m p_i (\sum_{j=1}^{n_i}q_{ij}(\sden{d_{ij}}   ))+
\sum_{i=1}^m p_i\Ex(\den{\Ssem}{L_i} \gamma \mid\gamma)\\
\end{array}\]}
and we further have:
\[\Ssem\left [\!\! \! \left [\bigoplus_{i=1}^m p_i \bigoplus_{j=1}^{n_i}q_{ij}d_{ij} \right ] \!\! \! \right ]\gamma =  \eta_\T\left (\sum_{i=1}^m p_i\sum_{j=1}^{n_i}q_{ij}\sden{d_{ij}} \right  )\]
So if $M\type \sigma$ and $N\type\sigma $ are in expectation PR-form over the same $L_1,\ldots,L_n$ then, for $\gamma\type \sden{\sigma} \rightarrow \R$, we have:
\[\Ex(\den{\Ssem}{M}  \gamma \mid\gamma) \geq \Ex(\den{\Ssem}{N}  \gamma \mid\gamma)
\iff
\den{\Ssem}{\Ex_s(M) \geq \Ex_s(N)}\gamma = \eta_\T(0)\]
(recall that $\sden{\true} = 0$).
With this observation, the pointwise argument for the last two equation schemas goes through.
\end{proof}
 In the case of $\Ssem_{\T_2}$ we inherit Equation~\ref{gather-eqn2} from  $\T_2$ so we additionally have:
 \begin{equation} x\cdot M \,+_p \, y\cdot M = (x \oplus_p y) \cdot M\label{gather-eqn2-syntax} \end{equation}
  Let $\Ax_2$ be $\Ax_1$ extended with this equation.
  In the case of $\Ssem_{\T_3}$ we inherit Equation~\ref{gather-eqn3} from  $\T_3$ so we have the stronger:
  \begin{equation} x\cdot M \,+_p \, y\cdot N =  (x \oplus_p y) \cdot M \,+_p\,  (x \oplus_p y) \cdot N\label{gather-eqn3-syntax} \end{equation}
    Let $\Ax_3$ be $\Ax_1$ extended with this equation.

Unfortunately, we do not have any results analogous to Theorem~\ref{theorem:equivalences} for any of the above three axiom systems for the probabilistic case---further axioms may well be needed to obtain  completeness for program equivalence at base types.
 We do, however have a completeness result for purity at base types.

 First, some useful consequences of these equations, are the following, where $M$ and $N$ are expectation PR-values over the same $L_1,\ldots,L_n$:
\begin{equation}
  \tag{PR$_1$}
M \,\myor\, N \;\; =  \;\;  M \quad (\mbox{if $\Ex_s(M)\geq \Ex_s(N)$})%
  \label{PR1}
\end{equation}
 \begin{equation}
  \tag{PR$_2$}
M \,\myor\, N \;\; =  \;\; N \quad (\mbox{if $\Ex_s(M) < \Ex_s(N)$})%
  \label{PR2}
\end{equation}

\begin{equation}
  \tag{PR$_3$}
(M \, \myor\,  L) \,\myor\,  N\; = \; M \, \myor\, L \quad (\mbox{if $\Ex_s(M)\geq \Ex_s(N)$})\%
  \label{PR3}
\end{equation}
\begin{equation}
  \tag{PR$_4$}
(M \, \myor\, L) \,\myor\,  N\; = \;   L \,\myor\, N \quad (\mbox{if $\Ex_s(M) <  \Ex_s(N)$})\%
  \label{PR4}
\end{equation}

Next, our equational system $\Ax_1$ allows us to put programs of base type into a weak canonical form.
First consider programs which are \emph{PR-effect values}, i.e., programs obtained by probabilistic and reward combinations of constants. Every such term is provably equivalent to one of the form
$\sum_{j = 1}^{m} p_j (d_j \cdot c_j)$ where $m > 0$, the $d_j$ and the $c_j$ are constants and no $d_j \cdot c_j$ is repeated. We call such terms \emph{canonical  PR-effect values}, and do not distinguish any two such if they are identical apart from the ordering of the $d_j \cdot c_j$.

We say that an effect value of base type is in \emph{weak canonical form} if (ignoring bracketing of $\myor$) it is an effect value of the form
\[ E_1 \,\myor\, \ldots \,\myor\, E_n\]
where $n > 0$, the $E_i$ are canonical PR-effect values, and no $E_i$ occurs twice. (We could have simplified canonical forms further by applying the PR$_i$, obtaining a stronger canonical form. However, we did not do so as, in any case, we do not have an equational completeness result.)
\begin{lem}  Every program $M$ of base type is provably equal to a weak canonical form $\CF(M)$.
\end{lem}
\begin{proof} By Proposition~\ref{proof}, $M$ can be proved equal to an effect value $E$. Using the associativity equation and the fact that $\give$ and $+_p$ distribute over $\myor$, the effect value $E$ can be proved equal to a term of the form
$E_1 \, \myor \, \ldots\, \myor \, E_n$  where each $E_i$ is a PR-effect term.
\end{proof}

Say that a theory $\Ax$, valid in $\S_\T$, is \emph{strongly purity complete for basic PR-effect values}, if for all PR-effect values $E\type b$ we have:
\[ \models_{\S_\T} E \downarrow_b \; \implies \;   \exists c\type b.\, \vdash_{\Ax} E =  c \type b   \]

\begin{lem}\label{ivpcomplete} $\Ax_i$ is strongly purity complete for basic PR-effect values, for $i=1,2,3$.
\end{lem}
\begin{proof} For $\T_1$ we have already noted that every PR-effect value is provably equal using $\Ax_1$ to a term $E\type b$ of the form
$\sum_{j = 1}^{m} p_j (d_j \cdot c_j)$ with no $d_j \cdot c_j$ repeated. For such a term $\models_{\S_{\T_1}} E \downarrow_b$ holds iff there is an $x \in \sden{b}$ such that $\den{\Ssem_{\T_1}}{E}\gamma = \eta_{\T_1}(b)$ for all $\gamma \type \sden{b} \rightarrow \R$. Taking $\gamma = 0$, for example, we then see that
$\sum_{j = 1}^{m} p_j (\sden{d_j} \cdot \sden{c_j}) = \eta_{\T_1}(b)$. As no $d_j \cdot c_j$ is repeated, neither is any $\sden{d_j} \cdot \sden{c_j}$. It follows that $m=1$ and $d_1 = 0$. In that case the term is provably equal, using $\Ax_1$, to $c_1$. The other two cases are similar:
for $\T_2$ we note that every PR-effect value is provably equal using $\Ax_2$ to a term of the form
$\sum_{j = 1}^{m} p_i (d_j \cdot c_j)$ with no $c_i$ repeated,
and for $\T_3$ we note that every PR-effect value is provably equal using $\Ax_3$ to a term of the form
$\sum_{j = 1}^{m} p_i (d \cdot c_j)$ with no $c_j$ repeated.
\end{proof}

 In order to establish purity completeness we need a condition (C) on $\R$. This is that for all $p\in (0,1)$ and $s < 0$ in $\R$, there are $l,r \in\R$ such that  $s + (r +_p l) >  l$.
 Condition (C) evidently holds when there are no negative elements as in our example of the nonnegative reals $\ropen{0,\infty}$ with the addition monoid. It also holds for our other examples of the reals, $(-\infty,\infty)$, and the positive reals, $(0,\infty)$, the former with the addition monoid and the latter with the multiplication monoid.
 Two further examples satisfying the condition are the real intervals $\lopen{-\infty,0}$ and $\lopen{0,1}$, both with the usual ordering, the first with the sum monoid, and the second with the multiplication monoid. In all these examples we employ the usual convex combination, and the verification of Condition (C) is straightforward.
 We give a counterexample to the condition below.

There are natural conditions that imply Condition (C), and which, together, account for these examples. Consider the equation:
\begin{equation} x +_{1/2} (y + z)  = (x + y) +_{1/2} z\label{shift} \end{equation}
and say that condition (D) holds if,
for all $p\in (0,1)$,  there is an $l\in \R$ such that for all $s\in \R$ there is an $r\in \R$ such that $r +_p s > l$. Condition (C) is satisfied if Equation~\ref{shift} holds or Condition (D) does.
All our examples with the addition monoid satisfy the equation, 
and all our examples other than the nonpositive reals satisfy Condition (D).

\begin{thm}[General purity completeness]\label{genpurity2} Suppose that  $\R$ satisfies condition (C).  Let $\Ax$ be a theory extending $\Ax_1$ that is valid in $\S_\T$. If $\Ax$ is strongly purity complete for basic PR-effect values, then it is strongly purity complete at base types, i.e., for all programs $M\type b$ we have:

\[ \models_{\S_\T} M \downarrow_b \implies   \exists c\type b.\, \vdash_{\Ax} M =  c \type b   \]
\end{thm}
\begin{proof}
We remark first that, in general, for any term $N\type b$ and any PR-effect value $E\type b$, if $\models_{\S_\T} N \downarrow_b $ and
$\den{\S_\T}{N}(\gamma) = \den{\S_\T}{E}(\gamma)$ for some $\gamma$, then $\den{\S_\T}{N}(\gamma) = \den{\S_\T}{E}(\gamma)$ for any $\gamma$, and so, also,
$\models_{\S_\T} E \downarrow_b $ and then $\vdash_{\Ax}  E = c$, for some $c\type b$ (this last using the strong purity completeness assumption).

It suffices to prove the claim for terms $M$  in weak canonical form, i.e., of the form
\[ E_1 \,\myor\, \ldots \,\myor\, E_n\]
where $n > 0$, and the $E_i$ are canonical PR-effect values. We proceed by induction on $n$.

So suppose that $\models_{\S_\T} M \downarrow_b$.  For some $E_{i_0}$ we have $\den{\S_\T}{M}(0) = \den{\S_\T}{E_{i_0}}(0) $,
and so, by the above remark, we see that $\den{\S_\T}{N} = \den{\S_\T}{E_{i_0}}$ and also that there is a $\ov{c}\type b$ such that $\vdash_{\Ax}  E_{i_0} = \ov{c}$.

In case $n=1$ we have shown that  $\vdash_{\Ax}  M = c$ for some $c\type b$, as required.
Otherwise
consider $E_{i_1} = \sum_{j = 1}^{n} p_j (d_j \cdot c_j)$ for an $i_1 \neq i_0$. If every  $c_j$ is $\ov{c}$ then both $E_{i_0}$ and $E_{i_1}$ are in expectation PR-value form over $\ov{c}$, and so one of the equations~\ref{PR1}--\ref{PR4} can be used to reduce the size of $M$, and the induction hypothesis can be applied.

Otherwise,  some $c_j$ is not $\ov{c}$, and we show next that, for some $\gamma$ we have
\[\Ex(\den{\S_\T}{E_{i_1}}\gamma \mid \gamma)\; > \;\Ex(\den{\S_\T}{E_{i_0}}\gamma \mid \gamma) \quad\; (\ast)\]
There are two cases. In the first case no $c_{j_1}$ is $\ov{c}$. Then choose $l < r \in \R$ and
define $\gamma \type \sden{b} \rightarrow \R$ by setting $\gamma(x) = r$ for $x\neq  \sden{\ov{c}}$, and  $\gamma(\sden{\ov{c}}) = r_0+ l$, where $r_0$ is the least of the  $\sden{d_j}$. Then we have:
\[\begin{array}{lcl}
\Ex(\den{\S_\T}{E_{i_1}}\gamma \mid \gamma)  & = &  \sum_{j = 1}^{n} p_j(\sden{d_j} + \gamma(c_j))\\
                                                                            & \geq   & \sum_{j = 1}^{n} p_j(r_0 + \gamma(c_j))\\
                                                                            & =   & \sum_{j = 1}^{n} p_j(r_0 + r)\\
                                                                             & = & r_0 + r\\
                                                                             & > & r_0 +  l\\
                                                                             & = & \gamma(\ov{c})\\
                                                                             & = & \Ex(\den{\S_\T}{E_{i_0}}\gamma \mid \gamma)\\
\end{array}\]

In the second case $c_{j_0} = \ov{c}$ for some unique $j_0$. Setting $p = \sum_{j \neq j_0}p_j$, note that $p\in (0,1)$; then, setting $p'_j = p_j/p$ for $j \neq j_0$, note that $\sum_{j \neq j_0} p'_j = 1$.
Taking $r_0$ to be the least of the  $\sden{d_j}$ as before, there are $l$ and $r$ in $\R$ such that
$r_0 + (r +_p l)  > l$.
For if $r_0 < 0$, condition (C) applies, and otherwise $r_0 \geq 0$ and we can choose any $l, r$ with $l <r$.
Define $\gamma \type \sden{b} \rightarrow \R$ by setting $\gamma(x) = r$ for $x\neq  \sden{\ov{c}}$, and  $\gamma(\sden{\ov{c}}) = l$. Then we have:
\[\begin{array}{lcl}
\Ex(\den{\S_\T}{E_{i_1}}\gamma \mid \gamma)  & = &  \sum_{j = 1}^{n} p_j(\sden{d_j} + \gamma(c_j))\\
                                                                            & \geq   & \sum_{j = 1}^{n} p_j(r_0 + \gamma(c_j))\\
                                                                           & = &  (\sum_{j \neq j_0} p'_j(r_0 + \gamma(c_j))) +_p (r_0 + \gamma(\ov{c}))\\
                                                                           & = &  (\sum_{j \neq j_0} p'_j (r_0 + r)) +_p (r_0 + l)\\
                                                                           & = &  (r_0 +  r ) +_p (r_0 + l)\\
                                                                           & = &  r_0 +  (r  +_p  l)\\
                                                                             & > &  l\\
                                                                             & = & \gamma(\ov{c})\\
                                                                             & = & \Ex(\den{\S_\T}{E_{i_0}}\gamma \mid \gamma)\\
\end{array}\]

This establishes $(\ast)$. So, for some $E_{i_2}$, with $i_2 \neq i_0$, $\den{\S_\T}{M}(\gamma)  = \den{\S_\T}{E_{i_2}}(\gamma)$, and
so $\den{\S_\T}{M} = \den{\S_\T}{E_{i_2}}$ and there is a $c_{i_2}\type b$ such that $\vdash_{\Ax}  E_{i_2} = c_{i_2}$. As $\Ax$ is valid in $\S_\T$ we have $\den{\S_\T}{\ov{c}} = \den{\S_\T}{c_{i_2}}$ and so $\ov{c} = c_{i_2}$. (Monad units are always injective and so is $\sden{\mbox{-}}\type \Val_b \rightarrow \sden{b}$.) We can therefore replace $E_{i_0}$ and $E_{i_2}$ by $\ov{c}$, apply one
of~\ref{PR1}--\ref{PR4}, to obtain a shorter canonical form, and then apply the induction hypothesis. This concludes the proof.
\end{proof}

So, using Lemma~\ref{ivpcomplete}, we see that strong purity completeness  at base types holds for $\T_i$ with respect to  the $\Ax_i$ (assuming $\R$ satisfies condition (C))\footnote{In~\cite{AP21} this was claimed without any assumption on $\R$; however there was an error in the proof.
}.
Regarding  products of base types, strong purity completeness for general $\T$ at products of base types follows from Theorem~\ref{genpurity2} (under the same assumptions as those of the theorem), and so, then, for the $\T_i$ (assuming $\R$ satisfies condition (C)).

While Condition (C) is not attractive, it is necessary:
\begin{thm} Suppose $\R$ does not satisfy condition $(C)$. Then $\Ax_1$ is not purity complete for $\T_1$. That is, there is a term $M$ such that $\models_{\S_\T} M \downarrow_b$ holds but $\vdash_{\Ax_1} M \downarrow_b$ does not.
\end{thm}
\begin{proof} As the condition fails, we can choose $p \in (0,1)$ and $s < 0$ such that, for all $l$ and $r$ we have
$s + (r +_p  l)  \leq l$. Take $M$ to be the term $\false \; \myor\; c\cdot(\true +_p \false)  $ where $\sden{c} = s$. Then for all
$\gamma: \sden{\bool} \rightarrow \R$ we have $\Ex(\den{\S_\T}{c\cdot(\true +_p \false)}\gamma \mid \gamma) \leq \Ex(\den{\S_\T}{\false}\gamma \mid \gamma)$ and so $\models_{\S_\T} M \downarrow_b$.  However, switching to any $\R$ satisfying condition (C), we see that if $\vdash_{\Ax_1} M \downarrow_b$ then, by consistency, we would have $\models_{\S_\T} M \downarrow_b$. But this is impossible as, using condition (C), we can find a $\gamma: \sden{\bool} \rightarrow \R$ such that
$\Ex(\den{\S_\T}{c\cdot(\true +_p \false)}\gamma \mid \gamma) > \Ex(\den{\S_\T}{\false}\gamma \mid \gamma)$ and so
$\den{\S_\T}{M}\gamma = \den{\S_\T}{c\cdot(\true +_p \false)}\gamma $ and this contradicts $\models_{\S_\T} M \downarrow_b$ as $\den{\S_\T}{c\cdot(\true +_p \false)}\gamma \neq \eta_{\T_1}(b)$ for any $b\in \bool$. \end{proof}

\newcommand{\mymin}{\mathrm{m}}
\newcommand{\myw}{\mathrm{w}}
\newcommand{\myprec}{\mathrm{p}}
To conclude our discussion of purity we construct a counterexample to Condition (C). We make use of the free barycentric commutative algebra $\R_M$ over a commutative monoid  $(M,+,0)$. This is the set of finite probability distributions over $M$, with the usual convex combination operations, with \emph{convolution} as the monoid operation, defined by:
\[(\sum_ip_ix_i) + (\sum_jq_jy_j) = \sum_{ij}p_iq_j(x_i + y_j)\]
and with $0$ the Dirac distribution $\delta_0$.

Consider the case where the monoid $M$ is  totally ordered, with the monoid operation preserving and reflecting the order. Every finite distribution over $M$ can then be written uniquely  in the form
$\mu = \sum_{i= 1}^np_ix_i$ with $x_1 > \cdots > x_n$ (and no $p_i$ zero). Set $\mymin(\mu) = x_n$, $\myw(\mu) = p_n$, and, if $n > 1$, $\myprec(\mu) = \sum_{i= 1}^{n-1}\frac{p_i}{1-p_n}x_i$. Note that
$\mymin(x + \mu) = x + \mymin(\mu)$, for $x
\in M$, and that $\mymin(\mu +_p \nu) = \min(\mymin(\mu),\mymin(\nu))$, for $p\in (0,1)$.

Let $\leq$ be the least relation on $\R_M$ such that:
\[\frac{\mymin(\mu) < \mymin(\nu)}{\mu \leq \nu} \qquad
\frac{\mymin(\mu) = \mymin(\nu), \myw(\mu) > \myw(\nu)}{\mu \leq \nu} \qquad
\frac{\mymin(\mu) = \mymin(\nu), \myw(\mu) = \myw(\nu) = 1}{\mu \leq \nu}\]
\[\frac{\mymin(\mu) = \mymin(\nu), \myw(\mu) = \myw(\nu) \neq 1, \myprec(\mu) \leq \myprec(\nu)}{\mu \leq \nu}\]
Intuitively, one decides whether $\mu \leq \nu$  or $\nu \leq \mu$ by comparing $\mymin(\mu)$ and $\mymin(\nu)$, and, if they are equal, comparing their corresponding probabilities, and then if they are equal, but not $1$, proceeding recursively to the rest of $\mu$ and $\nu$. It can be shown that $\leq$ is a  total order, preserved and reflected by probabilistic choice and addition.
 Note that if $\mu \leq \nu$ then $\mymin(\mu) \leq \mymin(\nu)$.

Suppose now that $M$ contains an element $s < 0$ (so $M$ could, for example,  be the nonpositive integers with the usual addition and order). Then we claim that $\R_M$ does not satisfy Condition (C). For, 
suppose there are $l,r$ such that $s + (r +_p l) >  l$.
 We have:
\[\begin{array}{lcl}\mymin(s + (r +_p l)) & = & s + \mymin(r +_p l) \\
& = & s + \min(\mymin(r), \mymin(l))\\
& \leq & s +  \mymin(l)\\
&  <  & \mymin(l)
\end{array}\]
However, this contradicts $s + (r +_p l) >  l$ as that implies that $\mymin(s + (r +_p l)) \geq  \mymin(l)$.

 \section{Conclusion}\label{conclusion}

This paper studies decision-making abstractions in the context of
simple higher-order programming languages, focusing on their
semantics, treating them operationally and denotationally. The
denotational semantics are compositional.  They are based on the selection
monad, which has rich connections with logic and game theory.
Unlike other  programming-language research (e.g.,~\cite{ajm:pcf,ho:pcf}),
the treatment of games in this paper is 
extensional, focusing on choices but ignoring other
aspects of computation, such as function calls and returns.
Moreover, the games are one-player games.
Going further, we
have started to explore extensions of our languages with multiple players,
where each choice and each reward is associated
with one player. For example, writing $\Ab$ and $\El$ for the players, we can program a version of the classic prisoners's dilemma:
\[
\begin{array}{l}
\mylet \texttt{silent}_\Ab, \texttt{silent}_\El  \type \bool \mybe (\true \,\myorA\, \false),(\true \,\myorE\, \false) \myin\\
\myif \texttt{silent}_\Ab \myand \texttt{silent}_\El \mythen -1 \cdot_{\Ab} -1 \cdot_{\El} \ast\\
\myelse \myif \texttt{silent}_\Ab \mythen -3 \cdot_{\Ab} \ast\\
\myelse \myif \texttt{silent}_\El \mythen -3 \cdot_{\El} \ast\\
\myelse -2 \cdot_{\Ab} -2 \cdot_{\El} \ast
\end{array}\]
Here,  $\texttt{silent}_\Ab$ and $\texttt{silent}_\El$ indicate whether the players remain silent, and the rewards, which are negative, correspond to years of prison. Semantically it would be natural to use the selection monad with $\mathbb{R}^2$ as the set of rewards, and with the writer monad as auxiliary monad. (One could envisage going further and treating probabilistic  games via a combination of the writer monad and a monad for probability.)
Many of our techniques carry over to languages with multiple players, which
give rise to interesting semantic
questions  (e.g., should we favor some players over others? require Nash equilibria?)
and may also be useful in practice.

Multi-objective optimization provides another area of interest. One could take $\R$ to be a product, with one component for each objective,  and use the selection monad augmented with auxiliary monad the combination $\mathcal{P}_{\small \mathrm{fin}}(\R \times - )$ of the finite powerset monad and a version of the writer monad enabling writing to  different components. One would aim for a semantics returning Pareto optimal choices.

In describing Software 2.0, Karpathy suggested specifying some goal on
the behavior of a desirable program, writing a ``rough skeleton'' of
the code, and using the computational resources at our disposal to
search for a program that works~\cite{Karpathy}.  While this vision
may be attractive, realizing it requires developing not only search
techniques but also the linguistic constructs to express goals and
code skeletons.  In the variant of this vision embodied in
SmartChoices, the skeleton is actually a complete program,
albeit in an extended language with decision-making abstractions.
Thus, in the brave new world of Software 2.0 and its relatives,
programming languages still have an important role to play, and their
study should be part of their development.
Our paper aims to contribute to one aspect of
this project;  much work remains.

In comparison with recent theoretical work on languages with 
differentiation (e.g,.~\cite{FongST19,abadiplotkin2020,BartheCLG20,pagani2020,cruttwell2019,HuotSV20}), our
languages are higher-level: %
they focus on how
optimization or machine-learning may be made available to a programmer
rather than on how they would be implemented. However, a convergence
of these  research lines is possible, and perhaps desirable. One thought
is to extend our languages with
differentiation primitives
to construct selection functions that use gradient descent.
These would be alternatives to ${\argmax}$ as discussed in the Introduction.
Monadic reflection and
reification, in the sense of Filinski~\cite{filinski:representing}, could support the use of such
alternatives, and more generally enhance programming flexibility.
Similarly, it would be attractive to deepen the connections
between our languages and probabilistic ones (e.g.,~\cite{Goodman:church}).
It may also be interesting to connect our semantics with particular techniques from
the literature on MDPs and RL, and further to explore whether monadic
ideas can contribute to implementations that include such techniques.
Finally, at the type level, the monadic approach
distinguishes ``selected'' values  and ``ordinary'' ones; the ``selected''
values are reminiscent of the ``uncertain'' values of type
$\mathrm{Uncertain}\!<\!\texttt{T}\!>$~\cite{uncertaint}, and the distinction may be useful as in that setting.

\section*{Acknowledgements}

We are grateful to Craig Boutilier, Eugene Brevdo, Daniel Golovin, Michael Isard, Eugene Kirpichov, Ohad Kammar, Matt Johnson, Dougal Maclaurin, Martin Mladenov, Adam Paszke, Sam Staton, Dimitrios Vytiniotis, and Jay Yagnik for discussions. 

\bibliography{SmartChoices.bib}{}

\newcommand{\etalchar}[1]{$^{#1}$}
\begin{thebibliography}{VFLF{\etalchar{+}}17}

\bibitem[AJM00]{ajm:pcf}
Samson Abramsky, Radha Jagadeesan, and Pasquale Malacaria.
\newblock Full abstraction for {PCF}.
\newblock {\em Inf. Comput.}, 163(2):409--470, 2000.
\newblock \href {https://doi.org/10.1006/inco.2000.2930}
  {\path{doi:10.1006/inco.2000.2930}}.

\bibitem[AP20]{abadiplotkin2020}
Mart{\'{\i}}n Abadi and Gordon~D. Plotkin.
\newblock A simple differentiable programming language.
\newblock {\em Proc. {ACM} Program. Lang.}, 4({POPL}):38:1--38:28, 2020.
\newblock \href {https://doi.org/10.1145/3371106} {\path{doi:10.1145/3371106}}.

\bibitem[AP21]{AP21}
Mart{\'i}n Abadi and Gordon Plotkin.
\newblock Smart choices and the selection monad.
\newblock In {\em 36th Annual {ACM/IEEE} Symposium on Logic in Computer
  Science, {LICS} 2019}. {IEEE}, 2021.

\bibitem[BCLG20]{BartheCLG20}
Gilles Barthe, Rapha{\"{e}}lle Crubill{\'{e}}, Ugo~Dal Lago, and Francesco
  Gavazzo.
\newblock On the versatility of open logical relations - continuity, automatic
  differentiation, and a containment theorem.
\newblock In Peter M{\"{u}}ller, editor, {\em Programming Languages and Systems
  - 29th European Symposium on Programming, {ESOP} 2020}, volume 12075 of {\em
  Lecture Notes in Computer Science}, pages 56--83. Springer, 2020.
\newblock \href {https://doi.org/10.1007/978-3-030-44914-8\_3}
  {\path{doi:10.1007/978-3-030-44914-8\_3}}.

\bibitem[Bel57]{Bel57}
Richard Bellman.
\newblock {\em Dynamic Programming}.
\newblock Princeton University Press, Princeton, 1957.

\bibitem[BHQK20]{rlax2020github}
David Budden, Matteo Hessel, John Quan, and Steven Kapturowski.
\newblock {RL}ax: {R}einforcement {L}earning in {JAX}, 2020.
\newblock URL: \url{http://github.com/deepmind/rlax}.

\bibitem[BHZ18]{bolt18}
Joe Bolt, Jules Hedges, and Philipp Zahn.
\newblock Sequential games and nondeterministic selection functions.
\newblock {\em CoRR}, abs/1811.06810, 2018.
\newblock URL: \url{http://arxiv.org/abs/1811.06810}, \href
  {http://arxiv.org/abs/1811.06810} {\path{arXiv:1811.06810}}.

\bibitem[BMM14]{uncertaint}
James Bornholt, Todd Mytkowicz, and Kathryn~S. McKinley.
\newblock Uncertain$<\!\mathrm{T}\!>$: a first-order type for uncertain data.
\newblock In Rajeev Balasubramonian, Al~Davis, and Sarita~V. Adve, editors,
  {\em Architectural Support for Programming Languages and Operating Systems,
  {ASPLOS} '14}, pages 51--66. {ACM}, 2014.
\newblock \href {https://doi.org/10.1145/2541940.2541958}
  {\path{doi:10.1145/2541940.2541958}}.

\bibitem[BMP20]{pagani2020}
Alo{\"{\i}}s Brunel, Damiano Mazza, and Michele Pagani.
\newblock Backpropagation in the simply typed lambda-calculus with linear
  negation.
\newblock {\em Proc. {ACM} Program. Lang.}, 4({POPL}):64:1--64:27, 2020.
\newblock \href {https://doi.org/10.1145/3371132} {\path{doi:10.1145/3371132}}.

\bibitem[BRST00]{Boutilier00a}
C.~Boutilier, R.~Reiter, M.~Soutchanski, and S.~Thrun.
\newblock Decision-theoretic, high-level robot programming in the situation
  calculus.
\newblock In {\em Proceedings of the AAAI National Conference on Artificial
  Intelligence}. AAAI, 2000.

\bibitem[Byc18]{spiral}
Vladimir Bychkovsky.
\newblock Spiral: Self-tuning services via real-time machine learning, 2018.
\newblock Blog post
  \href{https://engineering.fb.com/data-infrastructure/spiral-self-tuning-services-via-real-time-machine-learning/}{here}.

\bibitem[CCD{\etalchar{+}}18]{smartchoices}
Victor Carbune, Thierry Coppey, Alexander~N. Daryin, Thomas Deselaers, Nikhil
  Sarda, and Jay Yagnik.
\newblock Smartchoices: hybridizing programming and machine learning.
\newblock {\em CoRR}, abs/1810.00619, 2018.
\newblock URL: \url{http://arxiv.org/abs/1810.00619}, \href
  {http://arxiv.org/abs/1810.00619} {\path{arXiv:1810.00619}}.

\bibitem[CGM19]{cruttwell2019}
Geoff Cruttwell, Jonathan Gallagher, and Ben MacAdam.
\newblock Towards formalizing and extending differential programming using
  tangent categories.
\newblock {\em Proc. ACT}, 2019.

\bibitem[CHR{\etalchar{+}}16]{Chang16}
Kai{-}Wei Chang, He~He, St{\'{e}}phane Ross, Hal~Daum{\'{e}} III, and John
  Langford.
\newblock A credit assignment compiler for joint prediction.
\newblock In Daniel~D. Lee, Masashi Sugiyama, Ulrike von Luxburg, Isabelle
  Guyon, and Roman Garnett, editors, {\em Advances in Neural Information
  Processing Systems 29: Annual Conference on Neural Information Processing
  Systems 2016}, pages 1705--1713, 2016.
\newblock URL:
  \url{http://papers.nips.cc/paper/6256-a-credit-assignment-compiler-for-joint-prediction}.

\bibitem[DPS18]{DPS18}
Fredrik Dahlqvist, Louis Parlant, and Alexandra Silva.
\newblock Layer by layer--combining monads.
\newblock In {\em International Colloquium on Theoretical Aspects of
  Computing}, pages 153--172. Springer, 2018.

\bibitem[DS21]{DS21}
Swaraj Dash and Sam Staton.
\newblock A monad for probabilistic point processes.
\newblock {\em arXiv preprint arXiv:2101.10479}, 2021.

\bibitem[Dub06]{dub06}
Eduardo~J. Dubuc.
\newblock {\em Kan extensions in enriched category theory}, volume 145.
\newblock Springer, 2006.

\bibitem[EO10]{escardo10}
Mart{\'{\i}}n~H{\"{o}}tzel Escard{\'{o}} and Paulo Oliva.
\newblock Selection functions, bar recursion and backward induction.
\newblock {\em Math. Struct. Comput. Sci.}, 20(2):127--168, 2010.
\newblock \href {https://doi.org/10.1017/S0960129509990351}
  {\path{doi:10.1017/S0960129509990351}}.

\bibitem[EO11]{escardo2011sequential}
Martin Escard{\'o} and Paulo Oliva.
\newblock Sequential games and optimal strategies.
\newblock {\em Proceedings of the Royal Society A: Mathematical, Physical and
  Engineering Sciences}, 467(2130):1519--1545, 2011.

\bibitem[EO12]{escardo12}
Mart{\'{\i}}n~H{\"{o}}tzel Escard{\'{o}} and Paulo Oliva.
\newblock The {P}eirce translation.
\newblock {\em Ann. Pure Appl. Log.}, 163(6):681--692, 2012.
\newblock \href {https://doi.org/10.1016/j.apal.2011.11.002}
  {\path{doi:10.1016/j.apal.2011.11.002}}.

\bibitem[EO17]{escardo17}
Mart{\'{\i}}n Escard{\'{o}} and Paulo Oliva.
\newblock The {H}erbrand functional interpretation of the double negation
  shift.
\newblock {\em J. Symb. Log.}, 82(2):590--607, 2017.
\newblock \href {https://doi.org/10.1017/jsl.2017.8}
  {\path{doi:10.1017/jsl.2017.8}}.

\bibitem[EOP11]{escardo11}
Mart{\'{\i}}n~H{\"{o}}tzel Escard{\'{o}}, Paulo Oliva, and Thomas Powell.
\newblock System {T} and the product of selection functions.
\newblock In Marc Bezem, editor, {\em Computer Science Logic, 25th
  International Workshop / 20th Annual Conference of the EACSL, {CSL} 2011},
  volume~12 of {\em LIPIcs}, pages 233--247. Schloss Dagstuhl - Leibniz-Zentrum
  f{\"{u}}r Informatik, 2011.
\newblock \href {https://doi.org/10.4230/LIPIcs.CSL.2011.233}
  {\path{doi:10.4230/LIPIcs.CSL.2011.233}}.

\bibitem[Esc15]{escardo15}
Mart{\'{\i}}n Escard{\'{o}}.
\newblock Constructive decidability of classical continuity.
\newblock {\em Math. Struct. Comput. Sci.}, 25(7):1578--1589, 2015.
\newblock \href {https://doi.org/10.1017/S096012951300042X}
  {\path{doi:10.1017/S096012951300042X}}.

\bibitem[Fel08]{Fel08}
Willliam Feller.
\newblock {\em An introduction to probability theory and its applications, vol
  2}.
\newblock John Wiley \& Sons, 2008.

\bibitem[FF87]{FF87}
Matthias Felleisen and Daniel~P. Friedman.
\newblock Control operators, the secd-machine, and the {\(\lambda\)}-calculus.
\newblock In Martin Wirsing, editor, {\em Formal Description of Programming
  Concepts - {III:} Proceedings of the {IFIP} {TC} 2/WG 2.2 Working Conference
  on Formal Description of Programming Concepts - III}, pages 193--222.
  North-Holland, 1987.

\bibitem[Fil94]{filinski:representing}
Andrzej Filinski.
\newblock Representing monads.
\newblock In {\em Proceedings of the 21st ACM SIGPLAN-SIGACT Symposium on
  Principles of Programming Languages}, POPL ’94, page 446–457. ACM, 1994.
\newblock \href {https://doi.org/10.1145/174675.178047}
  {\path{doi:10.1145/174675.178047}}.

\bibitem[FST19]{FongST19}
Brendan Fong, David~I. Spivak, and R{\'{e}}my Tuy{\'{e}}ras.
\newblock Backprop as functor: {A} compositional perspective on supervised
  learning.
\newblock In {\em 34th Annual {ACM/IEEE} Symposium on Logic in Computer
  Science, {LICS} 2019}, pages 1--13. {IEEE}, 2019.
\newblock \href {https://doi.org/10.1109/LICS.2019.8785665}
  {\path{doi:10.1109/LICS.2019.8785665}}.

\bibitem[GMR{\etalchar{+}}12]{Goodman:church}
Noah~D. Goodman, Vikash~K. Mansinghka, Daniel~M. Roy, Keith Bonawitz, and
  Joshua~B. Tenenbaum.
\newblock Church: a language for generative models.
\newblock {\em CoRR}, abs/1206.3255, 2012.
\newblock URL: \url{http://arxiv.org/abs/1206.3255}, \href
  {http://arxiv.org/abs/1206.3255} {\path{arXiv:1206.3255}}.

\bibitem[Hed15]{hedges15}
Jules Hedges.
\newblock The selection monad as a {CPS} transformation.
\newblock {\em CoRR}, abs/1503.06061, 2015.
\newblock URL: \url{http://arxiv.org/abs/1503.06061}, \href
  {http://arxiv.org/abs/1503.06061} {\path{arXiv:1503.06061}}.

\bibitem[HLPP07a]{HL07}
Martin Hyland, Paul~Blain Levy, Gordon Plotkin, and John Power.
\newblock Combining algebraic effects with continuations.
\newblock {\em Theoretical Computer Science}, 375(1-3):20--40, 2007.

\bibitem[HLPP07b]{HLPP07}
Martin Hyland, Paul~Blain Levy, Gordon~D. Plotkin, and John Power.
\newblock Combining algebraic effects with continuations.
\newblock {\em Theor. Comput. Sci.}, 375(1-3):20--40, 2007.
\newblock \href {https://doi.org/10.1016/j.tcs.2006.12.026}
  {\path{doi:10.1016/j.tcs.2006.12.026}}.

\bibitem[HO00]{ho:pcf}
J.~M.~E. Hyland and C.{-}H.~Luke Ong.
\newblock On full abstraction for {PCF:} {I}, {II}, and {III}.
\newblock {\em Inf. Comput.}, 163(2):285--408, 2000.
\newblock \href {https://doi.org/10.1006/inco.2000.2917}
  {\path{doi:10.1006/inco.2000.2917}}.

\bibitem[HPP06]{HPP06}
Martin Hyland, Gordon~D. Plotkin, and John Power.
\newblock Combining effects: Sum and tensor.
\newblock {\em Theor. Comput. Sci.}, 357(1-3):70--99, 2006.
\newblock \href {https://doi.org/10.1016/j.tcs.2006.03.013}
  {\path{doi:10.1016/j.tcs.2006.03.013}}.

\bibitem[HSV20]{HuotSV20}
Mathieu Huot, Sam Staton, and Matthijs V{\'{a}}k{\'{a}}r.
\newblock Correctness of automatic differentiation via diffeologies and
  categorical gluing.
\newblock In Jean Goubault{-}Larrecq and Barbara K{\"{o}}nig, editors, {\em
  Foundations of Software Science and Computation Structures - 23rd
  International Conference, {FOSSACS} 2020}, volume 12077 of {\em Lecture Notes
  in Computer Science}, pages 319--338. Springer, 2020.
\newblock \href {https://doi.org/10.1007/978-3-030-45231-5\_17}
  {\path{doi:10.1007/978-3-030-45231-5\_17}}.

\bibitem[Jac21]{Jac21}
Bart Jacobs.
\newblock From multisets over distributions to distributions over multisets.
\newblock In {\em 2021 36th Annual ACM/IEEE Symposium on Logic in Computer
  Science (LICS)}, pages 1--13. IEEE, 2021.

\bibitem[Kar17]{Karpathy}
Andrej Karpathy.
\newblock Software 2.0, 2017.
\newblock Blog post
  \href{https://medium.com/@karpathy/software-2-0-a64152b37c35}{here}.

\bibitem[Kel80]{K80}
Max Kelly.
\newblock A unified treatment of transfinite constructions for free algebras,
  free monoids, colimits, associated sheaves, and so on.
\newblock {\em Bulletin of the Australian Mathematical Society}, 22(1):1--83,
  1980.

\bibitem[Koc72]{kock1972strong}
Anders Kock.
\newblock Strong functors and monoidal monads.
\newblock {\em Archiv der Mathematik}, 23(1):113--120, 1972.

\bibitem[KP93]{K93}
Max Kelly and John Power.
\newblock Adjunctions whose counits are coequalizers, and presentations of
  finitary enriched monads.
\newblock {\em Journal of pure and applied algebra}, 89(1-2):163--179, 1993.

\bibitem[KP17]{KeimelP16}
Klaus Keimel and Gordon~D. Plotkin.
\newblock Mixed powerdomains for probability and nondeterminism.
\newblock {\em Log. Methods Comput. Sci.}, 13(1), 2017.
\newblock \href {https://doi.org/10.23638/LMCS-13(1:2)2017}
  {\path{doi:10.23638/LMCS-13(1:2)2017}}.

\bibitem[LPT03]{LevyPT03}
Paul~Blain Levy, John Power, and Hayo Thielecke.
\newblock Modelling environments in call-by-value programming languages.
\newblock {\em Inf. Comput.}, 185(2):182--210, 2003.
\newblock \href {https://doi.org/10.1016/S0890-5401(03)00088-9}
  {\path{doi:10.1016/S0890-5401(03)00088-9}}.

\bibitem[LS18]{LS18}
Aliaume Lopez and Alex Simpson.
\newblock Basic operational preorders for algebraic effects in general, and for
  combined probability and nondeterminism in particular.
\newblock In Dan~R. Ghica and Achim Jung, editors, {\em 27th {EACSL} Annual
  Conference on Computer Science Logic, {CSL} 2018}, volume 119 of {\em
  LIPIcs}, pages 29:1--29:17. Schloss Dagstuhl - Leibniz-Zentrum f{\"{u}}r
  Informatik, 2018.
\newblock \href {https://doi.org/10.4230/LIPIcs.CSL.2018.29}
  {\path{doi:10.4230/LIPIcs.CSL.2018.29}}.

\bibitem[McC63]{jmc:mtc}
John McCarthy.
\newblock A basis for a mathematical theory of computation.
\newblock In P.~Braffort and D.~Hirschberg, editors, {\em Computer Programming
  and Formal Systems}, volume~35 of {\em Studies in Logic and the Foundations
  of Mathematics}, pages 33 -- 70. Elsevier, 1963.
\newblock \href {https://doi.org/10.1016/S0049-237X(08)72018-4}
  {\path{doi:10.1016/S0049-237X(08)72018-4}}.

\bibitem[MGH{\etalchar{+}}98]{pddl}
D.~Mc{D}ermott, M.~Ghallab, A.~Howe, C.~Knoblock, A.~Ram, M.~Veloso, D.~Weld,
  and D.~Wilkins.
\newblock {PDDL} - the planning domain definition language.
\newblock Technical Report TR-98-003, Yale Center for Computational Vision and
  Control, 1998.

\bibitem[Mog89]{Moggi89}
Eugenio Moggi.
\newblock Computational lambda-calculus and monads.
\newblock In {\em Proceedings of the Fourth Annual Symposium on Logic in
  Computer Science {(LICS} '89)}, pages 14--23. {IEEE} Computer Society, 1989.
\newblock \href {https://doi.org/10.1109/LICS.1989.39155}
  {\path{doi:10.1109/LICS.1989.39155}}.

\bibitem[PP01]{PP01}
Gordon Plotkin and John Power.
\newblock Adequacy for algebraic effects.
\newblock In Furio Honsell and Marino Miculan, editors, {\em Foundations of
  Software Science and Computation Structures}, pages 1--24. Springer Berlin
  Heidelberg, 2001.

\bibitem[PP03]{PP03}
Gordon~D. Plotkin and John Power.
\newblock Algebraic operations and generic effects.
\newblock {\em Applied Categorical Structures}, 11(1):69--94, 2003.
\newblock \href {https://doi.org/10.1023/A:1023064908962}
  {\path{doi:10.1023/A:1023064908962}}.

\bibitem[PR95]{pumplun1995convexity}
Dieter Pumpl{\"u}n and Helmut R{\"o}hrl.
\newblock Convexity theories {IV}. {K}lein-{H}ilbert parts in convex modules.
\newblock {\em Applied Categorical Structures}, 3(2):173--200, 1995.

\bibitem[S{\etalchar{+}}10]{rddl}
Scott Sanner et~al.
\newblock Relational dynamic influence diagram language (rddl): Language
  description.
\newblock {\em Unpublished ms. Australian National University}, 32:27, 2010.

\bibitem[Sto49]{stone1949postulates}
Marshall~Harvey Stone.
\newblock Postulates for the barycentric calculus.
\newblock {\em Annali di Matematica Pura ed Applicata}, 29(1):25--30, 1949.

\bibitem[SW15]{SW15}
Ana Sokolova and Harald Woracek.
\newblock Congruences of convex algebras.
\newblock {\em Journal of Pure and Applied Algebra}, 219(8):3110--3148, 2015.

\bibitem[VFLF{\etalchar{+}}17]{Dyna-mapl}
Tim Vieira, Matthew Francis-Landau, Nathaniel~Wesley Filardo, Farzad Khorasani,
  and Jason Eisner.
\newblock Dyna: Toward a self-optimizing declarative language for machine
  learning applications.
\newblock In {\em Proceedings of the 1st ACM SIGPLAN International Workshop on
  Machine Learning and Programming Languages}, MAPL 2017, page 8–17. ACM,
  2017.
\newblock \href {https://doi.org/10.1145/3088525.3088562}
  {\path{doi:10.1145/3088525.3088562}}.

\bibitem[VW06]{VW06}
Daniele Varacca and Glynn Winskel.
\newblock Distributing probability over non-determinism.
\newblock {\em Mathematical Structures in Computer Science}, 16(1):87--113,
  2006.

\bibitem[XZH{\etalchar{+}}20]{XiaZHHMPZ20}
Li{-}yao Xia, Yannick Zakowski, Paul He, Chung{-}Kil Hur, Gregory Malecha,
  Benjamin~C. Pierce, and Steve Zdancewic.
\newblock Interaction trees: representing recursive and impure programs in coq.
\newblock {\em Proc. {ACM} Program. Lang.}, 4({POPL}):51:1--51:32, 2020.
\newblock \href {https://doi.org/10.1145/3371119} {\path{doi:10.1145/3371119}}.

\end{thebibliography}
 \bibliographystyle{alphaurl}

\end{document}